\documentclass[11pt]{article}

\usepackage[utf8]{inputenc} % set input encoding (not needed with XeLaTeX)

\usepackage{geometry} % to change the page dimensions
\usepackage{graphicx} % support the \includegraphics command and options
\graphicspath{ {./figures/} }

\usepackage{booktabs} % for much better looking tables
\usepackage{array} % for better arrays (eg matrices) in maths
\usepackage{paralist} % very flexible & customisable lists (eg. enumerate/itemize, etc.)
\usepackage{verbatim} % adds environment for commenting out blocks of text & for better verbatim
\usepackage{subfig} % make it possible to include more than one captioned figure/table in a single float
\usepackage{fancyhdr} % This should be set AFTER setting up the page geometry
\usepackage{sectsty}
\allsectionsfont{\sffamily\mdseries\upshape} % (See the fntguide.pdf for font help)
\usepackage[nottoc,notlof,notlot]{tocbibind} % Put the bibliography in the ToC
\usepackage[titles,subfigure]{tocloft} % Alter the style of the Table of Contents

\usepackage{physics}
\usepackage{siunitx}

\usepackage{xcolor}
\usepackage{multirow}
\usepackage{amssymb}

\usepackage{subfig}
\usepackage{tikz}
\usepackage{tikzit}
\tikzstyle{boundary vertex}=[inner sep=0mm, minimum size=1mm, shape=circle, draw=black, fill=black]
\tikzstyle{grey_dot}=[fill={rgb,255: red,191; green,191; blue,191}, draw={rgb,255: red,191; green,191; blue,191}, shape=circle, minimum size=1mm, inner sep=0mm]
\tikzstyle{blue_dot}=[fill={rgb,255: red,202; green,251; blue,255}, draw=black, shape=circle, minimum size=1mm, inner sep=0mm]
\tikzstyle{white_dot}=[fill=white, draw=black, shape=circle, minimum size=1.5mm, inner sep=0mm]

\tikzstyle{arrow}=[->]
\tikzstyle{red_arrow}=[->, draw=red]
\tikzstyle{cyan_arrow}=[->, draw=cyan]
\tikzstyle{red_dash}=[-, dashed, draw=red]
\tikzstyle{grey dash}=[-, fill=none, draw={rgb,255: red,191; green,191; blue,191}, dashed]
\tikzstyle{dashed arrow}=[->, dashed]
\tikzstyle{blue_edge}=[-, draw={rgb,255: red,46; green,126; blue,255}]
\tikzstyle{red_edge}=[-, draw=red]

\tikzstyle{gate}=[shape=rectangle, text height=1.5ex, text depth=0.25ex, yshift=0.5mm, fill=white, draw=black, minimum height=5mm, yshift=-0.5mm, minimum width=5mm, font={\small}, tikzit category=circuit]
\tikzstyle{big gate}=[shape=rectangle, text height=1.5ex, text depth=0.25ex, yshift=0.5mm, fill=white, draw=black, minimum height=18mm, yshift=-0.5mm, minimum width=5mm, font={\small}, tikzit category=circuit]
\tikzstyle{Z dot}=[inner sep=0mm, minimum size=2mm, shape=circle, draw=black, fill={rgb,255: red,221; green,255; blue,221}, tikzit category=zx]
\tikzstyle{Z phase dot}=[minimum size=5mm, font={\footnotesize\boldmath}, shape=rectangle, rounded corners=2mm, inner sep=0.2mm, outer sep=-2mm, scale=0.8, tikzit shape=circle, draw=black, fill={rgb,255: red,221; green,255; blue,221}, tikzit draw=blue, tikzit category=zx]
\tikzstyle{X dot}=[Z dot, shape=circle, draw=black, fill={rgb,255: red,255; green,136; blue,136}, tikzit category=zx]
\tikzstyle{X phase dot}=[Z phase dot, tikzit shape=circle, tikzit draw=blue, fill={rgb,255: red,255; green,136; blue,136}, font={\footnotesize\boldmath}, tikzit category=zx]
\tikzstyle{hadamard}=[fill=yellow, draw=black, shape=rectangle, inner sep=0.6mm, minimum height=1.5mm, minimum width=1.5mm, tikzit category=zx]
\tikzstyle{paulibox}=[fill={rgb,255: red,221; green,221; blue,255}, draw=black, shape=rectangle, inner sep=0.6mm, minimum height=5mm, minimum width=5mm, font={\footnotesize}, text height=1.5ex, text depth=0.25ex, tikzit category=zx]
\tikzstyle{vertex}=[inner sep=0mm, minimum size=1mm, shape=circle, draw=black, fill=black, tikzit category=misc]
\tikzstyle{vertex set}=[inner sep=0mm, minimum size=1mm, shape=circle, draw=black, fill=white, font={\footnotesize\boldmath}, tikzit category=misc]
\tikzstyle{small black dot}=[fill=black, draw=black, shape=circle, inner sep=0pt, minimum width=1.2mm, tikzit category=circuit]
\tikzstyle{cnot ctrl}=[fill=black, draw=black, shape=circle, inner sep=0pt, minimum width=1.2mm, tikzit category=circuit]
\tikzstyle{cnot targ}=[fill=white, draw=white, shape=circle, tikzit category=circuit, label={center:$\oplus$}, inner sep=0pt, minimum width=2.1mm, tikzit fill={rgb,255: red,102; green,204; blue,255}, tikzit draw=black]
\tikzstyle{ket}=[fill=white, draw=black, shape=regular polygon, regular polygon sides=3, regular polygon rotate=-30, scale=0.7, inner sep=1pt, tikzit category=circuit, tikzit shape=rectangle, tikzit fill=green]
\tikzstyle{bra}=[fill=white, draw=black, shape=regular polygon, regular polygon sides=3, regular polygon rotate=30, scale=0.7, inner sep=1pt, tikzit category=circuit, tikzit shape=rectangle, tikzit fill=red]
\tikzstyle{scalar}=[shape=rectangle, text height=1.5ex, text depth=0.25ex, yshift=0.5mm, fill=white, draw=black, minimum height=5mm, yshift=-0.5mm, minimum width=5mm, font={\small}]
\tikzstyle{clabel}=[fill=white, draw=none, shape=rectangle, tikzit fill={rgb,255: red,56; green,255; blue,242}, font={\footnotesize}, inner sep=1pt, tikzit category=labels]
\tikzstyle{empty diagram}=[draw={gray!40!white}, dashed, shape=rectangle, minimum width=1cm, minimum height=1cm, tikzit category=misc]

\tikzstyle{hadamard edge}=[-, dashed, dash pattern=on 2pt off 0.5pt, thick, draw={rgb,255: red,68; green,136; blue,255}]
\tikzstyle{box edge}=[-, dashed, dash pattern=on 2pt off 0.5pt, thick, draw={rgb,255: red,203; green,192; blue,225}]
\tikzstyle{brace edge}=[-, tikzit draw=blue, decorate, decoration={brace,amplitude=1mm,raise=-1mm}]
\tikzstyle{diredge}=[->]
\tikzstyle{double edge}=[-, double, shorten <=-1mm, shorten >=-1mm, double distance=2pt]
\tikzstyle{gray edge}=[-, {gray!60!white}]
\tikzstyle{pointer edge}=[->, very thick, gray]
\tikzstyle{boldedge}=[-, line width=1.6pt, shorten <=-0.17mm, shorten >=-0.17mm]

\tikzstyle{rmat}=[draw, signal, fill=gray!30, signal to=east, signal from=west, inner sep=1.5pt, minimum height=9pt]
\tikzstyle{lmat}=[draw, signal, fill=gray!30, signal to=west, signal from=east, inner sep=1.5pt, minimum height=9pt]
\tikzstyle{ggreen}=[fill=green, draw=black, shape=circle, tikzit category=SZX, tikzit fill=green, tikzit draw=black, line width=1pt, inner sep=2pt]
\tikzstyle{gred}=[fill=red, draw=black, shape=circle, rounded corners=2mm,  tikzit category=SZX, inner sep=2pt, tikzit fill=red, line width=1pt]
\tikzstyle{divide}=[regular polygon, regular polygon sides=3, shape border rotate=90, draw=black,fill=gray!30, inner sep=1.6pt, tikzit category=scal, rounded corners=0.8mm]
\tikzstyle{gather}=[fill=gray!30, draw=black, tikzit category=scal, rounded corners=0.8mm, regular polygon, regular polygon sides=3, shape border rotate=-90, inner sep=1.6pt]
\tikzstyle{A}=[fill=white, shape=circle, tikzit category=scal, inner sep=1pt]

\tikzstyle{very thick}=[-, line width=1pt, tikzit draw=red]

\usepackage{tabularray}
\usepackage{subfig}
\usepackage{amsmath}

\usepackage{amsthm}
\newtheorem{theorem}{Theorem}[section] 
\newtheorem{lemma}[theorem]{Lemma} % shares numbering with theorems
\newtheorem{corollary}[theorem]{Corollary} % shares numbering with theorems
\newtheorem{definition}[theorem]{Definition} % shares numbering with theorems
\newtheorem{deflem}[theorem]{Definition-lemma} % shares numbering with theorems

\DeclareMathOperator{\Span}{Span}

\usepackage{makecell}
\usepackage{float}

\usepackage{hyperref}

\usepackage[numbers,sort&compress]{natbib}

\usepackage{authblk}
\usepackage{blindtext}
\usepackage{mathtools}
\usepackage{quantikz}

\title{Local distillation from Reed Muller codes unfolding}
\date{} % Activate to display a given date or no date (if empty),
\author[1]{Vivien Londe}
\affil[1]{Alice \& Bob, 49 Bd du Général Martial Valin, 75015 Paris, France}
\affil[ ]{vivien.londe@alice-bob.com}

\begin{document}
\maketitle

\tableofcontents

\newpage

\begin{abstract}
We generalize the unfolding of a Reed Muller distillation factory of \cite{ruiz2025unfolded} by exhibiting the algebraic structure that the unfolding is based on. We describe a 2D local layout for the Z stabilizers of a distance 4 Reed Muller distillation factory and a 3D local layout for the Z stabilizer of a distance 4 and a distance 7 Reed Muller distillation factory. Given input T states with infidelities $p=10^{-3}$, the 2D local distillation factory with distance 4 outputs a CCZ state with infidelity $p=8.256 \times 10^{-9}$ and the 3D local distillation factory with distance 7 outputs a T state with infidelity $p=1.1811 \times 10^{-17}$.
\end{abstract}

\section{Introduction}

Universal fault-tolerant quantum computation is achievable with Clifford gates and noisy magic states \cite{knill2004fault_schemes, knill2004fault_threshold, bravyi2005universal} as long as the magic states fidelity is above a known threshold. From a magic state, a non-Clifford gate is teleported (by means of Clifford gates only) onto a logical qubit. Since the fidelity of the non-Clifford gate is determined by the fidelity of the magic state, the magic state fidelity is improved before it is injected, for example through distillation \cite{bravyi2005universal, bravyi2012magic}. Intensive work went into refining distillation protocols \cite{meier2012magic, jones2013multilevel, jones2013low, jones2013composite, eastin2013distilling, paetznick2013universal, haah2017magic, campbell2017unifying, haah2018codes, litinski2019magic, gidney2019efficient, guillaud2021error, gidney2024magic}. It often relies on a triorthogonal code, such as the popular [[15,1,3]] quantum Reed Muller code \cite{knill1998resilient, koutsioumpas2022smallest}. In \cite{ruiz2025unfolded}, a basis of the Z stabilizer group of the [[15,1,3]] quantum Reed Muller code is unfolded in a local 2D layout. When biased qubits such as cat qubits \cite{mirrahimi2014dynamically, guillaud2019repetition, puri2020bias} are the quantum information carrier, it is not problematic that the $X$ stabilizer group is not local. Thus, the authors of \cite{ruiz2025unfolded} describe a very compact (in number of cat qubits times number of error correction rounds) distillation scheme for biased noise qubits.

This article builds upon their work and generalizes it to other quantum Reed Muller codes. We first observe that the $Z$ stabilizer group\footnote{In this article, we describe factories as $\ket{T}$ state factories, which is consistent with most of the litterature on distillation. The authors of \cite{ruiz2025unfolded} use the opposite convention for $X$ and $Z$ stabilizer groups, in order for them to be consistent with the litterature on cat qubits. Therefore they present the small unfolded code as a $\ket{X^{1/4}}$ state factory, such that its $X$ stabilizer group is local in 2D.} of the (small) unfolded code from \cite{ruiz2025unfolded} has a product structure that explains its local 2D layout. More precisely we unfold the 4-dimensional description of the [[15,1,3]] code (see e.g. \cite{barg2025geometric}) into a 2D layout for its Z stabilizer group. In \cite{ruiz2025unfolded}, this product structure was not apparent since the 2D layout had been found with a SAT solver. We then generalize the unfolding to larger quantum Reed Muller codes. In particular, we highlight what we call the big unfolded code which is an interpolation between two quantum Reed Muller codes that are described on a 6-dimensional cube in \cite{barg2025geometric}. The Z distance of the big unfolded code is 4, therefore the distillation protocol associated to it improves the infidelity of magic states from $p$ to $\theta(p^4)$. More precisely, the big unfolded code is a $64 \ket{T}$ to $\ket{CCZ}$ factory such that $64$ states $\ket{T}$, each with infidelity $p$, are consumed to produce a $\ket{CCZ}$ state with infidelity $8256 \, p^4$. If $p=10^{-3}$, then $8256 \, p^4 = 8.256 \times 10^{-9}$.

We also showcase another code on 64 physical qubits whose $Z$ distance is 4. At the cost of 42 additional qubits, we describe a 3D layout - which we call the rubik's cube layout - of this code such that a basis of the $Z$ stabilizer group of this code is made of local cubes in 3D. The rubik's code is a 64 $\ket{T}$ to 15 $\ket{CCZ}$ factory. It encodes 15 logical qubits and we describe which triples of logical qubits undergo a $CCZ$ gate when a transversal T gate is applied. The logical error probability is $10416 p^4$. If $p=10^{-3}$, then $10416 p^4 = 1.0416  \times 10^{-8}$.

Finally, we showcase a code on 127 physical qubits whose $Z$ distance is 7. This code is obtained from puncturing a quantum Reed Muller code defined on a 7-dimensional cube. At the cost of 152 additional physical qubits, we describe a 3D layout of this code such that a basis of the $Z$ stabilizer group of this code is made of local cubes in 3D. This code is a 127 $\ket{T}$ to $\ket{T}$ factory. For an error probability $p$ on each of the 127 $T$ gates, the logical infidelity is $11811 p^7$. If $p=10^{-3}$, then $11811 p^7 = 1.1811 \times 10^{-17}$.

During completion of this work, we were made aware of independent results described in \cite{tiurev2026parity}, that also build on the results of \cite{ruiz2025unfolded} by unfolding quantum Reed Muller codes that are larger than the punctured $QRM_4(1,1)$. More precisely, they unfold the punctured $QRM_m(1,1)$ codes. While our work improves on the minimum distance of distillation factories, their work allows to distill states that lay higher than $\ket{T}$ in the Clifford hierarchy. As a consequence, the distillation codes from their work and from this work are different.

\section{Definitions and technical lemmas}

In this article, we consider quantum (Reed Muller) error correction codes defined on $n=2^m$ physical qubits. We define Z stabilizers and X stabilizers by giving the subset of the $2^m$ vertices of an $m$-cube on which the stabilizer acts non-trivially. Given a subset $V$ of vertices of the $m$-cube:
$$Z(V) = \prod_{v \in \mathbb{F}_2^m} (Z_v)^{v \in V}$$
and 
$$X(V) = \prod_{v \in \mathbb{F}_2^m} (X_v)^{v \in V},$$
where $Z_v$ (respectively $X_v$) acts like $Z$ (respectively $X$) on the physical qubit at vertex $v$ and acts like $I$ on the $2^m-1$ other physical qubits.

We characterize subsets of the set of vertices of an $m$-cube geometrically (with subcubes) and algebraicaly (with polynomials).

\begin{definition}
Let $P \in \mathbb{F}_2[X_1, \cdots, X_m]$\footnote{In this article, a polynomial $P \in \mathbb{F}_2[X_1, \cdots, X_m]$ is considered only through its evaluation function over $\mathbb{F}_2$. Therefore $X_i^2=X_i$. Thus, we abuse notations and write $\mathbb{F}_2[X_1, \cdots, X_m]$ to actually denote the quotient of $\mathbb{F}_2[X_1, \cdots, X_m]$  by $(\prod_{i \in \{1, \dots, m\}} X_i(X_i+1))$.}.
The subset of vertices of the $m$-cube associated with $P$ is the subset $V_P$ where $P$ evaluates to $1$:
$$V_P = \{v \in \mathbb{F}_2^m \, | \, P(v)=1 \}.$$
The stabilizer associated with $P$ acts non-trivially on $V_P$ and trivially (by the identity) on $\mathbb{F}_2^m \setminus V_P$.
\end{definition}

For instance, in the $3$-cube, a square (i.e. a $2$-subcube) corresponds to a degree $1$ polynomial. An edge (i.e. a $1$ subcube) corresponds to a degree $2$ polynomial. In the sequel, we abuse notations and refer to $P$ and $V_P$ interchangeably. We also refer to a subcube and to the set of its vertices interchangeably. Finally, we refer to a stabilizer by its corresponding subcube or its corresponding polynomial. We mean that the stabilizer acts nontrivially on the set of vertices corresponding to this subcube or to this polynomial and trivially on the other vertices of the $m$-cube.

\begin{definition}
Let $m \geq 1$. Let $J$ be a subset of $\{1, \dots, m\}$. We denote by $\overline{J}$ the complement of $J$ in $\{1, \dots, m\}$:
$$\overline{J} = \{i \in \{1, \dots, m\} \, | \, i \notin J\}.$$
\end{definition}

\begin{definition}[\cite{barg2025geometric}, Definition 2.1]
A subcube of type $J$, where $J$ is a subset of $\{1, \dots, m\}$, is a subcube that corresponds to a polynomial of degree $m - |J|$ in the variables $(X_j)_{j \in \overline{J}}$. $|J|$ is the dimension of the subcube.
\label{def:subcube}
\end{definition}

For instance in the $4$-cube, $X_2 X_3 (X_4+1)$ is an edge (i.e. a subcube of dimension $1$) of type $\{1\}$.

Note that the type of a subcube corresponds to (the complement of) the variables of the associated polynomial and that translating this subcube corresponds to adding $+1$ to some of these variables. For instance the $4$ edges of type $\{3\}$ of the $3$-cube correspond to the $4$ following polynomials: $X_1 X_2$, $(X_1+1) X_2$, $X_1 (X_2+1)$ and $(X_1+1) (X_2+1)$.

\begin{deflem}[\cite{barg2025geometric}, Definition 2.1]
A subcube is characterized by one of its vertices and its type. The subcube of type $J$ that contains $v$ is denoted $v + \langle J \rangle$. 
\end{deflem}

\begin{proof}
The subcube $v + \langle J \rangle$ corresponds to the polynomial
$$\prod_{j \in \overline{J}}(X_j + v_j + 1).$$
\end{proof}

\begin{definition}
Let $\mathcal{S}$ be a subspace of $\mathbb{F}_2[X_i, i \in \{1, \dots, m\}]$. The subspace of $\mathbb{F}_2^{\mathbb{F}_2^m}$ (i.e. the power set of $\mathbb{F}_2^m$ considered as an $\mathbb{F}_2$ vector space) associated to $\mathcal{S}$ is the subspace $\{V_P \, | \, P \in \mathcal{S}\}$, where $V_P = \{v \in \mathbb{F}_2^m \, | \, P(v)=1 \}.$
\end{definition}

\begin{definition}
The standard subcube of type $J$ is the unique subcube of type $J$ that contains the vertex $0 \dots 0$.
\end{definition}

\begin{definition}
Given $T \subset \mathcal{P}(\{1,\dots,m\})$ a set of types of the $m$-cube, we denote by
$$\text{Subcubes}(T)$$
the space generated by all subcubes whose type belongs to $T$.
\end{definition}

\begin{definition}[classical Reed Muller]
For integers $m, r$ such that $m \geq 1$ and $0 \leq r \leq m$, $RM_m(r)$ is the classical error correction code of length $n=2^m$ (i.e. on $2^m$ physical bits) associated to the space of polynomials of degree at most $r$ in $m$ variables: $(\mathbb{F}_2[X_i, i \in \{1, \dots, m\}])_{\leq r}$.
\end{definition}

Equivalently, $RM_m(r)$ is generated by all subcubes of dimension $m-r$ of the $m$-cube (see \cite{barg2025geometric}, section~5.1):
$$RM_m(r) = \text{Subcubes}(T)$$
where $T$ is the set of subsets of cardinal $m-r$ of $\{1,\dots,m\}$.

\begin{definition}[quantum Reed Muller]
For integers $m, q, r$ such that $m \geq 1$ and $0 \leq q \leq r \leq m$, $QRM_m(q, r)$ is the quantum error correction code with $n=2^m$ physical qubits whose X stabilizer group is $RM_m(q)$ and Z stabilizer group is $RM_m(m-r-1)$.
\end{definition}

Equivalently, the X stabilizer group is generated by all subcubes of dimension $m-q$ and the Z stabilizer group is generated by all subcubes of dimension $r+1$.

Note that the X logical group of $QRM_m(r)$ is $RM_m(r)$ and its Z logical group is $RM_m(m-q-1)$ (see \cite{barg2025geometric}, section~5.2). Equivalently, the X logical group is generated by all subcubes of dimension $m-r$ and the Z logical group is generated by all subcubes of dimension $q+1$.

\begin{lemma}
The subspace generated by all subcubes of type $J$ is the subspace associated to $\mathbb{F}_2[X_i, i \in \overline{J}]$.
\label{lem:6_cube_1_type}
\end{lemma}

For instance, in the $6$-cube, the subspace generated by all squares (i.e. subcubes of dimension $2$) of type $\{1,6\}$ is $\mathbb{F}_2[X_2, X_3, X_4, X_5]$.

\begin{proof}
The set of polynomials corresponding to subcubes of type $J$ is $$ \mathcal{F}_J = \{ \prod_{i \in \overline{J}} (X_i + a_i), \, a_i \in \mathbb{F}_2 \}.$$

$\mathcal{F}_J$ is a free family since only one of its elements evaluates to $1$ on $v \in \mathbb{F}_2^m$ such that $v_i = a_i +1$ (and such a $v$ exists for each element of $\mathcal{F}_J$). $\mathcal{F}_J$ has cardinal $2^{m - |J|}$ and is a family of elements of $\mathbb{F}_2[X_i, i \in \overline{J}]$, which has dimension $2^{m - |J|}$.
Therefore $$\Span(\mathcal{F}_J) = \mathbb{F}_2[X_k, k \in \overline{J}].$$
\end{proof}

\begin{deflem}
Let $k_1 \leq i_1 < i_2 \leq k_2 \in \{1, \dots, m\}$. For $1 \leq k < k_1$ and for $k_2 < k \leq m$, let $a_k \in \mathbb{F}_2$. Let $I = \{i_1, i_2\}$.
$$ V_{k_1, i_1, i_2, k_2, \, (a_k)} := \left( \prod_{1 \leq k < k_1} (X_k + a_k) \right) \,  \mathbb{F}_2[X_k, k \in \overline{I}, k_1 \leq k \leq k_2] \, \left( \prod_{k_2 < k \leq m} (X_k + a_k) \right)$$

is the space generated by the set of squares of type $\{i_1, i_2\}$ whose vertices $v$ all satisfy $v_k = a_k+1$ for $k$ in $\{1, \dots, k_1-1\} \cup \{k_2+1, \dots, m\}$. In other words, it is the space generated by a square of type $\{i_1, i_2\}$ whose vertices $v$ all satisfy $v_k = a_k+1$ for $k$ in $\{1, \dots, k_1-1\} \cup \{k_2+1, \dots, m\}$ and all its translations along coordinates in $\{k_1, \dots, k_2\}$:
$$ V_{k_1, i_1, i_2, k_2, \, (a_k)} = \langle \, (t | w | u) + \langle \{i_1, i_2\} \rangle , w \in \mathbb{F}_2^{k_2-k_1+1} \, \rangle,$$
where $t \in \mathbb{F}_2^{k_1-1}$ is the following bitstring\footnote{we use the symbol $|$ to denote concatenation of bitstrings. When bitstrings of length $1$ are concatenated, it amounts to describing a bitstring coordinate by coordinate.} of length $k_1 - 1$:
$$t = (a_1+1) | \dots | (a_{k_1-1}+1)$$
and $u \in \mathbb{F}_2^{m-k_2}$ is the following bitstring of length $m - k_2$:
$$u = (a_{k_2+1}+1) | \dots | (a_m+1).$$
\end{deflem}

\begin{proof}
We prove the two inclusions:
\begin{itemize}
\item $\subset$. \\
Let $P \in V_{k_1, i_1, i_2, k_2, \, (a_k)}$. There exists $Q \in \mathbb{F}_2[X_k, k \in \overline{I}, k_1 \leq k \leq k_2]$ such that
$$P = \left( \prod_{1 \leq k < k_1} (X_k + a_k) \right) \, Q \, \left( \prod_{k_2 < k \leq m} (X_k + a_k) \right).$$
$Q$ corresponds to an element of the space generated by squares of type $\{i_1, i_2\}$ in the cube of dimension $k_2 - k_1 +1$ with coordinate indices in $\{k_1, \dots, k_2\}$.

Let $v \in \mathbb{F}_2^{m}$. There exists $t \in \mathbb{F}_2^{k_1-1}$, $w \in \mathbb{F}_2^{k_2 - k_1 + 1}$ and $u \in \mathbb{F}_2^{m-k_2}$ such that $v = t | w | u$. $P(v)=1$ if and only if
$$t = (a_1+1) | \dots | (a_{k_1-1}+1),$$
$$Q(w)=1$$
and
$$u = (a_{k_2+1}+1) | \dots | (a_m+1).$$
Therefore $P$ belongs to the state generated by the square $(a_1+1) | \dots | (a_{k_1-1}+1) | w | (a_{k_2+1}+1) | \dots | (a_m+1) + \langle \{i_1, i_2\} \rangle$ and all its translations along coordinates in $\{k_1, \dots, k_2\}$.

\item $\supset$. \\
Let $w \in \mathbb{F}_2^{m-k_1-k_2}$. The four vertices of the square $(a_1+1) | \dots | (a_{k_1-1}+1) | w | (a_{k_2+1}+1) | \dots | (a_m+1) + \langle \{i_1, i_2\} \rangle$ satisfy $X_k=a_k$ for $k \in \{1, \dots, k_1-1\} \cup \{k_2+1, \dots, m\}$. Therefore the square $(a_1+1) | \dots | (a_{k_1-1}+1) | w | (a_{k_2+1}+1) | \dots | (a_m+1) + \langle \{i_1, i_2\} \rangle$ belongs to $V_{k_1, i_1, i_2, k_2, \, (a_k)}$.
\end{itemize}
\end{proof}

\begin{lemma}
Let $I, J \in \{1, \dots, m\}$ be two subsets of cardinal $d$ such that $|I \cap J| = d-1$. Let $i_0$ be such that $I = (I \cap J) \cup \{i_0\}$. For any $b_{i_0} \in \mathbb{F}_2$, the subspace generated by all subcubes of type $I$ or $J$ is the subspace associated to $\mathbb{F}_2[X_k, k \in \overline{I}] \oplus (X_{i_0} + b_{i_0}) \, \mathbb{F}_2[X_k, k \in \overline{I \cup J}]$. It has dimension $2^{m-d} + 2^{m-d-1}$. Written algebraically,
$$\mathbb{F}_2[X_k, k \in \overline{I}] + \mathbb{F}_2[X_k, k \in \overline{J}] = \mathbb{F}_2[X_k, k \in \overline{I}] \oplus (X_{i_0} + b_{i_0}) \, \mathbb{F}_2[X_k, k \in \overline{I \cup J}]$$
\label{lem:subspace_two_types}
\end{lemma}

For instance, in the $6$-cube, the subspace generated by all squares (i.e. subcubes of dimension $2$) of type $\{1,6\}$ or $\{1,5\}$  is $\mathbb{F}_2[X_2, X_3, X_4, X_5] \oplus \mathbb{F}_2[X_2, X_3, X_4] \, X_6$ and is also $\mathbb{F}_2[X_2, X_3, X_4, X_5] \oplus \mathbb{F}_2[X_2, X_3, X_4] \, (X_6 + 1)$.

\begin{corollary}
Let $I = \{i_1, i_2\}$, with $i_1, i_2 \in \{1, \dots, m\}$.

For every $k_1 \leq i_1$, every $k_2 > i_2$, every $(a_k) \in \mathbb{F}_2^{(k_1-1)+(m-k_2)}$ and every $a_{k_2} \in \mathbb{F}_2$, the space generated by the square $(a_1+1) | \dots | (a_{k_1-1}+1) | w | (a_{k_2+1}+1) | \dots | (a_m+1) + \langle \{i_1, i_2\} \rangle$ and all its translations along coordinates in $\{k_1, \dots, k_2\}$ is a subspace of the sum of the space generated by the square $(a_1+1) | \dots | (a_{k_1-1}+1) | w | (a_{k_2}+1) | \dots | (a_m+1) + \langle \{i_1, i_2\} \rangle$ and all its translations along coordinates in $\{k_1, \dots, k_2-1\}$ and the space generated by squares of type $\{i_1, k_2\}$.

Written algebraically,
$$\forall k_1 \leq i_1, \forall k_2 > i_2, \quad V_{k_1, i_1, i_2, k_2, \, (a_k)} \subset
V_{k_1, i_1, i_2, k_2-1, \, (a_k)} + \mathbb{F}_2[X_k, k \in \overline{\{i_1, k_2\}}].$$

Similarly,
$$\forall k_1 < i_1, \forall k_2 \geq i_2, \quad V_{k_1, i_1, i_2, k_2, \, (a_k)} \subset
V_{k_1+1, i_1, i_2, k_2, \, (a_k)} + \mathbb{F}_2[X_k, k \in \overline{\{k_1, i_2\}}].$$
\label{cor:translations_from_another_type}
\end{corollary}

Informally, Corollary~\ref{cor:translations_from_another_type} states that squares of type $\{i_1, k_2\}$ give translations along the coordinate $k_2$ to squares of type $\{i_1, i_2\}$. This is illustrated in Figure \ref{fig:translation}.

\begin{figure}[H]
    \centering
    \includegraphics[width=0.7\textwidth]{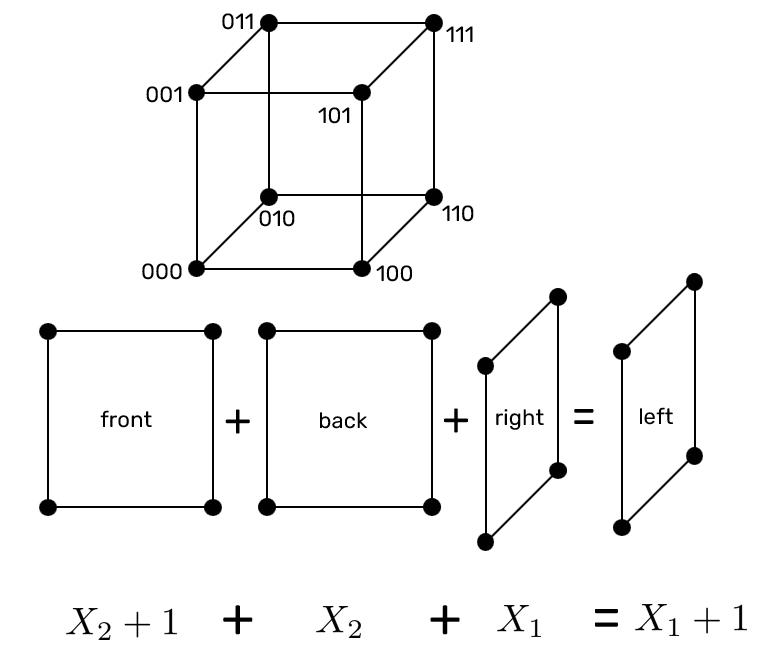}
    \caption{The squares of type $\{1,3\}$ ({\it front} and {\it back}) give translations along the $1^{st}$ coordinate to the square {\it right}, of type $\{2,3\}$. {\it front} corresponds to the polynomial $X_2+1$, {\it back} corresponds to the polynomial $X_2$, {\it right} corresponds to the polynomial $X_1$ and {\it left} corresponds to the polynomial $X_1+1$. The algebraic proof is elementary: $(X_2+1) + (X_2) + (X_1) = (X_1+1)$ in $\mathbb{F}_2[X_1, X_2, X_3]$.}
    \label{fig:translation}
\end{figure}

\begin{proof}[Proof of Lemma~\ref{lem:subspace_two_types}]
Let
$$\mathcal{F}_I = \{ \prod_{k \in \overline{I}} (X_k + a_k), \, (a_k)_{k \in \overline{I}} \in \mathbb{F}_2^{|\overline{I}|} \},$$
$$\mathcal{F}_J = \{ \prod_{k \in \overline{J}} (X_k + a_k), \, (a_k)_{k \in \overline{J}} \in \mathbb{F}_2^{|\overline{J}|} \}$$
and
$$\mathcal{F}_{I \cup J, i_0} = \{ (X_{i_0} + b_{i_0}) \prod_{k \in \overline{I \cup J}} (X_k + a_k), \, (a_k)_{k \in \overline{J}} \in \mathbb{F}_2^{|\overline{I \cup J}|} \}.$$

\begin{itemize}
\item We first show that $\mathcal{F}_J \subset \Span(\mathcal{F}_{I \cup J, i_0} \cup \mathcal{F}_I)$. Let $P \in \mathcal{F}_J$. Let $i_0$ be such that $I = (I \cap J) \cup \{i_0\}$. $P = \prod_{k \in \overline{J}} (X_k + a_k)$ where $a_k \in \mathbb{F}_2^{m - |J|}$. Developing the term $(X_{i_0} + a_{i_0}) = ((X_{i_0} + b_{i_0}) + (a_{i_0} + b_{i_0})$ gives $P = Q + (a_{i_0} + b_{i_0}) R$ with
$$Q = (X_{i_0} + b_{i_0}) R \in (X_{i_0} + b_{i_0})) \mathbb{F}_2[X_k, k \in \overline{I \cup J}] = \Span(\mathcal{F}_{I \cup J, i_0})$$
and
$$R = \prod_{k \in \overline{I \cup J}} (X_k + a_k) \in \mathbb{F}_2[X_k, k \in \overline{I \cup J}] \subset \Span(\mathcal{F}_{I}).$$

Since $\Span(\mathcal{F}_J) = \mathbb{F}_2[X_i, i \in \overline{J}]$  we have obtained
$$\mathbb{F}_2[X_i, i \in \overline{J}] \subset \Span(\mathcal{F}_{I \cup J, i_0} \cup \mathcal{F}_I).$$
Since $\Span(\mathcal{F}_I) = \mathbb{F}_2[X_i, i \in \overline{I}]$,
$$\mathbb{F}_2[X_i, i \in \overline{I}] + \mathbb{F}_2[X_i, i \in \overline{J}] \subset \Span(\mathcal{F}_{J, i_0} \cup \mathcal{F}_I) = \Span(\mathcal{F}_I) \oplus \Span(\mathcal{F}_{J, i_0}).$$
Therefore
$$\mathbb{F}_2[X_i, i \in \overline{I}] + \mathbb{F}_2[X_i, i \in \overline{J}] \subset \mathbb{F}_2[X_i, i \in \overline{I}] \oplus (X_{i_0} + b_{i_0}) \, \mathbb{F}_2[X_k, k \in \overline{I \cup J}].$$

\item We now show the dimension equality.

Since $\mathbb{F}_2[X_k, k \in \overline{I}] \cap \mathbb{F}_2[X_k, k \in \overline{J}] = \mathbb{F}_2[X_k, k \in \overline{I \cup J}]$,
$$\dim(\mathbb{F}_2[X_k, k \in \overline{I}] \, + \, \mathbb{F}_2[X_k, k \in \overline{J}]) = 2^{m-d} + 2^{m-d} - 2^{m-d-1} = 2^{m-d} + 2^{m-d-1}.$$
Finally, since $\mathcal{F}_I \cup \mathcal{F}_{I \cup J, i_0}$ is a basis of $\mathbb{F}_2[X_k, k \in \overline{I}] \oplus X_{j_0} \, \mathbb{F}_2[X_k, k \in \overline{I \cap J}]$ (which proves that the sum is indeed direct), this space also has dimension $2^{m-d} + 2^{m-d-1}$.
\end{itemize}
An inclusion and the equality of dimensions give the equality of $\Span(\mathcal{F}_I) + \Span(\mathcal{F}_J)$ and $\Span(\mathcal{F}_I) \oplus \Span(\mathcal{F}_{I \cup J, i_0})$, which is the result.
\end{proof}

\begin{proof}[Proof of Corollary~\ref{cor:translations_from_another_type}]
Lemma~\ref{lem:subspace_two_types} applied to $I = \{i_1, k_2\}$ and $J = \{i_1, i_2\}$ in the cube of dimension $k_2-k_1+1$ with coordinates in $\{k_1, \dots, k_2\}$ gives
\begin{align*}
& \mathbb{F}_2[X_k, k \in \overline{\{i_1, i_2\}}, k_1 \leq k \leq k_2] \\
\subset \quad
& \mathbb{F}_2[X_k, k \in \overline{\{i_1, i_2, k_2\}}, k_1 \leq k \leq k_2] (X_{k_2}+a_{k_2}) \\
\oplus \quad
& \mathbb{F}_2[X_k, k \in \overline{\{i_1, k_2\}}, k_1 \leq k \leq k_2].
\end{align*}
Multiplying by $\prod_{k < k_1} (X_k+a_k)$ and by $\prod_{k > k_2} (X_k+a_k)$ gives
\begin{align*}
& V_{k_1, i_1, i_2, k_2} \\
\subset \quad
& \left( \prod_{k < k_1} (X_k+a_k) \right) \mathbb{F}_2[X_k, k \in \overline{\{i_1, i_2, k_2\}}, k_1 \leq k \leq k_2] X_{k_2} \, \left( \prod_{k > k_2} (X_k+a_k) \right) \\
\oplus \quad
& \left( \prod_{k < k_1} (X_k+a_k) \right) \, \mathbb{F}_2[X_k, k \in \overline{\{i_1, k_2\}}, k_1 \leq k \leq k_2] \, \left( \prod_{k > k_2} (X_k+a_k) \right).
\end{align*}
Observe that
$$\left( \prod_{k < k_1} (X_k+a_k) \right) \mathbb{F}_2[X_k, k \in \overline{\{i_1, i_2, k_2\}}] X_{k_2} \, \left( \prod_{k > k_2} (X_k+a_k) \right) = V_{k_1, i_1, i_2, k_2-1, \, (a_k)}$$
and that
$$\left( \prod_{k < k_1} (X_k+a_k) \right) \, \mathbb{F}_2[X_k, k \in \overline{\{i_1, k_2\}}, k_1 \leq k \leq k_2] \, \left( \prod_{k > k_2} (X_k+a_k) \right) \subset \mathbb{F}_2[X_k, k \in \overline{\{i_1, k_2\}}].$$
Therefore
$$V_{k_1, i_1, i_2, k_2, \, (a_k)} \subset V_{k_1, i_1, i_2, k_2-1, \, (a_k)} + \mathbb{F}_2[X_k, k \in \overline{\{i_1, k_2\}}].$$
The proof of the second statement of the corollary is identical and we omit it.
\end{proof}

Definition \ref{def:diagonals} is motivated by Table~\ref{tab:subspaces_by_type_QRM_4_1_1} (for the $4$-cube) and by Table~\ref{tab:subspaces_by_type_QRM_6_1_1} (for the $6$-cube). Indeed, $V_{f_{max}}$ is the direct sum of spaces generated by the $f_{max}$ first top-left to bottom-right diagonals (starting at the top-right corner) of Table~\ref{tab:subspaces_by_type_QRM_4_1_1}~or~\ref{tab:subspaces_by_type_QRM_6_1_1}.

\begin{definition}
Let $f_{max} \in \{1, \dots, m\}$. For every $I =\{i_1, i_2\} \subset \{1, \dots, m\}$ such that $i_2 \geq i_1 + m - f_{max}$, for every $k \in \{1, \dots, k_1-1\} \cup \{k_2+1, \dots, m\}$, let $a_k^{(i_1,i_2)} \in \mathbb{F}_2$.
$$V_{f_{max}} := \bigoplus_{1 \leq f \leq f_{max}} \bigoplus_{\substack{I=\{i_1, i_2\}, \\ i_2 = i_1 + m - f, \\ I \subset \{1, \dots, m\}}} V_{i_1, i_1, i_2, i_2, \, (a_k^{(i_1,i_2)})}.$$
\label{def:diagonals}
\end{definition}

Theorem~\ref{thm:first_diagonals} justifies the notation $V_{f_{max}}$ by ensuring that $V_{f_{max}}$ doesn't depend on $(a_k^{(i_1,i_2)})_{i_1, i_2, k}$.

\begin{theorem}
$V_{f_{max}}$ is equal to the space generated by all squares of type $I = \{i_1, i_2\} \subset \{1, \dots, m\}$ with $i_2 \leq i_1 + m - f_{max}$. Written algebraically,
\begin{equation}
V_{f_{max}} = \sum_{\substack{I=\{i_1, i_2\}, \\ i_2 \geq i_1 + m - f_{max}, \\ I \subset \{1, \dots, m\}}} \mathbb{F}_2[X_k, k \in \overline{I}].
\label{eq:diagonal}
\end{equation}
\label{thm:first_diagonals}
\end{theorem}

\begin{proof}
The proof is by induction on $f_{max}$. The base case is given by Lemma~\ref{lem:6_cube_1_type}.

Assume that the result is proven for $f_{max}$, with $f_{max} \leq m-2$. Let $I = \{i_1, i_2\} \subset \{1, \dots, m\}$ with $i_2 = i_1 + m - (f_{max}+1)$.

We want to prove that $\mathbb{F}_2[X_k, k \in \overline{I}]$ (i.e. $V_{1, i_1, i_2, m}$, i.e. the space generated by all squares of type $I$) is a subset of $V_{f_{max}+1}$ and will use the fact that $V_{i_1, i_1, i_2, i_2, \, (a_k^{(i_1,i_2)})}$ (i.e. a space generated by some squares of type $I$) is a subset of $V_{f_{max}+1}$ (by definition of $V_{f_{max}+1}$). \\

The first statement of Corollary~\ref{cor:translations_from_another_type} applied to $I$, $k_1=1$ and $k_2$ that ranges from $m$ down to $i_2+1$ gives
$$\forall k_2 \in \{i_2+1, \dots, m\}, V_{1, i_1, i_2, k_2, \, (a_k^{(i_1,i_2)})} \subset V_{1, i_1, i_2, k_2-1, \, (a_k^{(i_1,i_2)})} + \mathbb{F}_2[X_k, k \in \overline{\{i_1, k_2\}}].$$
Therefore
$$V_{1, i_1, i_2, m} \subset V_{1, i_1, i_2, i_2, \, (a_k^{(i_1,i_2)})} + \sum_{k_2 \in \{i_2+1, \dots, m\}} \mathbb{F}_2[X_k, k \in \overline{\{i_1, k_2\}}]$$

The second statement of Corollary~\ref{cor:translations_from_another_type} applied to $I$, $k_1$ that ranges from $1$ to $i_1-1$ and $k_2=i_2$ gives
 $$\forall k_1 \in \{1, \dots, i_1-1\}, V_{k_1, i_1, i_2, i_2, \, (a_k^{(i_1,i_2)})} \subset V_{k_1+1, i_1, i_2, i_2, \, (a_k^{(i_1,i_2)})} + \mathbb{F}_2[X_k, k \in \overline{\{k_1, i_2\}}].$$
Therefore
$$V_{1, i_1, i_2, k_2, \, (a_k^{(i_1,i_2)})} \subset V_{i_1, i_1, i_2, k_2, \, (a_k^{(i_1,i_2)})} + \sum_{k_1 \in \{1, \dots, i_1-1\}} \mathbb{F}_2[X_k, k \in \overline{\{k_1, i_2\}}].$$

Combining these two results gives
{\small
$$\mathbb{F}_2[X_k, k \in \overline{I}] \subset V_{i_1, i_1, i_2, k_2, \, (a_k^{(i_1,i_2)})} + \sum_{k_1 \in \{1, \dots, i_1-1\}} \mathbb{F}_2[X_k, k \in \overline{\{k_1, i_2\}}] + \sum_{k_2 \in \{i_2+1, \dots, m\}} \mathbb{F}_2[X_k, k \in \overline{\{i_1, k_2\}}].$$
}

Since for every $k_1 \leq i_1 - 1, \,\, i_2 - k_1 \geq m - f_{max}$ and for every $k_2 \leq i_2 + 1, \,\, k_2 - i_1 \geq m - f_{max}$, the induction hypothesis gives
$$\mathbb{F}_2[X_k, k \in \overline{I}] \subset V_{i_1, i_1, i_2, k_2, \, (a_k^{(i_1,i_2)})} + V_{f_{max}}.$$

Therefore
$$\mathbb{F}_2[X_k, k \in \overline{I}] \subset V_{f_{max}+1},$$
which proves the result for $f_{max}+1$.
\end{proof}

\section{Planar layout for the Z stabilizer group generators of $QRM_4(1,1)$}
\label{QRM_4_1_1}

The (small) unfolded code described in \cite{ruiz2025unfolded},~Figure~5 is $QRM_4(1,1)$ with one puncture (see Appendix~\ref{sec:puncturing} for the definition of a puncture). The planar layout of a set of generators of the Z stabilizer group of $QRM_4(1,1)$ is depicted in Figure~\ref{fig:layout_small_unfolded}.

\begin{figure}[H]
    \centering
    \includegraphics{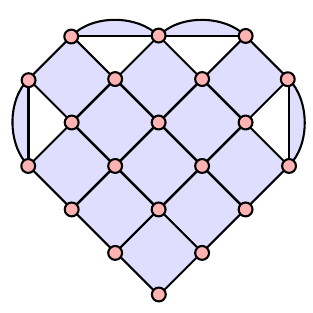}
    \caption{Planar layout of a set of generators of the Z stabilizer group of $QRM_4(1,1)$.  Each vertex is a physical qubit. Each square is a Z stabilizer of weight 4. Each half-round shape is a Z stabilizer of weight 2. Half-round shapes are used to make every weight 4 stabilizer into a square. They are not described in the main text.}
    \label{fig:layout_small_unfolded}
\end{figure}

In this section, we rederive the planar layout of $QRM_4(1,1)$ by grouping two coordinates of the $4$-cube in one direction (the North-West/South-East direction in Figure~\ref{fig:layout_small_unfolded} and \ref{fig:product_of_Gray_code_squares}) and grouping the two other coordinates in anoter direction (the North-East/South-West direction in Figure~\ref{fig:layout_small_unfolded} and \ref{fig:product_of_Gray_code_squares}). In \cite{ruiz2025unfolded}, this planar layout was found with a SAT solver. The viewpoint of grouping coordinates of an $m$-cube to flatten it allows us to derive planar layouts of larger Quantum Reed Muller codes in Sections~\ref{sec:QRM_6_1_1} and \ref{sec:big_unfolded_code} and 3D layouts of larger Quantum Reed Muller codes in Sections~\ref{sec:QRM_6_1_2} and \ref{sec:QRM_7_2_2}.

\begin{figure}[H]
    \centering
    \includegraphics{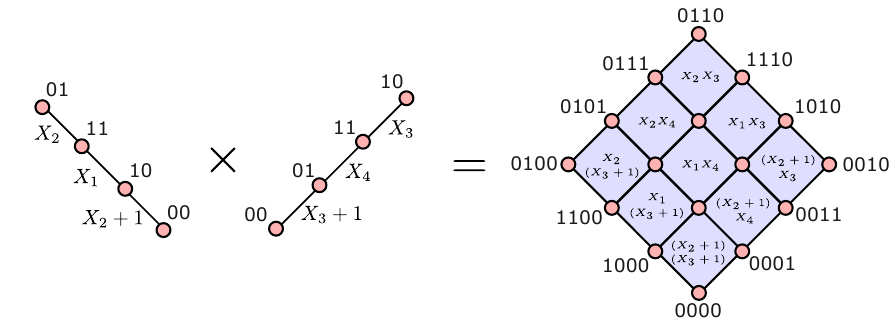}
    \caption{The cartesian product of a basis for edges in the $2$-cube with coordinate indices $1$ and $2$ by a basis for edges in the $2$-cube with coordinate indices $3$ and $4$ yields a basis for the space generated by squares of type $\{i, j\}$ with $i \in \{1,2\}$ and $j \in \{3,4\}$ in the $4$-cube. As depicted in green on Table~\ref{tab:cartesian_product_in_4_cube}, the basis is made of $4$ squares of type $\{1,4\}$, 2 squares of type $\{1,3\}$, 2 squares of type $\{2,4\}$ and 1 square of type $\{2,3\}$. To generate the full $Z$ stabilizer group of $QRM_{4}(1,1)$, Table~\ref{tab:cartesian_product_in_4_cube} shows that a square of type $\{1,2\}$ and a square of type $\{3,4\}$ must be appended to this basis. The square of type $\{3,4\}$ is the upper-left square of Figure~\ref{fig:layout_small_unfolded}. It corresponds to the polynomial $(X_1+1)X_2$. The square of type $\{1,2\}$ is the upper-right square of Figure~\ref{fig:layout_small_unfolded}. It corresponds to the polynomial $X_3(X_4+1)$.}
    \label{fig:product_of_Gray_code_squares}
\end{figure}

\begin{lemma}
The Gray code family $(Y+1, X, Y)$ is a basis of $(\mathbb{F}_2[X, Y])_{\leq 1}$, the space of polynomials of $\mathbb{F}_2[X, Y]$ of degree at most $1$.
\label{lem:Gray_code_square}
\end{lemma}

\begin{proof}
Let $\mathcal{B} = (Y+1, X, Y)$. The proof consists in showing that $\mathcal{B}$ is free and that its cardinal is the dimension of $(\mathbb{F}_2[X, Y])_{\leq 1}$.
\begin{itemize}
\item We prove that $\mathcal{B}$ is free by evaluating its elements succesively on vertices of the $2$-cube in the Gray code order depicted in Figure~\ref{fig:Gray_code_square}. Indeed, let $a, b, c \in \mathbb{F}_2$ be such that $a (Y+1) + b X + c Y = 0$. Evaluting on $00$ shows that $a=0$. Then evaluating on $10$ shows that $b=0$. Therefore $c=0$ and thus $\mathcal{B}$ is a free family.
\item $\dim (\mathbb{F}_2[X, Y])_{\leq 1} = 3 = |\mathcal{B}|$.
\end{itemize}
Therefore $\mathcal{B}$ is a basis of $(\mathbb{F}_2[X, Y])_{\leq 1}$.
\end{proof}

\begin{figure}[H]
    \centering
    \includegraphics{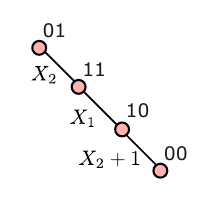}
    \caption{The $4$ vertices of a $2$-cube (i.e. a square) are grouped on a single axis (North-West/South-East) with a Gray code ordering. Each vertex is labeled $x_1 x_2$. Edges are labeled with the (degree 1) polynomial $P$ such that $P(v)=1$ if and only if $v$ is one of the two vertices incident to this edge. Lemma~\ref{lem:Gray_code_square} shows that the 3 edges of the Figure form a basis of $(\mathbb{F}_2[X_1, X_2])_{\leq 1}$.}
    \label{fig:Gray_code_square}
\end{figure}

\begin{table}[H]
\centering
\begin{tblr}{
colspec = |p{2cm}||p{0.5cm}|p{2cm}|p{2cm}|p{2.3cm}|,
cell{2}{4} = {green6},
cell{2}{5} = {green6},
cell{3}{4} = {green6},
cell{3}{5} = {green6},
}
\hline
coordinate indices & 1 & 2 & 3 & 4  \\
\hline \hline
1 &    & $X_3 (X_4+1)$ & $\mathbb{F}_2[X_2]X_4$ & $\mathbb{F}_2[X_2, X_3]$ \\
\hline
2 &    &    & $X_1 X_4$ & $X_1 \mathbb{F}_2[X_3]$ \\
\hline
3 &    &    &    & $(X_1+1) X_2$ \\
\hline
4 &    &    &    &    \\
\hline
\end{tblr}
\caption{Rows and columns are indexed by coordinate indices of the $4$-cube. Starting from the top right corner and going down and left diagonal by diagonal, the entry in $(i,j)$ shows the polynomial subspace that is added (as a direct sum) to the Z stabilizer group when adding squares of type $\{i, j\}$ to the set of generators.}
\label{tab:subspaces_by_type_QRM_4_1_1}
\end{table}

Theorem~\ref{thm:first_diagonals} states that the space generated by the direct sum over all entries of Table~\ref{tab:subspaces_by_type_QRM_4_1_1} is the Z stabilizer group of $QRM_4(1,1)$. Table~\ref{tab:cartesian_product_in_4_cube} gives the dimensions of the spaces of Table~\ref{tab:subspaces_by_type_QRM_4_1_1}.

\begin{table}[H]
\centering
\begin{tblr}{
colspec = |p{2cm}||p{1cm}|p{1cm}|p{1cm}|p{1cm}|,
cell{2}{4} = {green6},
cell{2}{5} = {green6},
cell{3}{4} = {green6},
cell{3}{5} = {green6},
}
\hline
coordinate indices & 1 & 2 & 3 & 4  \\
\hline \hline
1 &    & 1 & 2 & 4 \\
\hline
2 &    &    & 1 & 2 \\
\hline
3 &    &    &    & 1 \\
\hline
4 &    &    &    &    \\
\hline
\end{tblr}
\caption{Rows and columns are indexed by coordinate indices of the $4$-cube. The entry in $(i,j)$ is the number of elements of type $\{i, j\}$ in the chosen basis of the Z stabilizer group of $QRM_4(1,1)$. The four entries with a green background correspond to the cartesian product of spaces of polynomials of degree at most $1$: $(\mathbb{F}_2[X_1, X_2])_{\leq 1} \times (\mathbb{F}_2[X_3, X_4])_{\leq 1}$. The two other entries correspond respectively to a degree $2$ polynomial in $\mathbb{F}_2[X_3, X_4]$ (i.e. a square of type $\{1, 2\}$) and to a degree $2$ polynomial in $\mathbb{F}_2[X_1, X_1]$ (i.e. a square of type $\{3, 4\}$). The Z stabilizer group of $QRM_4(1,1)$ has dimension $11$ which corresponds to the sum of the entries of the Table.}
\label{tab:cartesian_product_in_4_cube}
\end{table}

The four entries in green in Table~\ref{tab:subspaces_by_type_QRM_4_1_1}~(and~\ref{tab:cartesian_product_in_4_cube}) are cartesian products of a subspace of $\mathbb{F}_2[X_1, X_2]$ and a subspace of $\mathbb{F}_2[X_3, X_4]$. We rely on this product structure to obtain a planar layout for the $4$-cube. \\

More precisely, coordinates $1$ and $2$ are unfolded on one axis of the planar layout and coordinates $3$ and $4$ are unfolded on the other axis of the planar layout. To unfold two coordinates on a single axis, the vertices are ordered with the Gray code shown in Figure~\ref{fig:Gray_code_square}. Thus, edges of the planar layout correspond to edges of the cube.

\begin{figure}[H]
    \centering
    % 3 TikZ subfigures
    \subfloat[$4$ squares of type $\{1, 4\}$.]{%
\resizebox{0.3\textwidth}{!}{%
\begin{tikzpicture}
	\begin{pgfonlayer}{nodelayer}
		\node [style=vertex set] (0) at (-2, 2) {1};
		\node [style=vertex set] (1) at (-2, -2) {2};
		\node [style=vertex set] (4) at (2, 2) {3};
		\node [style=vertex set] (5) at (2, -2) {4};
		\node [style=none] (6) at (1.25, -0.75) {\tiny \textit{4}};
	\end{pgfonlayer}
	\begin{pgfonlayer}{edgelayer}
		\draw (0) to (5);
	\end{pgfonlayer}
\end{tikzpicture}
}
    }\hspace{1em}
    \subfloat[$4$ squares of type $\{1, 4\}$, $2$ squares of type $\{1, 3\}$ and $2$ squares of type $\{2, 4\}$.  In total, $8$ squares of $3$ different types.]{%
\resizebox{0.3\textwidth}{!}{%
\begin{tikzpicture}
	\begin{pgfonlayer}{nodelayer}
		\node [style=vertex set] (0) at (-2, 2) {1};
		\node [style=vertex set] (1) at (-2, -2) {2};
		\node [style=vertex set] (4) at (2, 2) {3};
		\node [style=vertex set] (5) at (2, -2) {4};
		\node [style=none] (6) at (1.25, -0.75) {\tiny \textit{4}};
		\node [style=none] (7) at (-0.25, 1.75) {\tiny \textit{2}};
		\node [style=none] (8) at (-0.25, -1.75) {\tiny \textit{2}};
	\end{pgfonlayer}
	\begin{pgfonlayer}{edgelayer}
		\draw (0) to (4);
		\draw (0) to (5);
		\draw (1) to (5);
	\end{pgfonlayer}
\end{tikzpicture}
}
    }\hspace{1em}
    \subfloat[$4$ squares of type $\{1, 4\}$, $2$ squares of type $\{1, 3\}$, $2$ squares of type $\{2, 4\}$, $1$ square of type $\{2, 3\}$, $1$ square of type $\{1, 2\}$ and  $1$ square of type $\{3, 4\}$.  In total, $11$ squares of $6$ different types.]{%
\resizebox{0.3\textwidth}{!}{%
\begin{tikzpicture}
	\begin{pgfonlayer}{nodelayer}
		\node [style=vertex set] (0) at (-2, 2) {1};
		\node [style=vertex set] (1) at (-2, -2) {2};
		\node [style=vertex set] (4) at (2, 2) {3};
		\node [style=vertex set] (5) at (2, -2) {4};
		\node [style=none] (6) at (1.25, -0.75) {\tiny \textit{4}};
		\node [style=none] (7) at (-0.25, 1.75) {\tiny \textit{2}};
		\node [style=none] (8) at (-0.25, -1.75) {\tiny \textit{2}};
		\node [style=none] (9) at (0.75, 1.25) {\tiny \textit{1}};
		\node [style=none] (10) at (1.75, 0.25) {\tiny \textit{1}};
		\node [style=none] (11) at (-1.75, 0) {\tiny \textit{1}};
	\end{pgfonlayer}
	\begin{pgfonlayer}{edgelayer}
		\draw (0) to (4);
		\draw (0) to (5);
		\draw (1) to (5);
		\draw (1) to (4);
		\draw (0) to (1);
		\draw (4) to (5);
	\end{pgfonlayer}
\end{tikzpicture}
}
    }
    \caption{The $4$ nodes on each subfigure represent the $4$ coordinate indices of the $4$-cube. A type of square is characterized by $2$ coordinate indices and corresponds therefore to an edge on this Figure. There are $4$ squares of each type. It is sufficient to add a lower power of $2$ (i.e. $2$ or $1$) number of squares of each type to obtain a basis of $RM(4, 1)$. The edge label is the number of squares of a given type in the basis of $RM(4, 1)$ depicted in Figure~\ref{fig:product_of_Gray_code_squares}. Figure (a) corresponds to the top right entry of Table~\ref{tab:subspaces_by_type_QRM_4_1_1}. Figure (b) corresponds to the top right entry of Table~\ref{tab:subspaces_by_type_QRM_4_1_1} and the first diagonal beneath it. Figure (c) corresponds all entries of Table~\ref{tab:subspaces_by_type_QRM_4_1_1}.}
    \label{fig:tikz_grid_small_unfolded}
\end{figure}
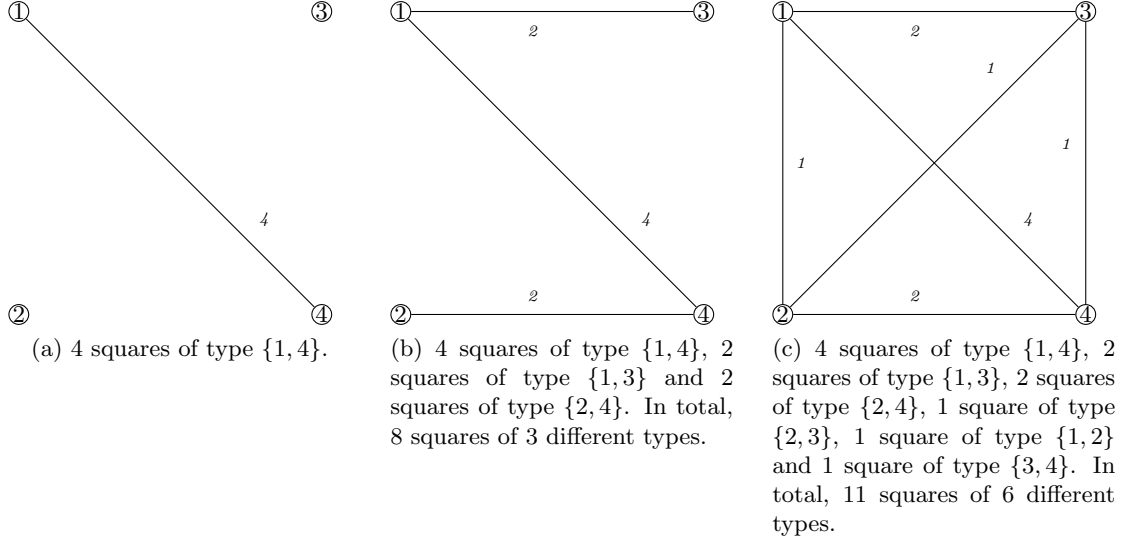

\section{Planar layout for the Z stabilizer group generators of $QRM_6(1,1)$}
\label{sec:QRM_6_1_1}

In $QRM_6(1,1)$, the Z stabilizer group is generated by squares (since $r=1$) of the $6$-cube (since $m=6$). We show in this section that the Z stabilizer group of $QRM_6(1,1)$ admits the set of generators of Figure~\ref{fig:layout_QRM_6_1_1}.

\begin{figure}[H]
    \centering
    \includegraphics{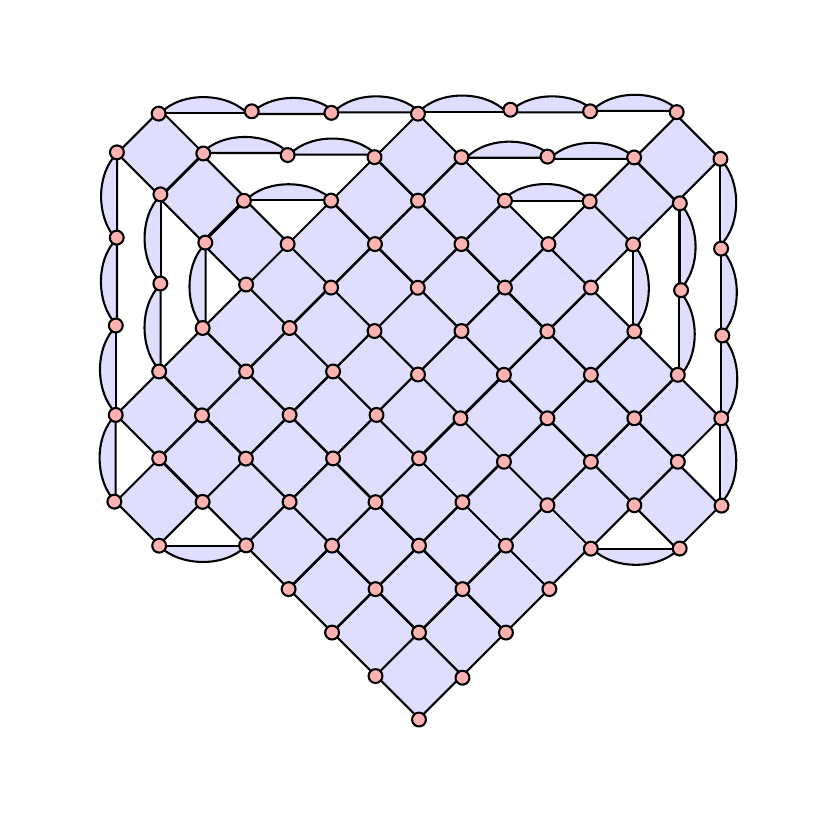}
    \caption{Planar layout of a set of generators of the Z stabilizer group of $QRM_6(1,1)$. Each vertex is a physical qubit. Each square is a Z stabilizer of weight 4. Each half-round shape is a Z stabilizer of weight 2. Half-round shapes are used to make every weight 4 stabilizer into a geometrical square. They are not described in the main text.}
    \label{fig:layout_QRM_6_1_1}
\end{figure}

To find a planar layout for a basis of the Z stabilizer group of $QRM_6(1,1)$ made of squares of the $6$-cube, we follow the strategy of Section~\ref{QRM_4_1_1} and partition coordinate indices of the $6$-cube into two groups: $\{1,2,3\}$ and $\{4,5,6\}$. Recall from Definition~\ref{def:subcube} that squares of the $6$-cube correspond to degree $4$ polynomials of $\mathbb{F}_2[X_1, X_2, X_3, X_4, X_5, X_6]$.

\begin{lemma}
The Gray code family
$$\left( (Y+1)(Z+1), X(Z+1), Y(Z+1), (X+1)Y, (Y+1)Z, XZ, YZ \right)$$
is a basis of $(\mathbb{F}_2[X, Y, Z])_{\leq 2}$, the space of polynomials of $\mathbb{F}_2[X, Y, Z]$ of degree at most $2$.
\label{lem:Gray_code_cube}
\end{lemma}

\begin{proof}
Let $\mathcal{B} = \left( (Y+1)(Z+1), X(Z+1), Y(Z+1), (X+1)Y, (Y+1)Z, XZ, YZ \right)$. The proof consists in showing that $\mathcal{B}$ is free and that its cardinal is the dimension of $(\mathbb{F}_2[X, Y, Z])_{\leq 2}$.
\begin{itemize}
\item We prove that $\mathcal{B}$ is free by evaluating its elements succesively on vertices of the $3$-cube in the Gray code order depicted in Figure~\ref{fig:Gray_code_cube_123}. Indeed, let $a, b, c, d, e, f, g \in \mathbb{F}_2$ be such that $a (Y+1)(Z+1) + b X(Z+1) + c Y(Z+1) + d (X+1)Y + e (Y+1)Z + f XZ + g YZ = 0$. Evaluting on $000$ shows that $a=0$. Then evaluating on $100$ shows that $b=0$. Then evaluating on $110$ shows that $c=0$. Then evaluating on $010$ shows that $d=0$. Then evaluating on $011$ shows that $e=0$. Then evaluating on $111$ shows that $f=0$. Therefore $g=0$ and thus $\mathcal{B}$ is a free family.
\item $\dim (\mathbb{F}_2[X, Y, Z])_{\leq 2} = 7 = |\mathcal{B}|$.
\end{itemize}
Therefore $\mathcal{B}$ is a basis of $(\mathbb{F}_2[X, Y, Z])_{\leq 2}$.
\end{proof}

Applying Lemma~\ref{lem:Gray_code_cube} to polynomials in $X_1$, $X_2$ and $X_3$ gives that $(\mathbb{F}_2[X_1, X_2, X_3])_{\leq 2})$ admits as a basis the following $7$ polynomials corresponding to the $7$ edges of Figure~\ref{fig:Gray_code_cube_123}:
\begin{align}
&(\mathbb{F}_2[X_1, X_2, X_3])_{\leq 2}) \notag \\
= &\Span(\{(X_2+1)(X_3+1) \}) \notag \\
\oplus &\Span(\{X_1(X_3+1) \}) \notag \\
\oplus &\Span(\{X_2(X_3+1) \}) \notag \\
\oplus &\Span(\{(X_1+1)X_2 \}) \notag \\
\oplus &\Span(\{X_2 X_3 \}) \notag \\
\oplus &\Span(\{X_1 X_3 \}) \notag \\
\oplus &\Span(\{(X_2+1) X_3 \})
\label{eq:Gray_code_cube_123}
\end{align}

\begin{figure}[H]
    \centering
    \includegraphics{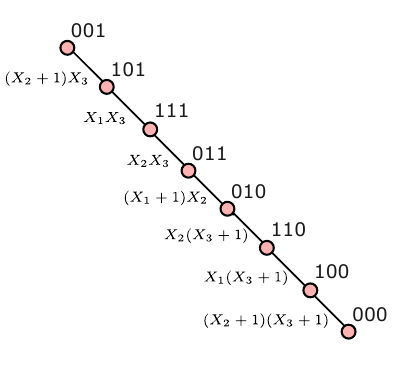}
    \caption{Gray code ordering for edges in the $3$-cube with coordinate indices $1$, $2$ and $3$. For instance the polynomial $X_1 (X_3+1)$ evaluates to $1$ on vertex $v = x_1 x_2 x_3$ of the $3$-cube if and only if $v=100$ or $v=110$ (since the term $X_1$ enforces that $x_1=1$ and the term $X_3+1$ enforces that $x_3=0$).}
    \label{fig:Gray_code_cube_123}
\end{figure}

Similarly, applying Lemma~\ref{lem:Gray_code_cube} to polynomials in $X_4$, $X_5$ and $X_6$ gives that $(\mathbb{F}_2[X_4, X_5, X_6])_{\leq 2})$ admits as a basis the following $7$ polynomials corresponding to the $7$ edges of Figure~\ref{fig:Gray_code_cube_123}:
\begin{align}
&(\mathbb{F}_2[X_4, X_5, X_6])_{\leq 2}) \notag \\
= &\Span(\{(X_4+1)(X_5+1) \}) \notag \\
\oplus &\Span(\{(X_4+1) X_6 \}) \notag \\
\oplus &\Span(\{(X_4+1) X_6 \}) \notag \\
\oplus &\Span(\{X_5 (X_6+1) \}) \notag \\
\oplus &\Span(\{X_4 X_5 \}) \notag \\
\oplus &\Span(\{X_4 X_6 \}) \notag \\
\oplus &\Span(\{X_4 (X_5+1) \})
\label{eq:Gray_code_cube_456}
\end{align}

\begin{figure}[H]
    \centering
    \includegraphics{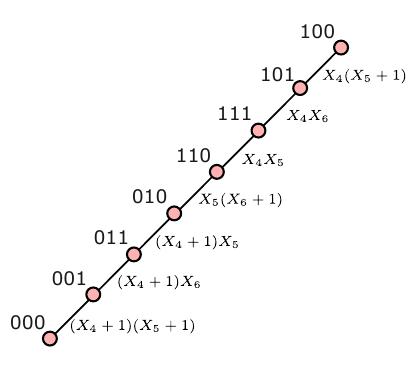}
    \caption{Gray code ordering for edges in the $3$-cube with coordinate indices $4$, $5$ and $6$. For instance the polynomial $X_4 X_5$ evaluates to $1$ on vertex $v = x_4 x_5 x_6$ of the $3$-cube if and only if $v=110$ or $v=111$ (since the term $X_4$ enforces that $x_4=1$ and the term $X_5$ enforces that $x_5=1$).}
    \label{fig:Gray_code_cube_456}
\end{figure}

Therefore, $(\mathbb{F}_2[X_1, X_2, X_3])_{\leq 2} \times (\mathbb{F}_2[X_4, X_5, X_6])_{\leq 2}$ admits as a basis the $49$ degree $4$ polynomials that are the product of a degree $2$ polynomial of Equation \ref{eq:Gray_code_cube_123} by a degree $2$ polynomial of Equation \ref{eq:Gray_code_cube_456}. These $49$ polynomials correspond to the $7 \times 7$ grid of squares of Figure~\ref{fig:layout_QRM_6_1_1}. \\

With the above Gray code orderings, the Z stabilizer group generated by the weight 4 stabilizers of Figure~\ref{fig:layout_QRM_6_1_1} therefore is equal to the following direct sum

\begin{align}
&(\mathbb{F}_2[X_1, X_2, X_3])_{\leq 2} \times (\mathbb{F}_2[X_4, X_5, X_6])_{\leq 2} \notag \\
 \notag \\
\oplus \, &\Span(\{X_2 X_4 (X_5+1) (X_6+1) \}) \notag \\
\oplus \, &\Span(\{X_1 X_4 (X_5+1) (X_6+1) \}) \notag \\
\oplus \, &\Span(\{(X_2+1) X_4 (X_5+1) (X_6+1) \}) \notag \\
\oplus \, &\Span(\{X_3 (X_4+1) (X_5+1) (X_6+1) \}) \notag  \\
 \notag \\
\oplus \, &\Span(\{(X_1+1) (X_2+1) X_3 X_5\}) \notag \\
\oplus \, &\Span(\{(X_1+1) (X_2+1) X_3 X_6\}) \notag \\
\oplus \, &\Span(\{(X_1+1) (X_2+1) X_3 (X_5+1)\}) \notag \\
\oplus \, &\Span(\{(X_1+1) (X_2+1) (X_3+1) X_4\})
\label{polynomial_basis_6_variables_degree_at_most_4}
\end{align}

We show that the space described in Equation~\ref{polynomial_basis_6_variables_degree_at_most_4} is the Z stabilizer group of $QRM_6(1,1)$. Indeed to generate the (dimension 57) Z stabilizer group of $QRM_6(1,1)$, Theorem~\ref{thm:first_diagonals} shows that it is sufficient to append to the basis of $(\mathbb{F}_2[X_1, X_2, X_3])_{\leq 2} \times (\mathbb{F}_2[X_4, X_5, X_6])_{\leq 2}$:
\begin{itemize}
\item $2$ polynomials in $\mathbb{F}_2[X_2, X_4, X_5, X_6]$ (corresponding to squares of type $\{1, 3\}$)
\item $1$ polynomial in $\mathbb{F}_2[X_1, X_4, X_5, X_6]$ (corresponding to a square of type $\{2, 3\}$)
\item $1$ polynomial in $\mathbb{F}_2[X_3, X_4, X_5, X_6]$ (corresponding to a square of type $\{1, 2\}$)
\item $2$ polynomials in $\mathbb{F}_2[X_1, X_2, X_3, X_5]$ (corresponding to squares of type $\{4, 6\}$)
\item $1$ polynomial in $\mathbb{F}_2[X_1, X_2, X_3, X_4]$ (corresponding to a square of type $\{5, 6\}$)
\item $1$ polynomial in $\mathbb{F}_2[X_1, X_2, X_3, X_6]$ (corresponding to a square of type $\{4, 5\}$).
\end{itemize}
\vspace{10pt}
The $3$ polynomials that correspond to the $3$ upper-left squares of Figure~\ref{fig:layout_QRM_6_1_1} are
\begin{align*}
                    &X_2 X_4 (X_5+1) (X_6+1), \\
                    &X_1 X_4 (X_5+1) (X_6+1)  \\
\text{and   } &(X_2+1) X_4 (X_5+1) (X_6+1).
\end{align*}
The $3$ polynomials that correspond to the $3$ upper-right squares of Figure~\ref{fig:layout_QRM_6_1_1} are
\begin{align*}
                    &(X_1+1) (X_2+1) X_3 X_5, \\
                    &(X_1+1) (X_2+1) X_3 X_6 \\
\text{and   }  &(X_1+1) (X_2+1) X_3 (X_5+1).
\end{align*}
The polynomial that corresponds to the lower-left square of Figure~\ref{fig:layout_QRM_6_1_1} is
$$X_3 (X_4+1) (X_5+1) (X_6+1).$$
The polynomial that corresponds to the lower-right square of Figure~\ref{fig:layout_QRM_6_1_1} is
$$(X_1+1) (X_2+1) (X_3+1) X_4.$$

\begin{table}[H]
\centering
\begin{tblr}{
colspec = |p{0.1cm}||p{0.1cm}|p{2.3cm}|p{2.4cm}|p{1.8cm}|p{2.0cm}|p{2.5cm}|,
}
\hline
   & 1 & 2 & 3 & 4 & 5 & 6  \\
\hline \hline
1&
&\scriptsize \makecell{$X_3 (X_4+1)$                                \\ $(X_5+1) (X_6+1)$}
&\scriptsize \makecell{$\mathbb{F}_2[X_2] X_4$                \\ $(X_5+1) (X_6+1)$}
&\scriptsize \makecell{$\mathbb{F}_2[X_2, X_3]$               \\ $X_5 (X_6+1)$}
&\scriptsize \makecell{$\mathbb{F}_2[X_2, X_3, X_4]$        \\ $X_6$}
&\scriptsize \makecell{$\mathbb{F}_2[X_2, X_3, X_4, X_5]$ \\  } \\
\hline
2&&
&\scriptsize \makecell{$X_1 X_4$                              \\ $(X_5+1) (X_6+1)$}
&\scriptsize \makecell{$X_1 \mathbb{F}_2[X_3]$ \\ $X_5 (X_6+1)$}
&\scriptsize \makecell{$X_1 \mathbb{F}_2[X_3, X_4]$ \\ $X_6$}
&\scriptsize \makecell{$X_1$                                     \\ $\mathbb{F}_2[X_3, X_4, X_5]$}   \\
\hline
3&&&
&\scriptsize \makecell{$(X_1+1) X_2$ \\ $X_5 (X_6+1)$}
&\scriptsize \makecell{$(X_1+1) X_2$ \\ $\mathbb{F}_2[X_4] X_6$}
&\scriptsize \makecell{$(X_1+1) X_2$ \\ $\mathbb{F}_2[X_4, X_5]$}   \\
\hline
4&&&&
&\scriptsize \makecell{$(X_1+1)$ \\ $(X_2+1) X_3 X_6$}
&\scriptsize \makecell{$(X_1+1) (X_2+1)$ \\ $X_3 \mathbb{F}_2[X_5]$}  \\
\hline
5&&&&&
&\scriptsize \makecell{$(X_1+1) (X_2+1)$ \\ $(X_3+1) X_4$}   \\
\hline
6&&&&&&   \\
\hline
\end{tblr}
\caption{Rows and columns are indexed by coordinate indices of the $6$-cube. Starting from the top right corner and going down and left diagonal by diagonal, the entry in $(i,j)$ shows the polynomial subspace that is added (as a direct sum) when adding squares of type $\{i, j\}$.}
\label{tab:subspaces_by_type_QRM_6_1_1}
\end{table}

\begin{table}[H]
\centering
\begin{tblr}{
colspec = |p{2cm}||p{1cm}|p{1cm}|p{1cm}|p{1cm}|p{1cm}|p{1cm}|,
cell{2}{5} = {green6},
cell{2}{6} = {green6},
cell{2}{7} = {green6},
cell{3}{5} = {green6},
cell{3}{6} = {green6},
cell{3}{7} = {green6},
cell{4}{5} = {green6},
cell{4}{6} = {green6},
cell{4}{7} = {green6},
}
\hline
coordinate indices & 1 & 2 & 3 & 4 & 5 & 6  \\
\hline \hline
1 &    & 1 & 2 & 4 & 8 & 16 \\
\hline
2 &    &    & 1 & 2 & 4 & 8   \\
\hline
3 &    &    &    & 1 & 2 & 4   \\
\hline
4 &    &    &    &    & 1 & 2   \\
\hline
5 &    &    &    &    &    & 1   \\
\hline
6 &   &    &    &    &    &      \\
\hline
\end{tblr}
\caption{Rows and columns are indexed by coordinate indices of the $6$-cube. The entry in $(i,j)$ is the number of elements of type $\{i, j\}$ in the chosen basis of the Z stabilizer group of $QRM_6(1,1)$. The nine entries with a green background correspond to the cartesian product $(\mathbb{F}_2[X_1, X_2, X_3])_{\leq 2} \times (\mathbb{F}_2[X_4, X_5, X_6])_{\leq 2}$. The six other entries correspond to squares whose type is either a subset of $\{1, 2, 3\}$ or a subset of $\{4, 5, 6\})$. The Z stabilizer group has dimension $57$.}
\label{tab:cartesian_product_in_6_cube}
\end{table}

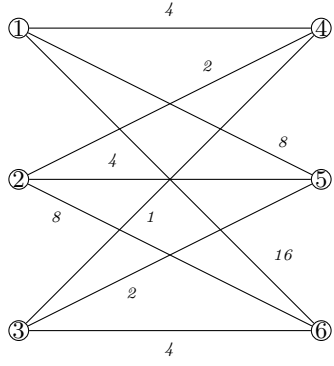
\begin{figure}[H]%
    \centering
\begin{tikzpicture}
	\begin{pgfonlayer}{nodelayer}
		\node [style=vertex set] (0) at (-2, 2) {1};
		\node [style=vertex set] (1) at (-2, 0) {2};
		\node [style=vertex set] (2) at (-2, -2) {3};
		\node [style=vertex set] (3) at (2, 2) {4};
		\node [style=vertex set] (4) at (2, 0) {5};
		\node [style=vertex set] (5) at (2, -2) {6};
		\node [style=none] (6) at (1.5, -1) {\tiny \textit{16}};
		\node [style=none] (7) at (1.5, 0.5) {\tiny \textit{8}};
		\node [style=none] (8) at (-1.5, -0.5) {\tiny \textit{8}};
		\node [style=none] (9) at (0, 2.25) {\tiny \textit{4}};
		\node [style=none] (10) at (-0.75, 0.25) {\tiny \textit{4}};
		\node [style=none] (11) at (0, -2.25) {\tiny \textit{4}};
		\node [style=none] (12) at (0.5, 1.5) {\tiny \textit{2}};
		\node [style=none] (13) at (-0.5, -1.5) {\tiny \textit{2}};
		\node [style=none] (14) at (-0.25, -0.5) {\tiny \textit{1}};
	\end{pgfonlayer}
	\begin{pgfonlayer}{edgelayer}
		\draw (0) to (4);
		\draw (0) to (3);
		\draw (0) to (5);
		\draw (1) to (3);
		\draw (1) to (4);
		\draw (1) to (5);
		\draw (2) to (5);
		\draw (2) to (4);
		\draw (2) to (3);
	\end{pgfonlayer}
\end{tikzpicture}
\caption{Taking the cartesian product of the $7$ edges that correspond to a basis of $(\mathbb{F}_2[X_1, X_2, X_3])_{\leq 2}$ by the $7$ edges that correspond to a basis of $(\mathbb{F}_2[X_4, X_5, X_6])_{\leq 2}$ gives the $49$ squares that correspond to the space generated by all squares of $9$ different types: $\{1, 4\}$, $\{1, 5\}$, $\{1, 6\}$, $\{2, 4\}$, $\{2, 5\}$, $\{2, 6\}$, $\{3, 4\}$, $\{3, 5\}$, $\{3, 6\}$. It corresponds to the green background entries in Table~\ref{tab:cartesian_product_in_6_cube}. On this figure, an edge represents a square type and is labeled by the number of squares of this type in the 2D layout of Figure~\ref{fig:layout_QRM_6_1_1}.}%
\end{figure}

Equation~\ref{6_cube_direct_sum_of_polynomial_space} groups the spaces described in Table~\ref{tab:subspaces_by_type_QRM_6_1_1} by North-West to South-East diagonals (starting at the top-right corner). Theorem~\ref{thm:first_diagonals} states that the direct sum of the $f_{max}$ first diagonals of Table~\ref{tab:subspaces_by_type_QRM_6_1_1} (i.e. the $f_{max}$ first blocks\footnote{Blocks of Equation~\ref{6_cube_direct_sum_of_polynomial_space} are seperated by an empty line.} of Equation~\ref{6_cube_direct_sum_of_polynomial_space}) is the space generated by all squares of type $\{i_1, i_2\}$ with $i_2 \geq i_1 + 6 - f_{max}$.

\begin{align}
&(\mathbb{F}_2[X_1, X_2, X_3, X_4, X_5, X_6])_{\leq 4} \notag \\
 \notag \\
=
&{\color{green6} \mathbb{F}_2[X_2, X_3, X_4, X_5] \notag} \\
& \notag \\
&{\color{green6}\oplus \mathbb{F}_2[X_2, X_3, X_4] \, X_6 \notag} \\
&{\color{green6}\oplus X_1 \, \mathbb{F}_2[X_3, X_4, X_5] \notag} \\
& \notag \\
&{\color{green6}\oplus \mathbb{F}_2[X_2, X_3] \, X_5 (X_6+1) \notag} \\
&{\color{green6}\oplus X_1 \, \mathbb{F}_2[X_3, X_4] \, X_6 \notag} \\
&{\color{green6}\oplus (X_1+1) X_2 \, \mathbb{F}_2[X_4, X_5] \notag} \\
& \notag \\
&\oplus \mathbb{F}_2[X_2] \, X_4 (X_5+1) (X_6+1) \notag \\
&{\color{green6}\oplus X_1 \, \mathbb{F}_2[X_3] \, X_5 (X_6+1) \notag} \\
&{\color{green6}\oplus (X_1+1) X_2 \, \mathbb{F}_2[X_4] \, X_6 \notag} \\
&\oplus (X_1+1) (X_2+1) X_3 \, \mathbb{F}_2[X_5] \notag \\
& \notag \\
&\oplus X_3 (X_4+1) (X_5+1) (X_6+1) \notag \\
&\oplus X_1 X_4 (X_5+1) (X_6+1) \notag \\
&{\color{green6}\oplus (X_1+1) X_2 X_5 (X_6+1) \notag} \\
&\oplus (X_1+1) (X_2+1) X_3 X_6 \notag \\
&\oplus (X_1+1) (X_2+1) (X_3+1) X_4.
\label{6_cube_direct_sum_of_polynomial_space}
\end{align}

Note that for $f \in \{1,\dots,5\}$, block $f$ of Equation~\ref{6_cube_direct_sum_of_polynomial_space} can be written as
$$\bigoplus_{\substack{I=\{i_1, i_2\}, \\ i_2 = i_1 + 6 - f, \\ I \subset \{1, \dots, 6\}}} V_{i_1, i_1, i_2, i_2, \, (a_k^{(i_1,i_2)})}.$$

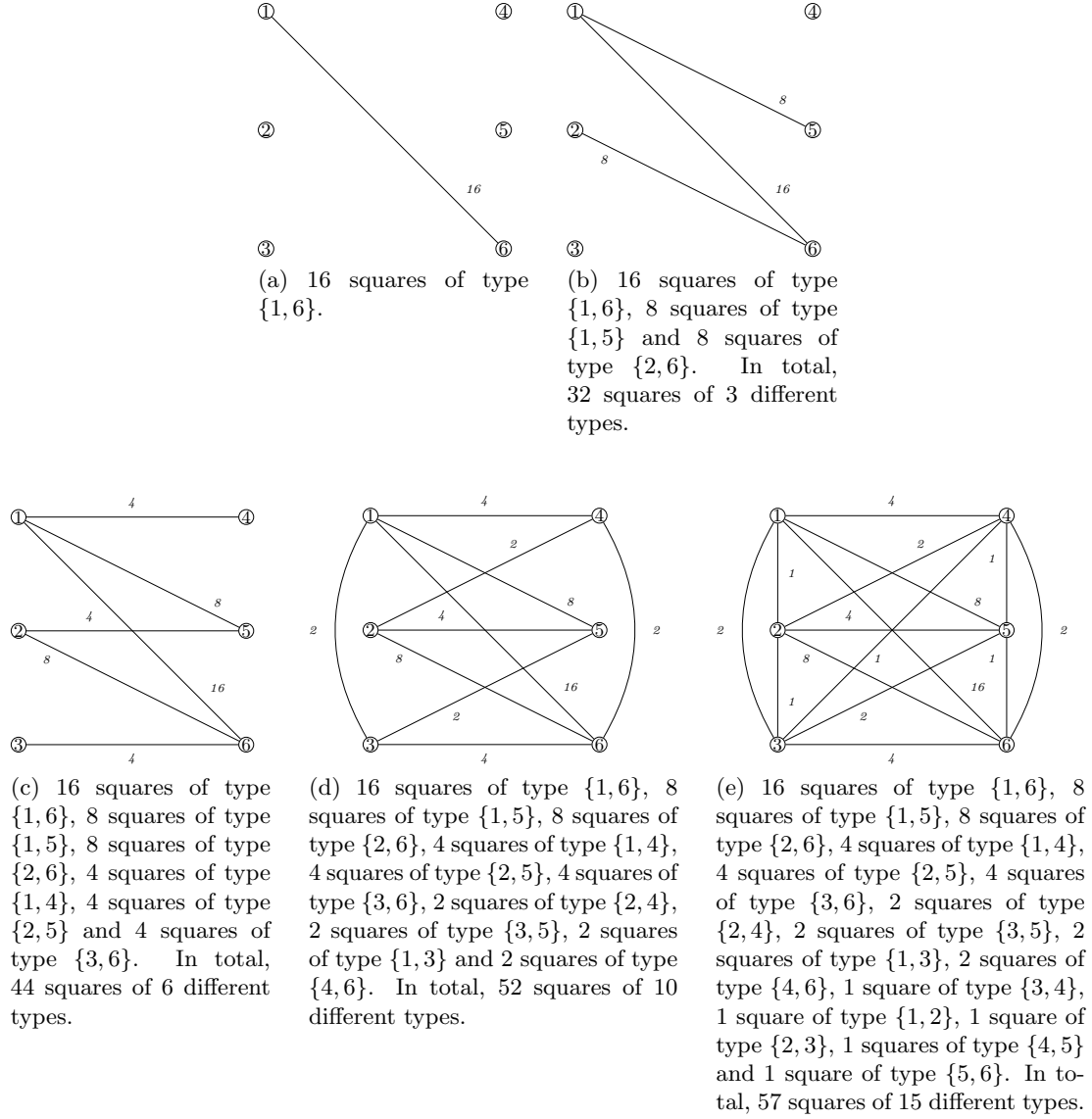
\begin{figure}[H]
    \centering
    % First row: 2 TikZ subfigures
    \subfloat[$16$ squares of type $\{1, 6\}$.]{%
\resizebox{0.24\textwidth}{!}{%
\begin{tikzpicture}
	\begin{pgfonlayer}{nodelayer}
		\node [style=vertex set] (0) at (-2, 2) {1};
		\node [style=vertex set] (1) at (-2, 0) {2};
		\node [style=vertex set] (2) at (-2, -2) {3};
		\node [style=vertex set] (3) at (2, 2) {4};
		\node [style=vertex set] (4) at (2, 0) {5};
		\node [style=vertex set] (5) at (2, -2) {6};
		\node [style=none] (6) at (1.5, -1) {\tiny \textit{16}};
	\end{pgfonlayer}
	\begin{pgfonlayer}{edgelayer}
		\draw (0) to (5);
	\end{pgfonlayer}
\end{tikzpicture}
}
    }\hspace{1em}
    \subfloat[$16$ squares of type $\{1, 6\}$, 8 squares of type $\{1, 5\}$ and 8 squares of type $\{2, 6\}$. In total, $32$ squares of $3$ different types.]{%
\resizebox{0.24\textwidth}{!}{%
\begin{tikzpicture}
	\begin{pgfonlayer}{nodelayer}
		\node [style=vertex set] (0) at (-2, 2) {1};
		\node [style=vertex set] (1) at (-2, 0) {2};
		\node [style=vertex set] (2) at (-2, -2) {3};
		\node [style=vertex set] (3) at (2, 2) {4};
		\node [style=vertex set] (4) at (2, 0) {5};
		\node [style=vertex set] (5) at (2, -2) {6};
		\node [style=none] (6) at (1.5, -1) {\tiny \textit{16}};
		\node [style=none] (7) at (1.5, 0.5) {\tiny \textit{8}};
		\node [style=none] (8) at (-1.5, -0.5) {\tiny \textit{8}};
	\end{pgfonlayer}
	\begin{pgfonlayer}{edgelayer}
		\draw (0) to (4);
		\draw (0) to (5);
		\draw (1) to (5);
	\end{pgfonlayer}
\end{tikzpicture}
}
    }

    \vspace{1em} % Space between rows

    % Second row: 3 TikZ subfigures
    \subfloat[$16$ squares of type $\{1, 6\}$, $8$ squares of type $\{1, 5\}$, $8$ squares of type $\{2, 6\}$, $4$ squares of type $\{1, 4\}$, $4$ squares of type $\{2, 5\}$ and $4$ squares of type $\{3, 6\}$. In total, $44$ squares of $6$ different types.]{%
\resizebox{0.23\textwidth}{!}{%
\begin{tikzpicture}
	\begin{pgfonlayer}{nodelayer}
		\node [style=vertex set] (0) at (-2, 2) {1};
		\node [style=vertex set] (1) at (-2, 0) {2};
		\node [style=vertex set] (2) at (-2, -2) {3};
		\node [style=vertex set] (3) at (2, 2) {4};
		\node [style=vertex set] (4) at (2, 0) {5};
		\node [style=vertex set] (5) at (2, -2) {6};
		\node [style=none] (6) at (1.5, -1) {\tiny \textit{16}};
		\node [style=none] (7) at (1.5, 0.5) {\tiny \textit{8}};
		\node [style=none] (8) at (-1.5, -0.5) {\tiny \textit{8}};
		\node [style=none] (9) at (0, 2.25) {\tiny \textit{4}};
		\node [style=none] (10) at (-0.75, 0.25) {\tiny \textit{4}};
		\node [style=none] (11) at (0, -2.25) {\tiny \textit{4}};
	\end{pgfonlayer}
	\begin{pgfonlayer}{edgelayer}
		\draw (0) to (4);
		\draw (0) to (3);
		\draw (0) to (5);
		\draw (1) to (4);
		\draw (1) to (5);
		\draw (2) to (5);
	\end{pgfonlayer}
\end{tikzpicture}
}
    }\hspace{1em}
    \subfloat[$16$ squares of type $\{1, 6\}$, $8$ squares of type $\{1, 5\}$, $8$ squares of type $\{2, 6\}$, $4$ squares of type $\{1, 4\}$, $4$ squares of type $\{2, 5\}$, $4$ squares of type $\{3, 6\}$, $2$ squares of type $\{2, 4\}$, $2$ squares of type $\{3, 5\}$, $2$ squares of type $\{1, 3\}$ and $2$ squares of type $\{4, 6\}$.  In total, $52$ squares of $10$ different types.]{%
\resizebox{0.33\textwidth}{!}{%
\begin{tikzpicture}
	\begin{pgfonlayer}{nodelayer}
		\node [style=vertex set] (0) at (-2, 2) {1};
		\node [style=vertex set] (1) at (-2, 0) {2};
		\node [style=vertex set] (2) at (-2, -2) {3};
		\node [style=vertex set] (3) at (2, 2) {4};
		\node [style=vertex set] (4) at (2, 0) {5};
		\node [style=vertex set] (5) at (2, -2) {6};
		\node [style=none] (6) at (1.5, -1) {\tiny \textit{16}};
		\node [style=none] (7) at (1.5, 0.5) {\tiny \textit{8}};
		\node [style=none] (8) at (-1.5, -0.5) {\tiny \textit{8}};
		\node [style=none] (9) at (0, 2.25) {\tiny \textit{4}};
		\node [style=none] (10) at (-0.75, 0.25) {\tiny \textit{4}};
		\node [style=none] (11) at (0, -2.25) {\tiny \textit{4}};
		\node [style=none] (12) at (0.5, 1.5) {\tiny \textit{2}};
		\node [style=none] (13) at (-0.5, -1.5) {\tiny \textit{2}};
		\node [style=none] (14) at (-3, 0) {\tiny \textit{2}};
		\node [style=none] (15) at (3, 0) {\tiny \textit{2}};
	\end{pgfonlayer}
	\begin{pgfonlayer}{edgelayer}
		\draw (0) to (4);
		\draw (0) to (3);
		\draw (0) to (5);
		\draw (1) to (3);
		\draw (1) to (4);
		\draw (1) to (5);
		\draw (2) to (5);
		\draw (2) to (4);
		\draw [bend left] (3) to (5);
		\draw [bend right] (0) to (2);
	\end{pgfonlayer}
\end{tikzpicture}
}
    }\hspace{1em}
    \subfloat[$16$ squares of type $\{1, 6\}$, $8$ squares of type $\{1, 5\}$, $8$ squares of type $\{2, 6\}$, $4$ squares of type $\{1, 4\}$, $4$ squares of type $\{2, 5\}$, $4$ squares of type $\{3, 6\}$, $2$ squares of type $\{2, 4\}$, $2$ squares of type $\{3, 5\}$, $2$ squares of type $\{1, 3\}$, $2$ squares of type $\{4, 6\}$, $1$ square of type $\{3, 4\}$, $1$ square of type $\{1, 2\}$, $1$ square of type $\{2, 3\}$, $1$ squares of type $\{4, 5\}$ and $1$ square of type $\{5, 6\}$.  In total, $57$ squares of $15$ different types.]{%
\resizebox{0.33\textwidth}{!}{%
\begin{tikzpicture}
	\begin{pgfonlayer}{nodelayer}
		\node [style=vertex set] (0) at (-2, 2) {1};
		\node [style=vertex set] (1) at (-2, 0) {2};
		\node [style=vertex set] (2) at (-2, -2) {3};
		\node [style=vertex set] (3) at (2, 2) {4};
		\node [style=vertex set] (4) at (2, 0) {5};
		\node [style=vertex set] (5) at (2, -2) {6};
		\node [style=none] (6) at (1.5, -1) {\tiny \textit{16}};
		\node [style=none] (7) at (1.5, 0.5) {\tiny \textit{8}};
		\node [style=none] (8) at (-1.5, -0.5) {\tiny \textit{8}};
		\node [style=none] (9) at (0, 2.25) {\tiny \textit{4}};
		\node [style=none] (10) at (-0.75, 0.25) {\tiny \textit{4}};
		\node [style=none] (11) at (0, -2.25) {\tiny \textit{4}};
		\node [style=none] (12) at (0.5, 1.5) {\tiny \textit{2}};
		\node [style=none] (13) at (-0.5, -1.5) {\tiny \textit{2}};
		\node [style=none] (14) at (-3, 0) {\tiny \textit{2}};
		\node [style=none] (15) at (3, 0) {\tiny \textit{2}};
		\node [style=none] (16) at (-1.75, -1.25) {\tiny \textit{1}};
		\node [style=none] (17) at (-1.75, 1) {\tiny \textit{1}};
		\node [style=none] (18) at (1.75, 1.25) {\tiny \textit{1}};
		\node [style=none] (19) at (1.75, -0.5) {\tiny \textit{1}};
		\node [style=none] (20) at (-0.25, -0.5) {\tiny \textit{1}};
	\end{pgfonlayer}
	\begin{pgfonlayer}{edgelayer}
		\draw (0) to (4);
		\draw (0) to (3);
		\draw (0) to (5);
		\draw (1) to (3);
		\draw (1) to (4);
		\draw (1) to (5);
		\draw (2) to (5);
		\draw (2) to (3);
		\draw (2) to (4);
		\draw (0) to (1);
		\draw (1) to (2);
		\draw (3) to (4);
		\draw (4) to (5);
		\draw [bend left] (3) to (5);
		\draw [bend right] (0) to (2);
	\end{pgfonlayer}
\end{tikzpicture}
}
    }
    \caption{The $6$ nodes represent the $6$ coordinate indices of the $6$-cube. A type of square is characterized by $2$ coordinate indices and corresponds therefore to an edge on this figure. There are $16$ squares of each type. It is sufficient to add a lower power of $2$ (i.e. $8$, $4$, $2$ or $1$) number of squares of each type to obtain a basis of $RM_6(1)$. The edge label is the number of squares of a given type in the basis of $RM_6(1)$ described in Equation~\ref{6_cube_direct_sum_of_polynomial_space}. Figure (a) corresponds to the top right entry of Table~\ref{tab:subspaces_by_type_QRM_6_1_1}. Figure (b) corresponds to the top right entry of Table~\ref{tab:subspaces_by_type_QRM_6_1_1} and the first diagonal beneath it. Figure (c) corresponds to the top right entry of Table~\ref{tab:subspaces_by_type_QRM_6_1_1} and the two first diagonals beneath it. Figure (d) corresponds to the top right entry of Table~\ref{tab:subspaces_by_type_QRM_6_1_1} and the three first diagonals beneath it. Figure (e) corresponds to all entries of Table~\ref{tab:subspaces_by_type_QRM_6_1_1}.}
    \label{fig:tikz_grid}
\end{figure}

\section{Logic in $QRM_6(1,2)$}

This section is entirely about $QRM_6(1,2)$. Statements about stabilizers, Pauli logical operators and logical action of other operators therefore refer to $QRM_6(1,2)$ and we don't always specify it explicitely.

Recall from \cite{barg2025geometric}, Section~5.2 that the following sets are generating sets for stabilizer groups and logical operator groups of $QRM_6(1,2)$:
\begin{itemize}
\item X stabilizers are generated by the set of $5$-subcubes of the $6$-cube:
$$\mathcal{S}_X = \langle X(Q) \, | \, Q \text{ is a $5$-subcube of the $6$-cube} \rangle.$$
\item Z stabilizers are generated by the set of $3$-subcubes of the $6$-cube:
$$\mathcal{S}_Z = \langle Z(Q) \, | \, Q \text{ is a $3$-subcube of the $6$-cube} \rangle.$$ 
\item X logical operators are generated by the set of $4$-subcubes of the $6$-cube:
$$\mathcal{L}_X = \langle X(Q) \, | \, Q \text{ is a $4$-subcube of the $6$-cube} \rangle.$$ 
\item Z logical operators are generated by the set of $2$-subcubes of the $6$-cube:
$$\mathcal{L}_Z = \langle Z(Q) \, | \, Q \text{ is a $2$-subcube of the $6$-cube} \rangle.$$ 
\end{itemize}

\begin{definition}[\cite{barg2025geometric}, Equations (74) and (75)]
To a one qubit unitary $U$ and a subcube $Q$ of the 6-cube, we associate two operators:
\begin{itemize}
\item An operator that acts on every qubit of $Q$ by $U$. By abuse of terminology, we refer to this operator as $U$ on the subcube $Q$.
\item An operator that acts as $U$ on qubits of $Q$ that have an even Hamming weight in the 6-cube and as $U^{\dag}$ on qubits of $Q$ that have an odd Hamming weight in the 6-cube. We refer to this operator as $\widetilde{U}$ on the subcube $Q$.
\end{itemize}
\end{definition}

\begin{theorem}[\cite{barg2025geometric}, Theorem~6.2, B.4 and B.5]
The logical action on $QRM_6(1,2)$ of $\widetilde{S}$, $S$, $\widetilde{T}$ and $T$ operators on subcubes of the $6$-cube depends on the dimension of the subcube as follows: 
\begin{itemize}
\item $\widetilde{S}$ and $S$ operators on subcubes of dimension $1$, $2$ or $3$ don't preserve the codespace.
\item $\widetilde{S}$ and $S$ operators on subcubes of dimension $4$ have a nontrivial logical action on the codespace.
\item $\widetilde{S}$ and $S$ operators on subcubes of dimension $5$ or $6$ and more have a trivial logical action on the codespace.
\item $\widetilde{T}$ and $T$ operators on subcubes of dimension $1$, $2$, $3$, $4$ or $5$ don't preserve the codespace.
\item $\widetilde{T}$ and $T$ operators on subcubes of dimension $6$ have a nontrivial logical action on the codespace.
\end{itemize}
\label{thm:validity}
\end{theorem}

\begin{proof}
These results are special cases of Theorem~6.2 of \cite{barg2025geometric}. For the reader's convenience, we reproduce their proof in this special case for $\widetilde{S}$ and $\widetilde{T}$ operators. \\

The proof is by induction on the levels of the Clifford hierarchy, by investigating the action by conjugation of $\widetilde{S}$ and $\widetilde{T}$ operators on stabilizers and logical operators. We therefore first prove the statements about $\widetilde{S}$ operators and then use these results to prove the statements about $\widetilde{T}$ operators. Since $\widetilde{S}$ and $\widetilde{T}$ operators commute with $Z$ stabilizers and Z logical operators, we focus on their action by conjugation on $X$ stabilizers and logical operators. \\

\underline{Proof of statements about $\widetilde{S}$ operators:} \\

Let $A$ be a subcube of the $6$-cube of dimension at least $1$. Let $B$ be a subcube of the $6$-cube.

Direct calculations show that
$$S X S^{\dagger} =-i Z X$$
and
$$S^{\dagger} X S = i Z X.$$
Therefore
$$\widetilde{S_A} X_B \widetilde{S_A}^{\dagger} = Z_{A \cap B} X_B.$$
Indeed, the phases $i$ and $-i$ cancel out since there are as many vertices of odd and even Hamming weights in $A$.

\begin{itemize}
\item If $\dim(A) \leq 3$, there exists a $5$-subcube $B_0$ such that $A \cap B_0$ is a $2$-subcube. Since $Z_{A \cap B_0}$ is not a stabilizer of $QRM_6(1,2)$, neither is $Z_{A \cap B_0} X_{B_0}$ and thus $\widetilde{S_A}$ doesn't preserve the codespace.
\item If $\dim(A) = 4$, for every $5$-subcube $B$, $\dim(A \cap B) \geq 3$. Therefore, $Z_{A \cap B} X_B$ is a stabilizer of $QRM_6(1,2)$ for every $5$-subcube $B$ and thus $\widetilde{S_A}$ preserves the codespace.

However, there exists a $4$-subcube $B_0$ such that $A \cap B_0$ is a $2$-subcube. Since $Z_{A \cap B_0}$ is a nontrivial logical operator, $Z_{A \cap B_0} X_{B_0}$ is not logically equivalent to $X_{B_0}$ and thus $\widetilde{S_A}$ has a nontrivial logical action on the codespace.
\item If $\dim(A) \geq 5$, for every $4$-subcube $B$, $\dim(A \cap B) \geq 3$. Therefore, $Z_{A \cap B}$ is a stabilizer of $QRM_6(1,2)$ and $Z_{A \cap B} X_B$ is logically equivalent to $X_B$ for every $5$-subcube $B$. Thus $\widetilde{S_A}$ preserves the codespace and acts on it as a logical identity.
\end{itemize}

\underline{Proof of statements about $\widetilde{T}$ operators:} \\

Let $A$ be a subcube of the $6$-cube of dimension at least $1$. Let $B$ be a subcube of the $6$-cube.

Direct calculations show that
$$T X T^{\dagger} = - e^{i\frac{\pi}{4}} S X$$
and
$$T^{\dagger} X T = e^{i\frac{\pi}{4}} S^{\dagger} X.$$
Therefore
$$\widetilde{T_A} X_B \widetilde{T_A}^{\dagger} = \widetilde{S_{A \cap B}} X_B.$$
Indeed, the phases $e^{i\frac{\pi}{4}}$ and $e^{-i\frac{\pi}{4}}$ cancel out since there are as many vertices of odd and even Hamming weights in $A$.

\begin{itemize}
\item If $\dim(A) \leq 5$, there exists a $5$-subcube $B_0$ such that $A \cap B_0$ is a $4$-subcube. Since $S_{A \cap B_0}$ has a nontrivial logical action on the codespace, so does $S_{A \cap B_0} X_{B_0}$ and thus $\widetilde{T_A}$ doesn't preserve the codespace.

\item If $\dim(A) = 6$, for every $5$-subcube $B$, $A \cap B = B$ and thus $\dim(A \cap B) = 5$. Therefore, $\widetilde{S_{A \cap B}}$ stabilizes the codespace and $\widetilde{S_{A \cap B}} X_B$ is logically equivalent to $X_B$ for every $5$-subcube $B$. Thus $\widetilde{T_A}$ preserves the codespace.

However, given any (one of them is sufficient for our result) $4$-subcube $B_0$,  $A \cap B_0 = B_0$. Since $\widetilde{S_{B_0}}$ has a nontrivial logical action on the codespace, $\widetilde{S_{B_0}} X_{B_0}$ is not logically equivalent to $X_{B_0}$ and thus $\widetilde{T_A}$ has a nontrivial logical action on the codespace.
\end{itemize}

We refer to Appendix~B.1 of \cite{barg2025geometric} for statements about $S$ and $T$ operators. Note that results of Appendix~B.1 of \cite{barg2025geometric} apply for $S$ operators since $S_A X_B S_A^{\dagger}$ is phase-free when $\dim(A \cap B) \geq 2$. They also apply for $T$ operators since $T_A X_B T_A^{\dagger}$ is phase-free when $\dim(A \cap B) \geq 3$. These conditions are satisfied for X logicals and stabilizers and $S$ and $T$ operators that preserve the codespace.
\end{proof}

Since $r=q+1=2$ in $QRM_6(1,2)$, logical qubits are indexed by subsets of $\{1, \cdots, 6\}$ of cardinal $2$ (see \cite{barg2025geometric}, Section 5.2). Given such a subset $I$, any $2$-subcube of type $I$ represents the Z logical operator $Z_I$ and any $4$-subcube of type $\overline{I}$ (the complementary of $I$ in $\{1, \cdots, 6\}$) represents the logical operator $X_I$. This case ($r=q+1$) is simpler than the general case described in \cite{barg2025geometric} where some care must be taken to ensure that the basis of X and Z logical operators is symplectic (see Lemma~5.13 from \cite{barg2025geometric}). Beware that $X$ logical operators are defined by operators acting nontrivially on subcubes whose type is the complement of the type that indexes this logical qubit: $X_{\langle \overline{I} \rangle}$ is a representant of the X logical operator of logical qubit $I$ and is therefore also denoted $X_I$ (without brackets). For Z logical operators, the situation is more convenient: $Z_{\langle I \rangle}$ is a representant of the Z logical operator of logical qubit $I$ and is therefore also denoted $Z_I$. \\

The rest of this section concerns the nontrivial logical actions stated in Theorem~\ref{thm:validity}: we characterize the logical action of $\tilde{S}$ and $S$ operators on subcubes of dimension 4 and of $\tilde{T}$ and $T$ operators on subcubes of dimension 6. We begin with Lemma~\ref{lem:logic_S_4_subcube} about the logical action of $\tilde{S}$ and $S$ gates on subcubes of dimension 4. Lemma~\ref{lem:logic_S_4_subcube} is then used in the proof of Theorem~\ref{thm:logic_T_6_subcube} about the logical action of $\tilde{T}$ and $T$ gates on the cube of dimension 6.

\begin{lemma}[special case of Theorem~7.10 from \cite{barg2025geometric}]
Let $A$ be a $4$-subcube of the $6$-cube. $\widetilde{S_A}$ (and $S_A$) act on the codespace like the product of every $CZ$ gates that acts on a pair of qubits that partitions the type $K_A$ of $A$.
\label{lem:logic_S_4_subcube}
\end{lemma}

For instance, for $K_A = \{1, 2, 3, 4\}$, $\widetilde{S_A}$ is logically equivalent to the product of the three $CZ$ gates that respectively act on pairs of qubits $\{ \{1, 2\}, \{3, 4\} \}$, $\{ \{1, 3\}, \{2, 4\} \}$ and $\{ \{1, 4\}, \{2, 3\} \}$. For the reader's convenience, we sketch the proof of Theorem~7.10 from \cite{barg2025geometric} for this special case. \\

\begin{proof}
The proof consists in computing the action by conjugation of $\widetilde{S_A}$ on logical operators and recognizing the action of the aforementionned product of $CZ$ gates.

Without loss of generality, assume that $A = \langle \{1, 2, 3, 4\} \rangle$ (i.e. $K_A = \{1, 2, 3, 4\}$). Let $B$ be a $4$-subcube of type $K_B$.
$$\widetilde{S_A} X_B \widetilde{S_A}^{\dagger} = Z_{A \cap B} X_B.$$

\begin{itemize}
\item if $K_A \cup K_B = \{1, \cdots, 6\}$, $A \cap B$ is a $2$-subcube of type $K_A \cap K_B$. Therefore, $Z_{A \cap B}$ is a nontrivial logical operator. More precisely, $Z_{A \cap B}$ is $Z_{K_A \cap K_B}$: the Z logical operator of the logical qubit $K_A \cap K_B$.
\item if $K_A \cup K_B$ is a strict subset of  $\{1, \cdots, 6\}$, $A \cap B$ is a subcube of dimension at least $3$. Therefore $Z_{A \cap B}$ is a stabilizer and thus $Z_{A \cap B} X_B$ is logically equivalent to $X_B$.
\end{itemize}

Putting things together, we see that for every pair of qubits $I$ and $J$ ($I$ and $J$ are subsets of $\{1, \cdots, 6\}$ of cardinal $2$) such that $I \cup J = K_A$,
$$\widetilde{S_A} X_I \widetilde{S_A}^{\dagger} \equiv Z_J X_I.$$
and for every qubit $I$ such that $I \not\subset K_A$,
$$\widetilde{S_A} X_I \widetilde{S_A}^{\dagger} \equiv X_I,$$
where $\equiv$ denotes the same logical action on the codespace.

We recognize (see Section~4.4 from \cite{barg2025geometric} for a formal derivation) the action by conjugation of the product of $CZ$ gates over all pairs of logical qubits $\{I, J\}$ such that $I \cup J = K_A$. Therefore
$$\widetilde{S_A} \equiv CZ_{\{ \{1, 2\}, \{3, 4\} \}} CZ_{\{ \{1, 3\}, \{2, 4\} \}} CZ_{\{ \{1, 4\}, \{2, 3\} \}}.$$

The action by conjugation of $S_A$ is identical to the one of $\widetilde{S_A}$ since $\dim(A \cap B) \geq 2$ in all the cases considered in this proof. 
\end{proof}

With Lemma~\ref{lem:logic_S_4_subcube} at hand, we are now ready to prove Theorem~\ref{thm:logic_T_6_subcube}, which describes the nontrivial logical action of $T_{\langle \{1, 2, 3, 4, 5, 6\} \rangle}$ and $\widetilde{T}_{\langle \{1, 2, 3, 4, 5, 6\} \rangle}$.

\begin{theorem}[special case of Theorem~7.10 from \cite{barg2025geometric}]
Let $A$ denote the 6-cube (in itself). The logical action of $\widetilde{T_A}$ (and of $T_A$) is given by the product of every $CCZ$ gates that acts on a triple of logical qubits that partitions $\{1, 2, 3, 4, 5, 6\}$. There are $15$ such partitions, as shown in Figure~\ref{fig:CCZ_circuit_15_logical_qubits}: $\{ \{1, 2\}, \{3, 4\}, \{5, 6\} \}$, $\{ \{1, 2\}, \{3, 5\}, \{4, 6\} \}$, $\{ \{1, 2\}, \{3, 6\}, \{4, 5\} \}$, $\{ \{1, 3\}, \{2, 4\}, \{5, 6\} \}$, $\{ \{1, 3\}, \{2, 5\}, \{4, 6\} \}$, $\{ \{1, 3\}, \{2, 6\}, \{4, 5\} \}$, $\{ \{1, 4\}, \{2, 3\}, \{5, 6\} \}$, $\{ \{1, 4\}, \{2, 5\}, \{3, 6\} \}$, $\{ \{1, 4\}, \{2, 6\}, \{3, 5\} \}$, $\{ \{1, 5\}, \{2, 3\}, \{4, 6\} \}$, $\{ \{1, 5\}, \{2, 4\}, \{3, 6\} \}$, $\{ \{1, 5\}, \{2, 6\}, \{3, 4\} \}$, $\{ \{1, 6\}, \{2, 3\}, \{4, 5\} \}$, $\{ \{1, 6\}, \{2, 4\}, \{3, 5\} \}$ and $\{ \{1, 6\}, \{2, 5\}, \{3, 4\} \}$:
\begin{align*}
\widetilde{T_A} \equiv &CCZ_{\{ \{1, 2\}, \{3, 4\}, \{5, 6\} \}} CCZ_{\{ \{1, 2\}, \{3, 5\}, \{4, 6\} \}} CCZ_{\{ \{1, 2\}, \{3, 6\}, \{4, 5\} \}} \\
&CCZ_{\{ \{1, 3\}, \{2, 4\}, \{5, 6\} \}} CCZ_{\{ \{1, 3\}, \{2, 5\}, \{4, 6\} \}}  CCZ_{\{ \{1, 3\}, \{2, 6\}, \{4, 5\} \}} \\
&CCZ_{\{ \{1, 4\}, \{2, 3\}, \{5, 6\} \}} CCZ_{\{ \{1, 4\}, \{2, 5\}, \{3, 6\} \}} CCZ_{\{ \{1, 4\}, \{2, 6\}, \{3, 5\} \}} \\
&CCZ_{\{ \{1, 5\}, \{2, 3\}, \{4, 6\} \}} CCZ_{\{ \{1, 5\}, \{2, 4\}, \{3, 6\} \}} CCZ_{\{ \{1, 5\}, \{2, 6\}, \{3, 4\} \}} \\
&CCZ_{\{ \{1, 6\}, \{2, 3\}, \{4, 5\} \}} CCZ_{\{ \{1, 6\}, \{2, 4\}, \{3, 5\} \}} CCZ_{\{ \{1, 6\}, \{2, 5\}, \{3, 4\} \}}.
\end{align*}
\label{thm:logic_T_6_subcube}
\end{theorem}

For the reader's convenience, we sketch the proof of Theorem~7.10 from \cite{barg2025geometric} for the special case of $\widetilde{T}_A$.

\begin{proof}
The proof consists in computing the action by conjugation of $\widetilde{T_A}$ on logical operators and recognizing the action of the aforementionned product of $CCZ$ gates. Let $B$ be a $4$-subcube of type $K_B$. The action by conjugation of $\widetilde{T_A}$ on $X_B$ is

$$\widetilde{T_A} X_B \widetilde{T_A}^{\dagger} = \widetilde{S_B} X_B$$
since $A \cap B = B$.

Since Lemma~\ref{lem:logic_S_4_subcube} states that $\widetilde{S_B}$ is logically equivalent to three $CZ$ gates on the three partitions of $K_B$ into two subsets of cardinal $2$,  we obtain
$$\widetilde{T_A} X_{K_B} \widetilde{T_A}^{\dagger} \equiv \prod_{\{I, J\} \text{ s.t. } |I|=|J|=2 \text{ and } I \cup J = K_B} CZ_{I,J} X_{K_B}.$$

We recognize (see Section~4.4 from \cite{barg2025geometric} for a formal derivation) the action by conjugation of the product of $CCZ$ gates over all triples of logical qubits $\{I, J, K\}$ such that $I \cup J \cup K = \{1, 2, 3, 4, 5, 6\}$. Therefore
\begin{align*}
\widetilde{T_A} \equiv &CCZ_{\{ \{1, 2\}, \{3, 4\}, \{5, 6\} \}} CCZ_{\{ \{1, 2\}, \{3, 5\}, \{4, 6\} \}} CCZ_{\{ \{1, 2\}, \{3, 6\}, \{4, 5\} \}} \\
&CCZ_{\{ \{1, 3\}, \{2, 4\}, \{5, 6\} \}} CCZ_{\{ \{1, 3\}, \{2, 5\}, \{4, 6\} \}}  CCZ_{\{ \{1, 3\}, \{2, 6\}, \{4, 5\} \}} \\
&CCZ_{\{ \{1, 4\}, \{2, 3\}, \{5, 6\} \}} CCZ_{\{ \{1, 4\}, \{2, 5\}, \{3, 6\} \}} CCZ_{\{ \{1, 4\}, \{2, 6\}, \{3, 5\} \}} \\
&CCZ_{\{ \{1, 5\}, \{2, 3\}, \{4, 6\} \}} CCZ_{\{ \{1, 5\}, \{2, 4\}, \{3, 6\} \}} CCZ_{\{ \{1, 5\}, \{2, 6\}, \{3, 4\} \}} \\
&CCZ_{\{ \{1, 6\}, \{2, 3\}, \{4, 5\} \}} CCZ_{\{ \{1, 6\}, \{2, 4\}, \{3, 5\} \}} CCZ_{\{ \{1, 6\}, \{2, 5\}, \{3, 4\} \}}.
\end{align*}

The action by conjugation of $T_A$ is identical to the one of $\widetilde{T_A}$ since $\dim(A \cap B) \geq 3$ in all the cases considered in this proof. 
\end{proof}

\begin{figure}[H]
    \centering
    \includegraphics[width=0.7\textwidth]{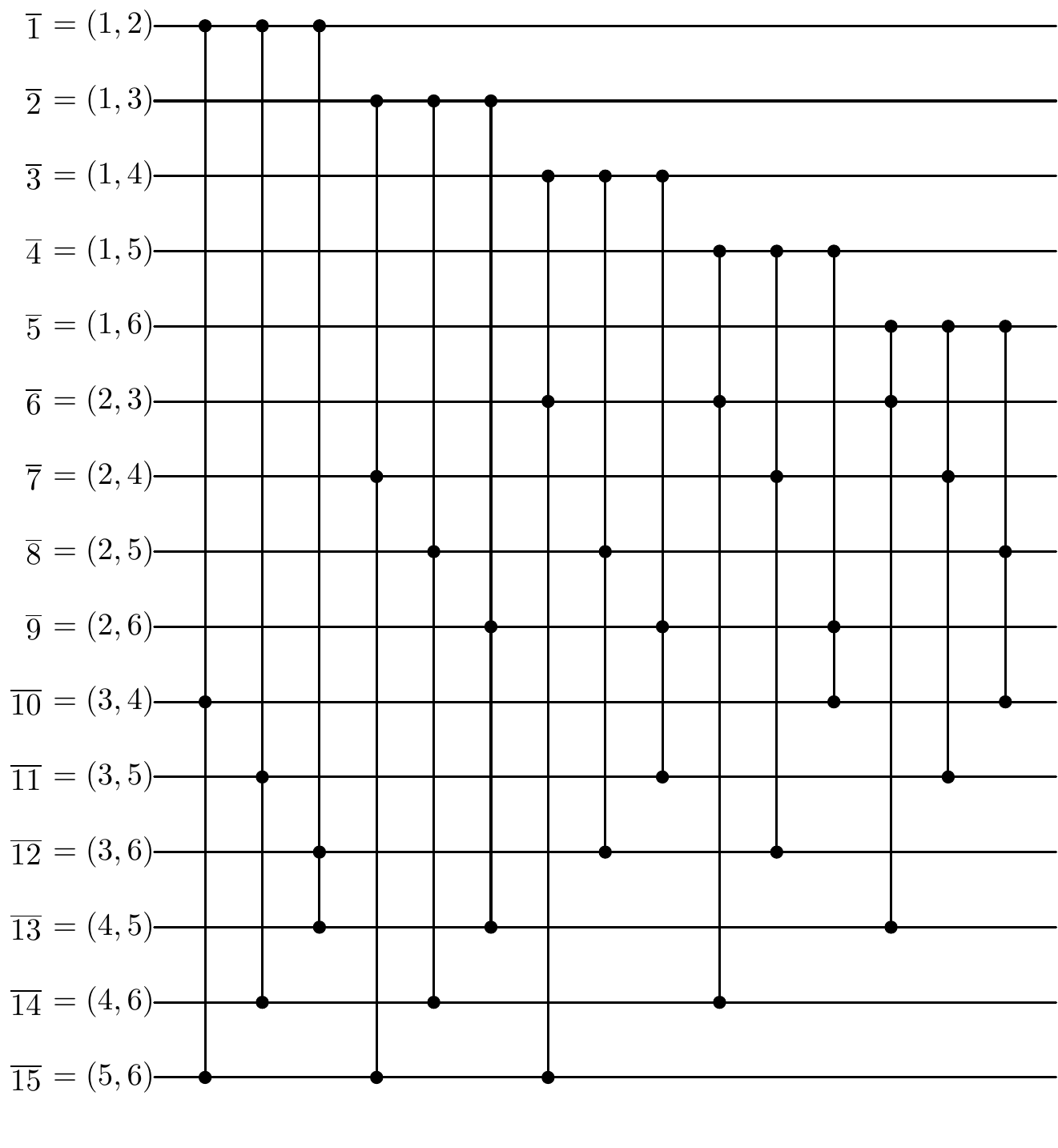}
    \caption{CCZ circuit from applying $T$ or $\widetilde{T}$ on the $15$ logical qubits of $QRM_6(1,2)$ (source: \cite{barg2025geometric})}
    \label{fig:CCZ_circuit_15_logical_qubits}
\end{figure}

\section{Planar layout for the Z stabilizer group generators of the big unfolded code}
\label{sec:big_unfolded_code}

$Z$ stabilizers are generated by cubes of dimension $r+1=3$ in $QRM_6(1,2)$ and therefore these generators have weight $2^3=8$. Finding a planar layout for weight $8$ stabilizers seems unlikely. In $QRM_6(1,1)$, $Z$ stabilizers are generated by cubes of dimension $r+1=2$ (squares) and therefore have weight $4$. But $QRM_6(1,1)$ encodes $0$ logical qubit. \\

In this section, we describe an interpolation between $QRM_6(1,2)$ and $QRM_6(1,1)$ that has $Z$ stabilizer generators of weight $4$ (a subset of the generators of $QRM_6(1,1)$, depicted in Figure~\ref{fig:layout_QRM_6_1_1}) and that encodes $3$ of the $15$ logical qubits of $QRM_6(1,2)$. We call this code the big unfolded code to insist on its similarity with the (small) unfolded code from \cite{ruiz2025unfolded}. Note that $QRM_6(1,2)$ and $QRM_6(1,1)$ have the same X stabilizer groups (and therefore so does the big unfolded code). The $15$ logical qubits of $QRM_6(1,2)$ correspond to the $6 \choose 2$ subsets of weight $2$ of $\{1,2,3,4,5,6\}$. The $3$ logical qubits of the big unfolded code correspond to subsets $\{1, 2\}$, $\{3,4\}$ and $\{5,6\}$ of $\{1,2,3,4,5,6\}$. 
As these $3$ subsets partition $\{1,2,3,4,5,6\}$, Theorem~\ref{thm:logic_T_6_subcube} ensures that $\widetilde{T}_{\{1,2,3,4,5,6\}}$ is logically equivalent to a $CCZ$ gate applied to the $3$ logical qubits of the big unfolded code. Therefore, the big unfolded code is a $64 \ket{T}$ to $\ket{CCZ}$ magic state factory.

To obtain 3 logical qubits (respectively corresponding to types $\{1, 2\}$, $\{3, 4\}$ and $\{5, 6\}$), we get rid of the $3$ Z stabilizer generators of $QRM_6(1,1)$ that correspond to the $3$ polynomials
\begin{align*}
                  &X_3 (X_4+1) (X_5+1) (X_6+1), \\
                  &(X_1+1) X_2 X_5 (X_6+1) \\
\text{and } &(X_1+1) (X_2+1) (X_3+1) X_4.
\end{align*}

They correspond to the $3$ red entries of Table~\ref{tab:big_unfolded_code}. In Figure~\ref{fig:layout_big_unfolded}, the omission of a Z stabilizer (a square) in the center gives rise to the logical qubit of type $\{3, 4\}$. Logical qubits of type $\{1, 2\}$ and $\{5, 6\}$ correspond to the $2$ other Z stabilizers (squares) that are in Figure~\ref{fig:layout_QRM_6_1_1} (below the grid of $7 \times 7$ squares) but are omitted in Figure~\ref{fig:layout_big_unfolded}. \\
 
\begin{table}[H]
\centering
\begin{tblr}{
colspec = |p{2cm}||p{1cm}|p{1cm}|p{1cm}|p{1cm}|p{1cm}|p{1cm}|,
cell{2}{3} = {red},
cell{4}{5} = {red},
cell{6}{7} = {red},
}
\hline
coordinate indices & 1 & 2 & 3 & 4 & 5 & 6  \\
\hline \hline
1 &    & 1 & 2 & 4 & 8 & 16 \\
\hline
2 &    &    & 1 & 2 & 4 & 8   \\
\hline
3 &    &    &    & 1 & 2 & 4   \\
\hline
4 &    &    &    &    & 1 & 2   \\
\hline
5 &    &    &    &    &    & 1   \\
\hline
6 &   &    &    &    &    &      \\
\hline
\end{tblr}
\caption{Rows and Columns are indexed by coordinate indices of the $6$-cube. The entry in $(i,j)$ is the number of squares of type $\{i, j\}$ in the generator set of the Z stabilizer group of the big unfolded code. The big unfolded code has three more logical qubits than $QRM_6(1,1)$ (which has no logical qubit) since three generators of the Z stabilizer group of $QRM_6(1,1)$ are not in the Z generator set of the big unfolded code. These three elements correspond to squares of types $\{1, 2\}$, $\{3, 4\}$ and $\{5, 6\}$. The Z stabilizer group of the big unfolded code has dimension $54$.}
\label{tab:big_unfolded_code}
\end{table}

\begin{figure}[H]
\centering
\includegraphics{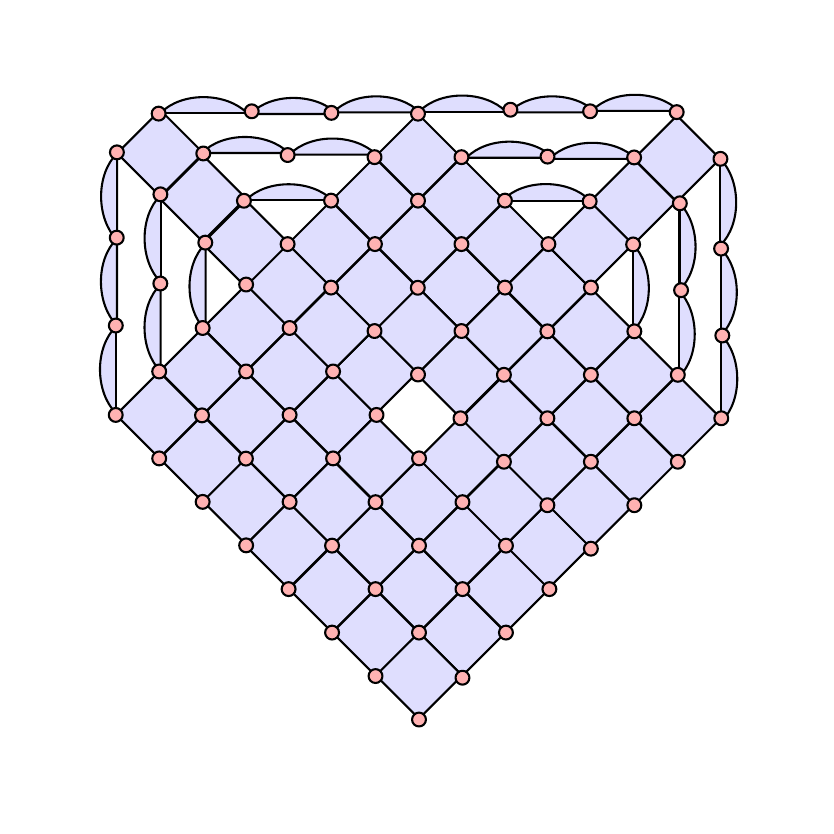}
\caption{Planar layout of the big unfolded code. The set of Z stabilizer generators of the big unfolded code is the set of Z stabilizer generators of $QRM_6(1,1)$ (depicted in Figure~\ref{fig:layout_QRM_6_1_1}) without $3$ generators: the square in the middle of the $7 \times 7$ grid of squares and the $2$ squares, respectively on the bottom-left and bottom-right sides of the $7 \times 7$ grid of squares. These $3$ generators of the Z stabilizer group of $QRM_6(1,1)$ are Z logical operators of the $3$ logical qubits of the big unfolded code.  Each vertex is a physical qubit. Each square is a Z stabilizer of weight 4. Each half-round shape is a Z stabilizer of weight 2. Half-round shapes are used to make every weight 4 stabilizer into a square. They are not described in the main text.}
\label{fig:layout_big_unfolded}
\end{figure}

\begin{theorem}
The $Z$ stabilizer group of $QRM_6(1,2)$ is contained in the $Z$ stabilizer group of the big unfolded code.
\label{thm:Z_stab_big_contains_Z_stab_6_1_2}
\end{theorem}

\begin{proof}
It is sufficient to prove that every $3$-subcube of the $6$-cube is a $Z$ stabilizer of the big unfolded code since the set of $3$-subcubes is a generating set of the $Z$ stabilizer group of the big unfolded code.

Let $Q$ be a $3$-subcube of the $6$-cube of type $I$, where $I$ is a subset of cardinal $3$ of $\{1,2,3,4,5,6\}$. The $Z$ stabilizer group of the big unfolded code is generated by every square of type $J$ for $J$ that ranges through the $12$ subsets of cardinal $2$ of $\{1,2,3,4,5,6\}$ that are not $\{1,2\}$, $\{3,4\}$ or $\{5,6\}$. Since this generating set is invariant under translations, so is the Z stabilizer group of the big unfoded code. Therefore we can assume without loss of generality that $Q$ is the standard cube of type $I$. We can leverage the symmetries of the set $\{\{1,2\}, \{3,4\}, \{5,6\}\}$ to assume, still without loss of generality, that $I = \{1,2,3\}$ or $I = \{1,3,5\}$.
\begin{itemize}
\item first case: $I = \{1,2,3\}$. \\
In this case $I = \{1,3\} \cup \{2,3\}$. Therefore the Z stabilizer corresponding to the standard $3$-subcube of type $I$ is the sum of the Z stabilizer corresponding to the standard square of type $\{1,3\}$ and the translation along the first coordinate of the standard square of type $\{2,3\}$:
$$Z_{\langle 1,2,3 \rangle} = Z_{\langle 1,3 \rangle}  + Z_{e_1 + \langle 2,3 \rangle}.$$

\item second case: $I = \{1,3,5\}$. \\
In this case $I = \{1,3\} \cup \{3,5\}$. Therefore the Z stabilizer corresponding to the standard $3$-subcube of type $I$ is the sum of the Z stabilizer corresponding to the standard square of type $\{1,3\}$ and the translation along the first coordinate of the standard square of type $\{3,5\}$:
$$Z_{\langle 1,3,5 \rangle} = Z_{\langle 1,3 \rangle}  + Z_{e_1 + \langle 3,5 \rangle}.$$
\end{itemize}
\end{proof}

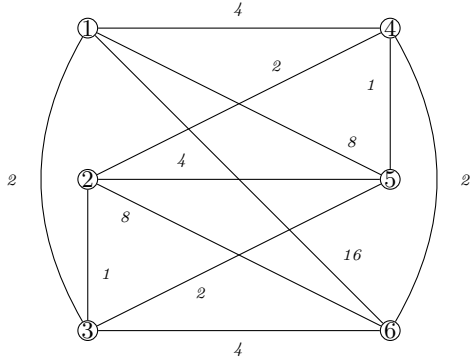
\begin{figure}[H]%
    \centering
\begin{tikzpicture}
	\begin{pgfonlayer}{nodelayer}
		\node [style=vertex set] (0) at (-2, 2) {1};
		\node [style=vertex set] (1) at (-2, 0) {2};
		\node [style=vertex set] (2) at (-2, -2) {3};
		\node [style=vertex set] (3) at (2, 2) {4};
		\node [style=vertex set] (4) at (2, 0) {5};
		\node [style=vertex set] (5) at (2, -2) {6};
		\node [style=none] (6) at (1.5, -1) {\tiny \textit{16}};
		\node [style=none] (7) at (1.5, 0.5) {\tiny \textit{8}};
		\node [style=none] (8) at (-1.5, -0.5) {\tiny \textit{8}};
		\node [style=none] (9) at (0, 2.25) {\tiny \textit{4}};
		\node [style=none] (10) at (-0.75, 0.25) {\tiny \textit{4}};
		\node [style=none] (11) at (0, -2.25) {\tiny \textit{4}};
		\node [style=none] (12) at (0.5, 1.5) {\tiny \textit{2}};
		\node [style=none] (13) at (-0.5, -1.5) {\tiny \textit{2}};
		\node [style=none] (14) at (-3, 0) {\tiny \textit{2}};
		\node [style=none] (15) at (3, 0) {\tiny \textit{2}};
		\node [style=none] (16) at (-1.75, -1.25) {\tiny \textit{1}};
		\node [style=none] (18) at (1.75, 1.25) {\tiny \textit{1}};
	\end{pgfonlayer}
	\begin{pgfonlayer}{edgelayer}
		\draw (0) to (4);
		\draw (0) to (3);
		\draw (0) to (5);
		\draw (1) to (3);
		\draw (1) to (4);
		\draw (1) to (5);
		\draw (2) to (5);
		\draw (2) to (4);
		\draw (1) to (2);
		\draw (3) to (4);
		\draw [bend left] (3) to (5);
		\draw [bend right] (0) to (2);
	\end{pgfonlayer}
\end{tikzpicture}
\caption{The $6$ nodes represent the $6$ coordinate indices of the $6$-cube. A type of square is characterized by $2$ coordinate indices and corresponds therefore to an edge on this Figure. The edge label is the number of squares of a given type in the basis of the Z stabilizer group of the big unfolded code. The above family with $54$ elements is a basis for the space generated by all squares of $12$ different types: $\{1, 3\}$, $\{1, 4\}$, $\{1, 5\}$, $\{1, 6\}$, $\{2, 3\}$, $\{2, 4\}$, $\{2, 5\}$, $\{2, 6\}$, $\{3, 5\}$, $\{3, 6\}$, $\{4, 5\}$, $\{4, 6\}$. The three missing types correspond to the red background entries in Table~\ref{tab:big_unfolded_code}. Figure~\ref{fig:layout_big_unfolded} shows a planar layout for the $54$ above Z stabilizer generators.}%
\end{figure}

\begin{theorem}
The logical action of $\widetilde{T}_{\langle \{1, 2, 3, 4, 5, 6\} \rangle}$ on the big unfolded code is given by the restriction of its action on $QRM_6(1,2)$ on the 3 logical qubits of the big unfolded code:
\begin{equation*}
\widetilde{T}_{\langle \{1, 2, 3, 4, 5, 6\} \rangle} \equiv CCZ_{\{ \{1, 2\}, \{3, 4\}, \{5, 6\} \}}.
\end{equation*}
\label{thm:logic_T_6_subcube_buc}
\end{theorem}

\begin{proof}
The big unfolded code and $QRM_6(1,2)$ have the same X stabilizer groups:
$$ \mathcal{S}^X_{\text{big unfolded code}} = \mathcal{S}^X_{QRM_6(1,2)}. $$ 
Therefore any $S$ or $T$ operator that is a logical operation of $QRM_6(1,2)$ is also a logical operation of the big unfolded code (since $S$ and $T$ operators trivially commute with $Z$ stabilizers).

Theorem~\ref{thm:Z_stab_big_contains_Z_stab_6_1_2} states that the Z stabilizer group of the big unfolded code contains the Z stabilizer group of $QRM_6(1,2)$:
$$ \mathcal{S}^Z_{\text{big unfolded code}} \supset \mathcal{S}^Z_{QRM_6(1,2)}. $$

Therefore the X logical operator group of the big unfolded code is contained in the X logical operator group of $QRM_6(1,2)$:
$$ \mathcal{L}^X_{\text{big unfolded code}}
\subset
\mathcal{L}^X_{QRM_6(1,2)}. $$

Therefore the projection $\pi$ from the quotient space that defines distinct X logical operators in $QRM_6(1,2)$ to its counterpart in the big unfolded code
$$
\mathcal{L}^X_{QRM_6(1,2)}
/
\mathcal{S}^X_{QRM_6(1,2)}
\xrightarrow{\pi}
\mathcal{L}^X_{\text{big unfolded code}}
/
\mathcal{S}^X_{\text{big unfolded code}}
$$
is given by imposing to logical qubits that exist in $QRM_6(1,2)$ but not in the big unfolded code to be in state $\ket{0}$ (since they are measured by a Z stabilizer in the big unfolded code). Therefore $CCZ$ gates that act in $QRM_6(1,2)$ on a logical qubit that does not exist in the big unfolded code are trivial in the big unfolded code. The result follows.
\end{proof}

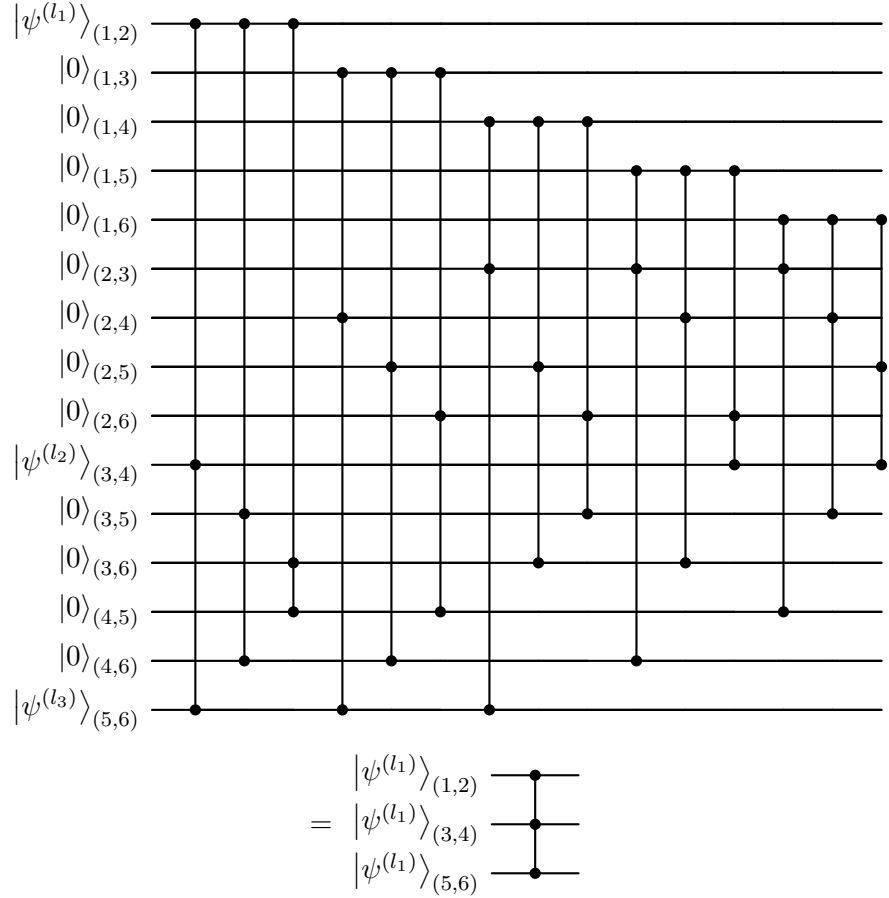
\begin{figure}[H]
\center
\begin{quantikz}
\lstick{$\ket{\psi^{(l_1)}}_{(1,2)}$}  & \ctrl{9}      & \ctrl{10}   & \ctrl{11}   & \qw          & \qw       & \qw          & \qw          & \qw         & \qw          & \qw       & \qw          & \qw          & \qw         & \qw          & \qw  \\
\lstick{$\ket{0}_{(1,3)}$}  & \qw          & \qw          & \qw         & \ctrl{5}      & \ctrl{6}  & \ctrl{7}      & \qw          & \qw         & \qw          & \qw       & \qw          & \qw          & \qw         & \qw          & \qw  \\
\lstick{$\ket{0}_{(1,4)}$}  & \qw          & \qw          & \qw         & \qw          & \qw       & \qw          & \ctrl{3}      & \ctrl{5}    & \ctrl{6}      & \qw       & \qw          & \qw          & \qw         & \qw          & \qw  \\
\lstick{$\ket{0}_{(1,5)}$}  & \qw          & \qw          & \qw         & \qw          & \qw       & \qw          & \qw          & \qw         & \qw          & \ctrl{2}   & \ctrl{3}     & \ctrl{5}      & \qw         & \qw          & \qw  \\
\lstick{$\ket{0}_{(1,6)}$}  & \qw          & \qw          & \qw         & \qw          & \qw       & \qw          & \qw          & \qw         & \qw          & \qw       & \qw          & \qw          & \ctrl{1}    & \ctrl{2}      & \ctrl{3} \\
\lstick{$\ket{0}_{(2,3)}$}  & \qw          & \qw          & \qw         & \qw          & \qw       & \qw          & \ctrl{9}      & \qw         & \qw          & \ctrl{8}  & \qw          & \qw          & \ctrl{7}      & \qw         & \qw  \\
\lstick{$\ket{0}_{(2,4)}$}  & \qw          & \qw          & \qw         & \ctrl{8}      & \qw       & \qw          & \qw          & \qw         & \qw          & \qw       & \ctrl{5}     & \qw          & \qw         & \ctrl{4}      & \qw  \\
\lstick{$\ket{0}_{(2,5)}$}  & \qw          & \qw          & \qw         & \qw          & \ctrl{6}   & \qw          & \qw          & \ctrl{4}    & \qw          & \qw       & \qw          & \qw          & \qw         & \qw          & \ctrl{2} \\
\lstick{$\ket{0}_{(2,6)}$}  & \qw          & \qw          & \qw         & \qw          & \qw       & \ctrl{4}     & \qw          & \qw         & \ctrl{2}      & \qw       & \qw          & \ctrl{1}     & \qw         & \qw          & \qw  \\
\lstick{$\ket{\psi^{(l_2)}}_{(3,4)}$}  & \ctrl{5}     & \qw          & \qw         & \qw          & \qw       & \qw          & \qw          & \qw         & \qw          & \qw       & \qw          & \control{} & \qw         & \qw          & \control{} \\
\lstick{$\ket{0}_{(3,5)}$}  & \qw          & \ctrl{3}     & \qw         & \qw          & \qw       & \qw          & \qw          & \qw         &\control{}  & \qw       & \qw          & \qw          & \qw         & \control{} & \qw  \\
\lstick{$\ket{0}_{(3,6)}$}  & \qw          & \qw          & \ctrl{1}    & \qw          & \qw       & \qw          & \qw          & \control{} & \qw          & \qw      &\control{}   & \qw          & \qw         & \qw          & \qw  \\
\lstick{$\ket{0}_{(4,5)}$}  & \qw          & \qw          & \control{}& \qw          & \qw       & \control{} & \qw          & \qw         & \qw          & \qw       & \qw          & \qw          & \control{} & \qw          & \qw  \\
\lstick{$\ket{0}_{(4,6)}$}  & \qw          & \control{} & \qw         & \qw          &\control{}& \qw         & \qw          & \qw         & \qw          &\control{}& \qw          & \qw          & \qw         & \qw          & \qw  \\
\lstick{$\ket{\psi^{(l_3)}}_{(5,6)}$}  & \control{} & \qw          & \qw         & \control{} & \qw       & \qw          & \control{} & \qw         & \qw          & \qw       & \qw          & \qw          & \qw         & \qw          & \qw 
\end{quantikz} \\
=
\begin{quantikz}
\lstick{$\ket{\psi^{(l_1)}}_{(1,2)}$} & \ctrl{1}      & \\
\lstick{$\ket{\psi^{(l_1)}}_{(3,4)}$} & \ctrl{1}      & \\
\lstick{$\ket{\psi^{(l_1)}}_{(5,6)}$} & \control{} &
\end{quantikz}
\caption{The big unfolded code has 3 of the 15 logical qubits of $QRM_6(1,2)$. The logical action of a transversal T gate on the big unfolded code is the same as on $QRM_6(1,2)$ if the 12 additional qubits of $QRM_6(1,2)$ are in the state $\ket{0}$. It is therefore a CCZ gate on the 3 logical qubits of the big unfolded code.}
\end{figure}

\section{3D layout for the Z stabilizer group generators of $QRM_6(1,2)$: the rubik's cube layout}
\label{sec:QRM_6_1_2}

We associate 2 coordinates to each of the 3 spatial directions: $x$ corresponds to $\{1, 2\}$, $y$ corresponds to $\{3, 4\}$ and $z$ corresponds to $\{5, 6\}$. The $Z$ stabilizer group of $QRM_6(1,2)$ has dimension $${6 \choose 6} + {6 \choose 5} + {6 \choose 4} + {6 \choose 3} = 1 + 6 + 15 + 20 = 42.$$ It can be generated by cubes (of dimension 3). Cubes of type $\{i, j, k\}$ with $i \in \{1, 2\}$, $j \in \{3, 4\}$ and $k \in \{5, 6\}$ are local cubes in 3D. Bases of the space generated by edges in a  square are given by:
\begin{itemize}
\item 2 edges of type $\{1\}$ and 1 edge of type $\{2\}$ in the square $\{1, 2\}$,
\item 2 edges of type $\{3\}$ and 1 edge of type $\{4\}$ in the square $\{3, 4\}$,
\item 2 edges of type $\{5\}$ and 1 edge of type $\{6\}$ in the square $\{5, 6\}$.
\end{itemize}

Taking the cartesian product of the 3 above bases (each basis has 3 elements) yields 27 cubes with the following types:
\begin{itemize}
\item 8 cubes of type $\{1, 3, 5\}$,
\item 4 cubes of type $\{2, 3, 5\}$,
\item 4 cubes of type $\{1, 4, 5\}$,
\item 4 cubes of type $\{1, 3, 6\}$,
\item 2 cubes of type $\{2, 4, 5\}$,
\item 2 cubes of type $\{2, 3, 6\}$,
\item 2 cubes of type $\{1, 4, 6\}$,
\item 1 cube of type $\{2, 4, 6\}$.
\end{itemize}

Lemma~\ref{lem:rubik_bulk} below states that these 27 cubes generate the space generated by all cubes of the 8 above types. When arranged into the 3 spatial directions x, y and z, they correspond exactly to the 27 cubes of a rubik's cube (see Figure~\ref{fig:3D_layout_6_cube}). \\

More formally, we define $T_{bulk}$ as the set of types corresponding to $\{1,2\} \times \{3,4\} \times \{5,6\}$:
\begin{align*}
T_{bulk} &= \{1,2\} \times \{3,4\} \times \{5,6\} \\
T_{bulk} &= \{ \{1,3,5\}, \{2,3,5\}, \{1,4,5\}, \{1,3,6\}, \{2,4,5\}, \{2,3,6\}, \{1,4,6\}, \{2,4,6\} \}.
\end{align*}

$$|T_{bulk}| = 8.$$

\begin{itemize}
\item $B_{edges \, of \, the \, square \{1,2\}}$ is made of the two edges of type $\{1\}$ (corresponding respectively to polynomials $X_2$ and $X_2+1$ in $\mathbb{F}_2[X_1, X_2]$) and one edge of type $\{2\}$ (corresponding to polynomial $X_1$ in $\mathbb{F}_2[X_1, X_2]$).
$$|B_{edges \, of \, the \, square \{1,2\}}| = 3.$$

\item $B_{edges \, of \, the \, square \{3,4\}}$ is made of the two edges of type $\{3\}$ (corresponding respectively to polynomials $X_4$ and $X_4+1$ in $\mathbb{F}_2[X_3, X_4]$) and one edge of type $\{4\}$ (corresponding to polynomial $X_3$ in $\mathbb{F}_2[X_3, X_4]$).
$$|B_{edges \, of \, the \, square \{3,4\}}| = 3.$$

\item $B_{edges \, of \, the \, square \{5,6\}}$ is made of the two edges of type $\{5\}$ (corresponding respectively to polynomials $X_6$ and $X_6+1$ in $\mathbb{F}_2[X_5, X_6]$) and one edge of type $\{6\}$ (corresponding to polynomial $X_5$ in $\mathbb{F}_2[X_5, X_6]$).
$$|B_{edges \, of \, the \, square \{5,6\}}| = 3.$$
\end{itemize}

$$B_{bulk} = B_{edges \, of \, the \, square \{1,2\}} \times B_{edges \, of \, the \, square \{3,4\}} \times B_{edges \, of \, the \, square \{5,6\}}.$$

$$|B_{bulk}| = 3*3*3 = 27.$$

\begin{lemma}
Let $B_a$ be a basis for subcubes of types $T_a$ in the cube $\mathbb{F}_2^K$ and $B_b$ be a basis for subcubes of types $T_b$ in the cube $\mathbb{F}_2^L$ where $K$ and $L$ are disjoint coordinate sets. Then $B_a \times B_b$ is a basis for subcubes of types $T_{a \times b}$ in the cube $\mathbb{F}_2^{K \cup L}$, where
$$ T_{a \times b} = \{ I \cup J \, | \, I \in T_a \text{ and } J \in T_b\} .$$
\label{lem:bases_product}
\end{lemma}

\begin{proof}
Let $I$ be a type in the cube $\mathbb{F}_2^K$ (i.e. a subset of $K$) and $J$ be a type in the cube $\mathbb{F}_2^K$ (i.e. a subset of $L$). Let $S_I$ be a subcube of type $I$ and $S_J$  be a subcube of type $J$. There exist elements $C_{a, 1}, \dots, C_{a, m_a}$ of $B_a$ and $C_{b, 1}, \dots, C_{b, m_b}$ of $B_b$ such that
\begin{align*}
S_I &= \sum_{i=1}^{m_a} C_{a, i} \\
S_J &= \sum_{j=1}^{m_b} C_{b, j}.
\end{align*} 
Therefore,
$$ S_I \times S_J =  \sum_{i=1}^{m_a} \sum_{j=1}^{m_b} C_{a, i} \times C_{b, j}. $$

This proves that $B_a \times B_b$ generates $T_{a \times b}$.

Assume now for contradiction that $B_a \times B_b$ is not free. Then there exists a non empty, trivial, linear combination of elements of $B_a \times B_b$. We can group these elements by their second factor:
$$ \sum_{j=1}^{m_b} (\sum_{i=1}^{m_{a,j}} C_{a, i, j}) \times C_{b, j} = 0.$$
Since at least one $\sum_{i=1}^{m_{a,j}} C_{a, i, j}$ is non empty, there exists $j_0$ such that $m_{a,j_0} \geq 1$ and
$$ \sum_{i=1}^{m_{a,j_0} }C_{a, i, j_0} = 0. $$
This contradicts the freedom of $B_a$ and therefore $B_a \times B_b$ is free.
\end{proof}

\begin{lemma}
The basis $B_{bulk}$ generates the same space as the set of cubes whose type belongs to $T_{bulk}$:
$$\Span(B_{bulk})
=
\Span(\text{Subcubes}(T_{bulk}).$$ 
\label{lem:rubik_bulk}
\end{lemma}

\begin{proof}
Applying Lemma~\ref{lem:bases_product} twice to $T_{a} = \{ \{1\}, \{2\} \}$, $T_{b} = \{ \{3\}, \{4\} \}$ and $T_{c} = \{ \{5\}, \{6\} \}$, gives the result.
\end{proof}

Lemma~\ref{lem:subspace_two_types} implies that cubes of type $\{i, j, k\}$ give translation in the $k$ coordinate to cubes of type $\{i, j, l\}$.

For instance a cube of type $\{1, 2, 3\}$ can be be translated along the fifth and sixth coordinates with cubes of type $\{i,j,k\}$ with $i \in \{1,2\}$, $j \in \{3,4\}$ and $k \in \{5,6\}$. To obtain translation along the 4th coordinates, we add a second cube of type $\{1, 2, 3\}$, image of the first by the translation along the fourth coordinate.

To the 27 cubes of type $\{i,j,k\}$ with $i \in \{1,2\}$, $j \in \{3,4\}$ and $k \in \{5,6\}$, we add the 15 following cubes to obtain a basis of the Z stabilizer group of $QRM_6(1,2)$:
\begin{itemize}
\item 2 cubes of type $\{1, 2, 3\}$ (translation of each other along the $4^{th}$ coordinate),
\item 1 cube of type $\{1, 2, 4\}$,
\item 1 cube of type $\{1, 3, 4\}$,
\item 1 cube of type $\{2, 3, 4\}$,
\item 2 cubes of type $\{3, 4, 5\}$ (translation of each other along the $6^{th}$ coordinate),
\item 1 cube of type $\{3, 4, 6\}$,
\item 1 cube of type $\{3, 5, 6\}$,
\item 1 cube of type $\{4, 5, 6\}$,
\item 2 cubes of type $\{1, 5, 6\}$ (translation of each other along the $2^{nd}$ coordinate),
\item 1 cube of type $\{2, 5, 6\}$,
\item 1 cube of type $\{1, 2, 5\}$,
\item 1 cube of type $\{1, 2, 6\}$.
\end{itemize}

More formally, we define:
\begin{align*}
T_{bottom} &= \{ \{1,2,3\}, \{1,2,4\} \} \\
T_{top} &= \{ \{1,3,4\}, \{2,3,4\} \} \\
T_{left} &= \{ \{3,4,5\}, \{3,4,6\} \} \\
T_{right} &= \{ \{3,5,6\}, \{4,5,6\} \} \\
T_{back} &= \{ \{1,5,6\}, \{2,5,6\} \} \\
T_{front} &= \{ \{1,2,5\}, \{1,2,6\} \}. \\
\end{align*}

Even though the following bases are trivial (they each have cardinal 1), we will need them later to define bases as products of other bases.
\begin{itemize}
\item $B_{square \, of \, the \, square \, \{1,2\}}$ is made of the entire square $\{1, 2\}$ in itself (corresponding to polynomial $1$ in $\mathbb{F}_2[X_1, X_2]$).
$$|B_{square \, of \, the \, square \{1,2\}}| = 1.$$

\item $B_{square \, of \, the \, square \, \{3,4\}}$ is made of the entire square $\{3, 4\}$ in itself (corresponding to polynomial $1$ in $\mathbb{F}_2[X_3, X_4]$).
$$|B_{square \, of \, the \, square \{3,4\}}| = 1.$$

\item $B_{square \, of \, the \, square \, \{5,6\}}$ is made of the entire square $\{5, 6\}$ in itself (corresponding to polynomial $1$ in $\mathbb{F}_2[X_5, X_6]$).
$$|B_{square \, of \, the \, square \, \{5,6\}}| = 1.$$
\end{itemize}

We are now ready to define $B_{bottom}$, $B_{top}$, $B_{left}$, $B_{right}$, $B_{back}$, $B_{front}$. 

\begin{equation*}
B_{bottom} = B_{square \, of \, the \, square \, \{1,2\}} \times B_{edges \, of \, the \, square \, \{3,4\}}
\end{equation*}
\begin{equation*}
|B_{bottom}| = 1 * 3 = 3.
\end{equation*}

\begin{lemma}
The concatenation of bases $(B_{bulk}, B_{bottom})$ generates the same space as the set of cubes whose type belongs to $T_{bulk}$ or to $T_{bottom}$:
$$\Span(B_{bulk}, B_{bottom})
=
\Span(\text{Subcubes}(T_{bulk} \cup T_{bottom}).$$ 
\label{lem:rubik_bulk_bottom}
\end{lemma}

\begin{proof}
Lemma~\ref{lem:rubik_bulk} states that $\Span(B_{bulk}) = \Span(\text{Subcubes}(T_{bulk}))$. Lemma~\ref{lem:bases_product} implies that $\Span(B_{bottom}) = \Span(\text{Subcubes}(T_{bottom}))$ in $\mathbb{F}_2^{\{1,2,3,4\}}$, the 4-cube with coordinates 1, 2, 3 and 4.

Let $I = \{i, j, k\} \in T_{bottom}$. Without loss of generality we can assume that $i \in \{1,2\}$ and $k \in \{3,4\}$. Therefore $\{i,k,5\}$ and $\{i,k,6\}$ belong to $T_{bulk}$ and Corollary~\ref{cor:translations_from_another_type} implies that we can translate cubes of type $I$ along coordinates 5 and 6. The result follows.
\end{proof}

\begin{equation*}
B_{top} = B_{square \, of \, the \, square \, \{3,4\}} \times (\text{edge}(X_2 \in \mathbb{F}_2[X_1, X_2]), \, \text{edge}(X_1 \in \mathbb{F}_2[X_1, X_2]))
\end{equation*}
\begin{equation*}
|B_{top}| = 1 * 2 = 2.
\end{equation*}

\begin{lemma}
The concatenation of bases $(B_{bulk}, B_{bottom}, B_{top})$ generates the same space as the set of cubes whose type belongs to $T_{bulk}$, $T_{bottom}$ or $T_{top}$:
$$\Span(B_{bulk}, B_{bottom}, B_{top})
=
\Span(\text{Subcubes}(T_{bulk} \cup T_{bottom} \cup T_{top}).$$ 
\label{lem:rubik_bulk_bottom_top}
\end{lemma}

\begin{proof}
The proof is almost the same as the one of Lemma~\ref{lem:rubik_bulk_bottom}. However there are only 2 elements in $B_{top}$ whereas there were 3 elements in $B_{bottom}$. This is because cubes of type $\{1,2,4\}$ give translations along the $2^{nd}$ coordinate to cubes of type $\{1,3,4\}$ and cubes of type $\{1,2,3\}$  give translations along the $1^{st}$ coordinate to cubes of type $\{2,3,4\}$.
\end{proof}

\begin{equation*}
B_{left} = B_{square \, of \, the \, square \, \{3,4\}} \times B_{edges \, of \, the \, square \, \{5,6\}}
\end{equation*}
\begin{equation*}
|B_{left}| = 1 * 3 = 3.
\end{equation*}

\begin{equation*}
B_{right} = B_{square \, of \, the \, square \, \{5,6\}} \times (\text{edge}(X_4 \in \mathbb{F}_2[X_3, X_4]), \, \text{edge}(X_3 \in \mathbb{F}_2[X_3, X_4]))
\end{equation*}
\begin{equation*}
|B_{right}| = 1 * 2 = 2.
\end{equation*}

\begin{equation*}
B_{back} = B_{square \, of \, the \, square \, \{5,6\}} \times B_{edges \, of \, the \, square \, \{1,2\}}
\end{equation*}
\begin{equation*}
|B_{left}| = 1 * 3 = 3.
\end{equation*}

\begin{equation*}
B_{front} = B_{square \, of \, the \, square \, \{1,2\}} \times (\text{edge}(X_6 \in \mathbb{F}_2[X_5, X_6]), \, \text{edge}(X_5 \in \mathbb{F}_2[X_5, X_6]))
\end{equation*}
\begin{equation*}
|B_{front}| = 1 * 2 = 2.
\end{equation*}

\begin{theorem}
The concatenation of bases
$$B_{QRM_6(1,2)} := (B_{bulk}, B_{bottom}, B_{top}, B_{left}, B_{right}, B_{back}, B_{front})$$
is a basis for the Z stabilizer group of $QRM_6(1,2)$.
\label{thm:rubik_basis}
\end{theorem}

\begin{proof}
It is sufficient to show that the concatenation of bases $B_{QRM_6(1,2)}$ generates the same space as the set of cubes whose type belongs to $T_{bulk}$, $T_{bottom}$, $T_{top}$, $T_{left}$, $T_{right}$, $T_{back}$ or $T_{front}$. Indeed $T_{bulk}$, $T_{bottom}$, $T_{top}$, $T_{left}$, $T_{right}$, $T_{back}$ and $T_{front}$ partition the subsets of cardinal 3 of $\{1,2,3,4,5,6\}$ and $ |B_{bulk}|+|B_{bottom}|+|B_{top}|+|B_{left}|+|B_{right}|+|B_{back}|+|B_{front}|=42$, which is the dimension of the Z stabilizer group of $QRM_6(1,2)$.
The same proof as Lemma~\ref{lem:rubik_bulk_bottom} shows that adding $B_{left}$ to $B_{bulk}$ generates all cubes of type $T_{left}$ and that adding $B_{back}$ to $B_{bulk}$ generates all cubes of type $T_{back}$. The same proof as Lemma~\ref{lem:rubik_bulk_bottom_top} shows that adding $B_{right}$ to $(B_{bulk}, B_{left})$ generates all cubes of type $T_{right}$ and that adding $B_{front}$ to $(B_{bulk}, B_{back})$ generates all cubes of type $T_{front}$. 
\end{proof}

As a sanity check, we verified numerically (see source code\footnote{\url{https://github.com/vivienlonde/unfolded_Quantum_Reed_Muller/tree/main/6_cube_3D/3D_Z_stab_basis.py}}) that $B_{QRM_6(1,2)}$ is indeed a basis for the Z stabilizer group of $QRM_6(1,2)$.

\begin{figure}[H]
\centering
\includegraphics[width=0.7\textwidth]{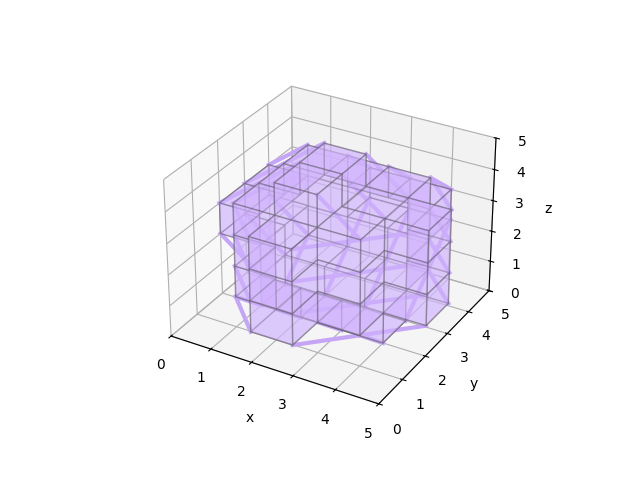}
\caption{A 3D interactive version can be generated with this script \url{https://github.com/vivienlonde/unfolded_Quantum_Reed_Muller/tree/main/6_cube/3D_layout.py}. The 42 generators of the Z stabilizer group of $QRM_6(1,2)$ correspond to 42 local cubes in 3D. Some weight 2 stabilizers are used to turn the 15 stabilizer generators of $B_{bottom}$, $B_{top}$, $B_{left}$, $B_{right}$, $B_{back}$ and $B_{front}$ into local cubes in 3D.}
\label{fig:3D_layout_6_cube}
\end{figure}

A transversal $T$ gate on the 64 physical qubits of $QRM_6(1,2)$ has the logical effect of 15 $CCZ$ gates on the 15 logical qubits of $QRM_6(1,2)$ (see Theorem~\ref{thm:logic_T_6_subcube} for the proof and Figure~\ref{fig:CCZ_circuit_15_logical_qubits} for an illustration of which triplets of logical qubits undergo a $CCZ$ gate).

\section{3D layout for the Z stabilizer group generators of $QRM_7(2,2)$}
\label{sec:QRM_7_2_2}

We partition the seven coordinates of $QRM_7(2,2)$ into 3 groups corresponding to the 3 spatial directions x, y and z: coordinates 1, 2 and 3 are unfolded in the x direction, coordinates 4, 5 and 6 are unfolded in the y direction and coordinate 7 occupies the z direction.

$Z$ stabilizers generators of $QRM_7(2,2)$ are cubes (of dimension 3). The $Z$ stabilizer group has dimension $${7 \choose 7} + {7 \choose 6} + {7 \choose 5} + {7 \choose 4} + {7 \choose 3} = 1 + 7 + 21 + 35 + 35 = 99.$$

Cubes of type $\{i, j, k\}$ with $i \in \{1, 2, 3\}$, $j \in \{4, 5, 6\}$ and $k \in \{7\}$ are local cubes in the 3D layout. There are 49 such cubes. Lemma~\ref{lem:bases_product} implies that these 49 cubes form a basis for cubes of these 9 types.

More formally, we define
\begin{align*}
&T_{main \, layer, \, center} = \{1, 2, 3\} \times \{4, 5, 6\} \times \{7\} \\
&= \{ \{1, 4, 7\}, \{1, 5, 7\}, \{1, 6, 7\}, \{2, 4, 7\}, \{2, 5, 7\}, \{2, 6, 7\}, \{3, 4, 7\}, \{3, 5, 7\}, \{3, 6, 7\} \}.
\end{align*}
$$ |T_{main layer, \, center}| = 9.$$

\begin{itemize}
\item $B_{edges \, of \, the \, cube \, \{1,2,3\}}$ is made of 4 edges of type $\{1\}$ (corresponding respectively to polynomials $(X_2+1)(X_3+1)$, $X_2 (X_3+1)$,  $X_2 X_3$ and $(X_2+1) X_3$ in $\mathbb{F}_2[X_1,X_2,X_3]$), 2 edges of type $\{2\}$ (corresponding respectively to polynomials $X_1 (X_3+1)$ and $X_1 X_3$ in $\mathbb{F}_2[X_1,X_2,X_3]$) and 1 edge of type $\{3\}$ (corresponding to polynomial $(X_1+1) X_2$ in $\mathbb{F}_2[X_1,X_2,X_3]$) (see Figure~\ref{fig:Gray_code_cube_123}).
$$|B_{edges \, of \, the \, cube \, \{1,2,3\}}| = 7.$$

\item $B_{edges \, of \, the \, cube \, \{4,5,6\}}$ is made of 4 edges of type $\{6\}$ (corresponding respectively to polynomials $(X_4+1)(X_5+1)$, $(X_4+1) X_5$,  $X_4 X_5$ and $X_4 (X_5+1)$ in $\mathbb{F}_2[X_4,X_5,X_6]$), 2 edges of type $\{5\}$ (corresponding respectively to polynomials $(X_4+1) X_6$ and $X_4 X_6$ in $\mathbb{F}_2[X_4,X_5,X_6]$) and 1 edge of type $\{4\}$ (corresponding to polynomial $X_5 (X_6+1)$ in $\mathbb{F}_2[X_4,X_5,X_6]$) (see Figure~\ref{fig:Gray_code_cube_456}).
$$|B_{edges \, of \, the \, cube \, \{4,5,6\}}| = 7.$$

\item $B_{edge \, of \, the \, edge \, \{7\}}$ is made of the entire edge $\{7\}$ itself (corresponding to polynomial $1$ in $\mathbb{F}_2[X_7]$).
$$|B_{edge \, of \, the \, edge \, \{7\}}| = 1.$$
\end{itemize}

$$ B_{main \, layer, \, center} = B_{edges \, of \, the \, cube \, \{1,2,3\}} \times B_{edges \, of \, the \, cube \, \{4,5,6\}} \times B_{edge \, of \, the \, edge \, \{7\}}. $$
$$ |B_{main \, layer, \, center}| = 7*7*1 = 49. $$

\begin{lemma}
The basis $B_{main \, layer, \, center}$ generates the same space as the set of cubes whose type belongs to $T_{main \, layer, \, center}$:
$$\Span(B_{main \, layer, \, center})
=
\Span(Subcubes(T_{main \, layer, \, center})).$$
\label{lem:main_layer_center}
\end{lemma}

\begin{proof}
Applying Lemma~\ref{lem:bases_product} twice to $T_a = \{ \{1\}, \{2\}, \{3\} \}$, $T_b = \{ \{4\}, \{5\}, \{6\} \}$ and $T_c = \{ \{7\} \}$ gives the result.
\end{proof}

\begin{align*}
T_{main \, layer, \, NE-SW} &= \{ \{1,2\}, \{1,3\}, \{2,3\} \} \times \{7\} \\
T_{main \, layer, \, NE-SW} &= \{ \{1,2,7\}, \{1,3,7\}, \{2,3,7\} \}.
\end{align*}
$$ |T_{main \, layer, \, NE-SW}| = 3.$$

$B_{squares \, of \, the \, cube \, \{1,2,3\}}$ is made of 2 squares of type $\{1, 3\}$ (corresponding to polynomials $X_2$ and $X_2+1$ in $\mathbb{F}_2[X_1,X_2,X_3]$), 1 square of type $\{1, 2\}$ (corresponding to polynomial $X_3$ in $\mathbb{F}_2[X_1,X_2,X_3]$) and 1 square of type $\{2, 3\}$ (corresponding to polynomial $X_1$ in $\mathbb{F}_2[X_1,X_2,X_3]$) (see Figure~\ref{fig:layout_QRM_6_1_1}).
$$ |B_{squares \, of \, the \, cube \, \{1,2,3\}}| = 4.$$

$$ B_{main \, layer, \, NE-SW} = B_{squares \, of \, the \, cube \, \{1,2,3\}} \times B_{edge \, of \, the \, edge \, \{7\}}. $$
$$ |B_{main \, layer, \, NE-SW}| = 4*1 = 4. $$

\begin{align*}
T_{main \, layer, \, NW-SE} &= \{ \{4,5\}, \{4,6\}, \{5,6\} \} \times \{7\} \\
T_{main \, layer, \, NW-SE} &= \{ \{4,5,7\}, \{4,6,7\}, \{5,6,7\} \}.
\end{align*}
$$ |T_{main \, layer, \, NW-SE}| = 3.$$

$B_{squares \, of \, the \, cube \, \{4,5,6\}}$ is made of 2 squares of type $\{4, 6\}$ (corresponding to polynomials $X_5$ and $X_5+1$ in $\mathbb{F}_2[X_4,X_5,X_6]$), 1 square of type $\{5, 6\}$ (corresponding to polynomial $X_4$ in $\mathbb{F}_2[X_4,X_5,X_6]$) and 1 square of type $\{4, 5\}$ (corresponding to polynomial $X_6$ in $\mathbb{F}_2[X_4,X_5,X_6]$) (see Figure~\ref{fig:layout_QRM_6_1_1}).
$$ |B_{squares \, of \, the \, cube \, \{4,5,6\}}| = 4.$$

$$ B_{main \, layer, \, NW-SE} = B_{squares \, of \, the \, cube \, \{4,5,6\}} \times B_{edge \, of \, the \, edge \, \{7\}}. $$
$$ |B_{main \, layer, \, NW-SE}| = 4*1 = 4. $$

\begin{lemma}
The concatenation of bases $B_{main \, layer} := (B_{main \, layer, \, center}, B_{main \, layer, \, NE-SW}, B_{main \, layer, \, NW-SE})$ generates the same space as the set of cubes whose type belongs to $T_{main \, layer, \, center}$, $T_{main \, layer, \, NE-SW}$ or $T_{main \, layer, \, NW-SE}$:
$$\Span(B_{main \, layer})
=
\Span(Subcubes(T_{main \, layer}))$$
where
$$T_{main \, layer} \coloneq T_{main \, layer, \, center} \cup T_{main \, layer, \, NE-SW} \cup T_{main \, layer, \, NW-SE}. $$
\label{lem:main_layer}
\end{lemma}

\begin{proof}
Cubes whose type is in $T_{main \, layer, \, center}$ give translations along the $4^{th}$, $5^{th}$ and $6^{th}$ coordinates to cubes of $B_{main \, layer, \, NE-SW}$ and translations along the $1^{st}$, $2^{nd}$ and $3^{rd}$ coordinates to cubes of $B_{main \, layer, \, NW-SE}$.
\end{proof}

Note that so far, we have built a basis for the cubes (of dimension 3) whose type contains 7 in the 7-cube by reproducing the planar layout of Section~\ref{sec:QRM_6_1_1} for the space generated by squares in the 6-cube. We have simply turned squares of the 6-cube into cubes of the 7-cube by extending them along the seventh coordinate. Thus $T_{main \, layer}$ is the set of types (of cardinal 3) that contain the coordinate 7. To generate the Z stabilizer group of $QRM_7(2,2)$, we need to generate cubes whose type is contained in $\{1,2,3,4,5,6\}$. Taking the cartesian product of a basis for squares of type $\{i, j\}$ with $i, j  \in \{1, 2, 3\}$ in the cube $\{1, 2, 3\}$  (space of dimension 4) with a basis for edges of type $\{k\}$ with $k \in \{4, 5, 6\}$ in the cube $\{4, 5, 6\}$ (space of dimension 7) gives a basis for cubes of type $\{i, j, k\}$ with $i, j  \in \{1, 2, 3\}$ and $k \in \{4, 5, 6\}$. This space has dimension $4 * 7 = 28$. The translations along the seventh coordinate are given by cubes of type $\{i, j, 7\}$.

More formally, we define
\begin{align*}
T_{bottom \, layer} &= \{ \{1,2\}, \{1,3\}, \{2,3\} \} \times \{4, 5, 6\} \\
T_{bottom \, layer} &= \{ \{1,2,4\}, \{1,2,5\}, \{1,2,6\}, \{1,3,4\}, \\
&\{1,3,5\}, \{1,3,6\}, \{2,3,4\}, \{2,3,5\}, \{2,3,6\} \}.
\end{align*}

$$ B_{bottom \, layer} = B_{squares \, of \, the \, cube \, \{1,2,3\}} \times B_{edges \, of \, the \, cube \, \{4,5,6\}}. $$
$$ |B_{bottom \, layer}| = 4*7 = 28. $$

\begin{lemma}
The concatenation of bases $(B_{main \, layer}, B_{bottom \, layer})$ generates the same space as the space of cubes whose type belongs to $T_{main \, layer}$ or to $T_{bottom \, layer}$:
$$\Span((B_{main \, layer}, B_{bottom \, layer})
=
\Span(Subcubes(T_{main \, layer} \cup T_{bottom \, layer})).$$
\end{lemma}

\begin{proof}
Cubes whose type is in $T_{main \, layer}$ give translations along the $7^{th}$ coordinate to cubes of $B_{bottom \, layer}$.
\end{proof}

For cubes of type $\{i, j, k\}$ with $i \in \{1, 2, 3\}$ and $j, k \in \{4, 5, 6\}$, all translations along the coordinate $x$ for $x \in \{1, 2, 3\} \setminus \{i\}$ are given by cubes of type $\{i, x, k\}$ and all translations along the seventh coordinate are given by cubes of type $\{i, j, 7\}$. It is therefore sufficient to consider only one edge of each type in the cube $\{1,2,3\}$ (instead of the seven edges needed to have a basis of the edge space of a cube).

More formally, we define
\begin{align*}
T_{top \, layer} &= \{1, 2, 3\} \times \{ \{4,5\}, \{4,6\}, \{5,6\} \} \\
T_{top \, layer} &= \{ \{1,4,5\}, \{1,4,6\}, \{1,5,6\}, \{2,4,5\}, \\
&\{2,4,6\}, \{2,5,6\}, \{3,4,5\}, \{3,4,6\}, \{3,5,6\} \}.
\end{align*}

$B_{one \, edge \, of \, each \, type \, in \, the \, cube \, \{1,2,3\}}$ is made of 1 edge of type $\{1\}$ (corresponding to polynomial $X_2 X_3$ in $\mathbb{F}_2[X_1,X_2,X_3]$, 1 edge of type $\{2\}$ (corresponding to polynomial $X_1 X_3$ in $\mathbb{F}_2[X_1,X_2,X_3]$ and 1 edge of type $\{3\}$ (corresponding to polynomial $X_1 X_2$ in $\mathbb{F}_2[X_1,X_2,X_3]$.

$$ B_{top \, layer} = B_{one \, edge \, of \, each \, type \, in \, the \, cube \, \{1,2,3\}} \times B_{squares \, of \, the \, cube \, \{4,5,6\}}. $$
$$ |B_{top \, layer}| = 3*4 = 12. $$

Note that $B_{top \, layer}$ is made of cubes of the following types,
\begin{itemize}
\item 1 cube of type $\{1, 4, 5\}$,
\item 2 cubes of type $\{1, 4, 6\}$ (image of each other by a translation along the fifth coordinate),
\item 1 cube of type $\{1, 5, 6\}$,
\item 1 cube of type $\{2, 4, 5\}$,
\item 2 cubes of type $\{2, 4, 6\}$ (image of each other by a translation along the fifth coordinate),
\item 1 cube of type $\{2, 5, 6\}$,
\item 1 cube of type $\{3, 4, 5\}$,
\item 2 cubes of type $\{3, 4, 6\}$ (image of each other by a translation along the fifth coordinate),
\item 1 cube of type $\{3, 5, 6\}$.
\end{itemize}

\begin{lemma}
The concatenation of bases $(B_{main \, layer}, B_{bottom \, layer}, B_{top \, layer})$ generates the same space as the space of cubes whose type belongs to $T_{main \, layer}$, $T_{bottom \, layer}$ or $T_{top \, layer}$:
\begin{align*}
&\Span((B_{main \, layer}, B_{bottom \, layer}, B_{top \, layer}) \\
=
&\Span(Subcubes(T_{main \, layer} \cup T_{bottom \, layer} \cup T_{top \, layer})).
\end{align*}
\label{lem:all_except_2}
\end{lemma}

\begin{proof}
Cubes whose type is in $T_{main \, layer}$ give translations along the $7^{th}$ coordinate to cubes whose type is in $T_{bottom \, layer}$ or $T_{top \, layer}$. Cubes whose type is in $T_{bottom \, layer}$ give translations along the $1^{st}$, $2^{nd}$ and $3^{rd}$ coordinates to cubes whose type is in $T_{top \, layer}$.
\end{proof}

Finally, to obtain a basis for the Z stabilizer of $QRM_7(2,2)$, we need one cube of type $\{1, 2, 3\}$ and one cube of type $\{4, 5, 6\}$.

More formally, we define
$$T_{\{1,2,3\}} = \{1,2,3\}.$$
$$T_{\{4,5,6\}} = \{4,5,6\}.$$
$B_{\{1,2,3\}}$ is made of the cube of type $\{1,2,3\}$ corresponding to the polynomial
$$X_4 (X_5+1) (X_6+1) X_7$$
in $\mathbb{F}_2[X_1,X_2,X_3,X_4,X_5,X_6,X_7]$. $B_{\{4,5,6\}}$ is made of the cube of type $\{4,5,6\}$ corresponding to the polynomial
$$(X_1+1) (X_2+1) X_3 X_7$$
in $\mathbb{F}_2[X_1,X_2,X_3,X_4,X_5,X_6,X_7]$. \\

\begin{theorem}
The concatenation of bases
$$B_{QRM_7(2,2)} := (B_{main \, layer}, B_{bottom \, layer}, B_{top \, layer}, B_{\{1,2,3\}}, B_{\{4,5,6\}})$$
generates the Z stabilizer group of $QRM_7(2,2)$. 
\end{theorem}

\begin{proof}
The Z stabilizer group of $QRM_7(2,2)$ is the space generated by cubes (of dimension 3) in the 7-cube. Since $T_{main \, layer}$, $T_{bottom \, layer}$, $T_{top \, layer}$, $T_{\{1,2,3\}}$ and $T_{\{4,5,6\}}$ partition the subsets of cardinal 3 of $\{1,2,3,4,5,6,7\}$, it is sufficient to show that the concatenation of bases $B_{QRM_7(2,2)}$ generates all cubes (of dimension 3). The only cubes missing to the result of Lemma~\ref{lem:all_except_2} are the cubes of type $\{1,2,3\}$ and $\{4,5,6\}$. Cubes of $(B_{main \, layer}, B_{bottom \, layer}, B_{top \, layer})$ give translation along the coordinates 4, 5, 6 and 7 to the cube of $B_{\{1,2,3\}}$ and along the coordinates 1, 2, 3 and 7 to the cube of $B_{\{4,5,6\}}$.
The freedom of $B_{QRM_7(2,2)}$ is ensured by the equality of its cardinal with the dimension of the Z stabilizer group of $QRM_7(2,2)$: both are equal to 99.
\end{proof}

As a sanity check, we verified numerically (see source code\footnote{\url{https://github.com/vivienlonde/unfolded_Quantum_Reed_Muller/tree/main/7_cube/3D_Z_stab_basis.py}}) that $B_{QRM_7(2,2)}$ as described above is indeed a basis for the Z stabilizer group of $QRM_7(2,2)$.

\begin{figure}[H]
\centering
\includegraphics[width=0.7\textwidth]{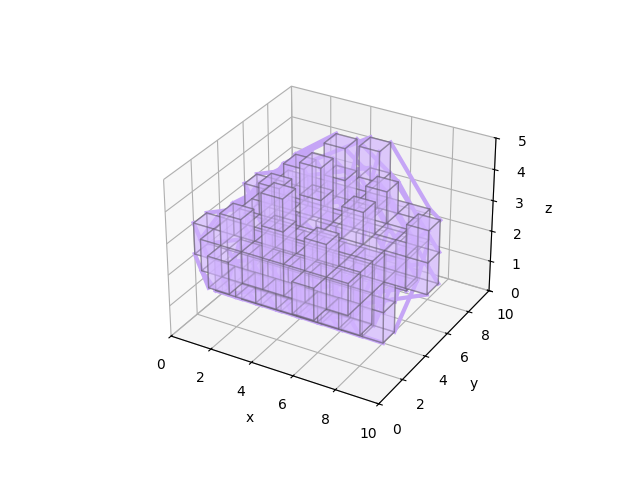}
\caption{3D local layout for the Z stabilizers of $QRM_7(2,2)$. The 3D interactive version of this figure can be generated with the following script: \url{https://github.com/vivienlonde/unfolded_Quantum_Reed_Muller/tree/main/7_cube/3D_layout.py}.}
\label{fig:local_layout_7_cube}
\end{figure}

$QRM_7(2,2)$ has 0 logical qubit. Similarly to $QRM_4(1,1)$, 1 logical qubit is obtained by puncturing $QRM_7(2,2)$ (see Appendix~\ref{sec:puncturing}). A transversal $T^{\dagger}$ gate on each of the 127 physical qubits of the punctured $QRM_7(2,2)$ has the logical effect of a $T$ gate on the logical qubit. The proof is essentially the same as the one of Theorem~\ref{thm:validity}. We refer to \cite{barg2025geometric}, Theorem 6.2 for the general case.

\section{Z minimum distance and number of minimum weight Z logical operators}

The Z minimum distance $d_Z$ of the quantum code that defines a distillation factory gives the exponent of the distillation process. The number $m$ of minimum weight (i.e. of weight $d_Z$) Z logical operators gives the prefactor of the distillation process. Indeed, to leading order, denoting by $p_{in}$ the infididelity before distillation and by $p_{out}$ the infididelity after distillation,
$$p_{out} = m \, p_{in}^{d_Z}.$$
See for instance \cite{haah2017magic} on the definition of a distillation factory from a quantum  code.
 
To count the number of minimum weight codewords in classical Reed Muller codes, we use a recursive definition of these codes that was introduced by Plotkin in 1960 \cite{plotkin1960binary}.

\begin{lemma}[Plotkin recursive construction \cite{plotkin1960binary}]
Let $m$ and $r$ be integers such that $0 \leq r \leq m$. If $r \geq 1$, $RM_m(r)$ is defined from $RM_{m-1}(r)$ and $RM_{m-1}(r-1)$ as follows
$$RM_m(r) = \{(u, u+v) \, | \, u \in RM_{m-1}(r), \, v \in RM_{m-1}(r-1)\}.$$
If $r=0$, $RM_m(0)$ is the trivial code (of dimension $0$).
\label{lem:Plotkin}
\end{lemma}

\begin{proof}
$RM_m(r)$ is the set of indicator vectors (i.e. preimages of $1$) of the polynomials of degree at most $r$ in $\mathbb{F}_2[X_1, \dots, X_m]$. Let $P \in \mathbb{F}_2[X_1, \dots, X_m]$.
There exists $Q \in \mathbb{F}_2[X_1, \dots, X_{m-1}]$ of degree at most $r$ and $R \in \mathbb{F}_2[X_1, \dots, X_{m-1}]$ of degree at most $r-1$ such that
\begin{equation}
P = Q + X_{m} R.
\label{eq:Plotkin}
\end{equation}
Since $Q$ corresponds to a codeword of $RM_{m-1}(r)$ and $R$ corresponds to a codeword of $RM_{m-1}(r-1)$, Equation~\ref{eq:Plotkin} defines a map from $RM_{m}(r)$ to $RM_{m-1}(r) \times 
RM_{m-1}(r-1)$. \\

Conversely, if we denote by $u$ the codeword of $RM_{m-1}(r)$ associated to $Q$ and by $v$ the codeword of $RM_{m-1}(r-1)$ associated to $R$, the unique codeword associated to $Q + X_{m} R$ is the concatenation of $u$ and $u+v$ (with the convention that the first $2^{m-1}$ bits correspond to $x_m = 0$ and the last $2^{m-1}$ bits correspond to $x_m = 1$).
\end{proof}

Lemma~\ref{lem:minimum_distance} gives the minimum distance of Reed Muller codes. We recall its proof because it is useful to characterize (and count) the minimum weight codewords of $RM_m(r)$.

\begin{lemma}[\cite{muller1954application}]
$RM_m(r)$ has minimum distance $2^{m-r}$.
\label{lem:minimum_distance}
\end{lemma}

\begin{proof}
The proof is by induction on $m+r$.

Let $m \geq 0$ be an integer. The base cases are $r=0$. In this case, the only polynomials of degree $0$ are $P=0$ and $P=1$ and thus the only codewords of $RM_m(0)$ are $0 \dots 0$ and $1 \dots 1$. Therefore the minimum distance of $RM_m(0)$ is $2^m$.

Let $m \geq r > 0$ be integers. Let $w = (u, u+v) \in RM_m(r)$, with $u \in RM_{m-1}(r)$ and $v \in RM_{m-1}(r-1)$. From the induction hypothesis, we have $|u| \geq 2^{m-r-1}$ and $|v| \geq 2^{m-r}$. Since $|v| = |u| + |u+v| - 2 |u \setminus v|$, we have
$$ |w| = |v| + 2 |u \setminus v|.$$
Therefore
\begin{equation}
|w| \geq |v| \geq 2^{m-r}.
\label{ineq:minimum_distance}
\end{equation}
\end{proof}

In Lemma~\ref{lem:minimum_weight_codewords}, we charaterize the minimum weight codewords of $RM_m(r)$ by studying the equality cases in Inequalities~\ref{ineq:minimum_distance}.

\begin{lemma}[\cite{plotkin1960binary}]
Let $m$ and $r$ be integers such that $m \geq 1$ and $0 \leq r \leq m$. Minimum weight codewords of $RM_{m}(r)$ are indicator vectors of products of $r$ linearly independent affine forms in $m$ variables.
\label{lem:minimum_weight_codewords}
\end{lemma}

\begin{proof}
The proof is by induction on $m$. The base case ($m=1$) is clear. Indeed, for $m=1$ and $r=0$, the constant polynomial $P=1$ is the product of $0$ linearly independent affine form and for $m=1$ and $r=1$, there are 4 polynomials of degree at most 1 in 1 variables: $0$, $1$, $X$ and $X+1$. The minimum weight nontrivial codewords of $RM_1(1)$ are $01$ (corresponding to $X$) and $10$ (corresponding to $X+1$) and indeed $X$ and $X+1$ are the only (product of 1 linearly independent) affine forms in 1 variable. \\

For $m>1$, if $r=0$ the result is true since the constant polynomial $P=1$ is the product of $0$ linearly independent affine form. \\

We use the induction hypothesis to prove the case $m>1$ and $r>0$. We use the notations of the proof of Lemma~\ref{ineq:minimum_distance} and study the equality case of inequality
\begin{equation}
|w| = |v| + 2 |u \setminus v| \geq 2^{m-r}.
\label{weight_inequality}
\end{equation}
Recall that $w = (u, u+v)$ with $u \in RM_{m-1}(r)$ and $v \in RM_{m-1}(r-1)$. By the induction hypothesis, $|v|=0$ or $|v| \geq 2^{m-r}$. Therefore only these two values for $|v|$ can lead to the equality case of Inequality \ref{weight_inequality}.

\begin{itemize}
\item \underline{case $|v| = 0$.} \\
In this case, $|w| = 2 |u|$ and the equality case becomes $|u| = 2^{m-r-1}$. By the induction hypothesis, $u$ is the indicator vector of the product of $r$ linearly independent affine forms in $m-1$ variables:
$$P_u = \prod_{i=1}^{r} (a_{i,0} + \sum_{j=1}^{m-1} a_{i,j} X_j),$$
where $a_{i,j} \in \mathbb{F}_2$ for $1 \leq i \leq r$ and $0 \leq j \leq m-1$.

Since $w = (u, 0)$, considering these $r$ affine forms to be affine forms in $m$ variables (with every coefficient in front of $X_m$ set to $0$) defines the polynomial $P_w$ whose indicator vector is $w$:
$$P_w = \prod_{i=1}^{r} (a_{i,0} + \sum_{j=1}^{m} a_{i,j} X_j),$$
where $a_{i,j}(P_w) := a_{i,j}(P_u)$ for $j \leq m-1$ and $a_{i,m}(P_w) := 0$.

\item \underline{case $|v| = 2^{m-r}$.} \\
In this case, the equality case of Inequality \ref{weight_inequality} becomes $|u \setminus v|=0$, or equivalently
$$\text{supp}(u) \subset \text{supp}(v),$$
where $\text{supp}(u)$, the support of $u$, is the set of indices $i$ such that $u_i=1$:
$$\text{supp}(u) := \{i \in \mathbb{F}_2^m \, | \,  u_i = 1\}.$$
By the induction hypothesis, $v$ is the indicator function of the product of $r-1$ linearly independent affine forms in $m-1$ variables:
$$P_v = \prod_{i=1}^{r-1} (a_{i,0} + \sum_{j=1}^{m-1} a_{i,j} X_j),$$
where $a_{i,j} \in \mathbb{F}_2$ for $1 \leq i \leq r-1$ and $0 \leq j \leq m-1$.
Since $\text{supp}(u) \subset \text{supp}(v)$, for every $x \in \mathbb{F}_2^{m-1}$, $P_u(x)=1$ implies that $P_v(x)=1$. Therefore, $P_u P_v$ and $P_v$ have the same indicator vector. Equivalently, $P_u P_v = P_u$ in $\mathbb{F}_2[X_1, \dots, X_{m-1}] / \prod_{i=1}^{m-1} (X_i^2 - X_i)$. Therefore $P_v$ divides $P_u$ in this quotient space.
%(argument : Something like, there is a one to one correspondence between vectors and polynomials of degree at most $m$ since any weight 1 vector corresponds to a degree m polynomial. if a polynomial evaluates to 0 when $x_1=0$, then it is divisible by $X_1$. If it evaluates to 0 on the zeroes of an affine form, then it is divisible by this affine form. Since the affine forms are linearly independent, we can take their product...).
And since $P_u$ has degree at most $r$ and $P_v$ is the product of $r-1$ affine forms, there exists an affine form $A_{(r)} := a_{r,0} + \sum_{j=1}^{m-1} a_{r,j} X_j$ such that
$$P_u = P_v A_{(r)}.$$
Therefore, since $P_w = P_u + X_m P_v$,
$$P_w = P_v (A_{(r)} + X_m).$$
Thus, $P_w$ is the product of $r$ affine forms. These $r$ forms are linearly independent since the first $(r-1)^{th}$ are linearly independent by the induction hypothesis and the $r^{th}$ affine form (defined as $A_{(r)} + X_m$) is the only one with an $X_m$ term.
\end{itemize}
\end{proof}

\begin{theorem}[\cite{plotkin1960binary}]
Let $m$ and $r$ be integers such that $0 \leq r \leq m$. There are
$$2^{r} \frac{(2^m-1) \dots (2^{m-r+1}-1)}{(2^r-1)\dots(2^1-1)}$$
minimum weight codewords in $RM_{m}(r)$.
\label{thm:number_of_minimum_weight_codewords}
\end{theorem}

\begin{proof}
From Lemma~\ref{lem:minimum_weight_codewords}, it is sufficient to count dimension $m-r$ affine subspaces of $\mathbb{F}_2^m$. We count such subspaces by multiplying by $2^{r}$ possible translations the number of dimension $m-r$ linear subspaces of $\mathbb{F}_2^m$. The number of dimension $m-r$ linear subspaces of $\mathbb{F}_2^m$ equals the number $F$ of ordered sets of $r$ independent linear forms divided by the number $B$ of different ordered sets of $r$ independent linear forms that define the same dimension $m-r$ linear subspace of $\mathbb{F}_2^m$.  \\

To calculate $F$, there are $2^m-1$ choices for the first linear form (any non zero linear form is possible), $2^m-2$ choices for the second linear form (any linear form that is not in the span of the first one is possible), $\dots$, $2^m-2^{r-1}$ choices for the $r^{th}$ linear form (any linear form that is not in the span of the $r-1$ first ones is possible). Therefore
$$F = (2^m-1) \dots (2^m-2^{r-1}).$$

Counting the number of different ordered sets of $r$ independent linear forms that define the same dimension $m-r$ linear subspace of $\mathbb{F}_2^m$ is like calculating $F$ in a space of dimension $r$ instead of $m$. Therefore
$$B = (2^r-1) \dots (2^r-2^{r-1}).$$

\begin{align*}
A &= 2^{r} \frac{F}{B} \\
A &= 2^{r} \prod_{i=0}^{r-1} \frac{2^m-2^i}{2^r-2^i} \\
A &= 2^{r} \prod_{i=0}^{r-1} \frac{2^{m-i}-1}{2^{r-i}-1}.
\end{align*}

\end{proof}

We apply Theorem~\ref{thm:number_of_minimum_weight_codewords} to the Z stabilizer group of the codes investigated in this article.

\begin{itemize}

\item \underline{$QRM_4(1,1)$} \\
The Z stabilizer group of $QRM_1(1,1)$ is $RM_4(2)$ and so is its Z logical group.
The minimum distance of $RM_4(2)$ is $4$ and there are $2^2 \frac{(2^4-1)(2^3-1)}{(2^2-1)(2^1-1)} = 4 * \frac{15*7}{3} = 140$ weight $4$ codewords in $RM_7(4)$.

\item \underline{Small unfolded code (punctured $QRM_4(1,1)$, 15 $\ket{T}$ to $\ket{T}$ factory)} \\
The Z stabilizer group of the punctured $QRM_4(1,1)$ is a shortened $RM_4(2)$ and its Z logical group is the punctured $RM_4(2)$ (see Appendix~\ref{sec:puncturing}). Therefore, the Z minimum distance is the minimum distance of $RM_4(2)$ minus $1$. Namely $4-1=3$. Minimum weight codewords correspond to minimum weight codewords of $RM_4(2)$ which have the punctured qubit in their support. If we consider that the punctured bit is $0000$ (all bits play the same role in $RM_4(2)$), minimum weight logical operators correspond to dimension 2 linear subspaces of forms over $\mathbb{F}_2$. Therefore, there are $\frac{(2^4-1)(2^3-1)}{(2^2-1)(2^1-1)} = \frac{15*7}{3} = 35$ Z logical operators of weight $3$ in the punctured $QRM_4(1,1)$. Thus, the dominant term of the logical error probability of the factory is $35p^3$.

\item \underline{$QRM_3(0,1)$ ($8$ $\ket{T}$ to $\ket{CCZ}$ factory \cite{jones2013low})} \\
The Z stabilizer group of $QRM_3(0,1)$ is $RM_3(1)$ and its Z logical group is $RM_3(2)$.
The Z minimum distance is $2$ and there are $2^2\frac{(2^3-1)(2^2-1)}{(2^2-1)(2^1-1)} = 4 \frac{7*3}{3} = 28$ Z logical operators of weight $2$. Thus, the dominant term of the logical error probability of the factory is $28p^2$. \\

\item \underline{Rubik's code ($QRM_6(1,2)$, 64 $\ket{T}$ to 15 $\ket{CCZ}$ factory)} \\
The Z stabilizer group of $QRM_6(1,2)$ is $RM_6(3)$ and its Z logical group is $RM_6(4)$. The Z minimum distance of $QRM_6(1,2)$ is $2^{6-4}=4$ and there are $2^{4} \frac{(2^6-1)(2^5-1)}{(2^2-1)(2^1-1)} = 16*\frac{63*31}{3} = 10416$ Z logical operators of weight $4$. Thus, the dominant term of the logical error probability of the rubik's cube factory is $10416p^4$. \\

\item \underline{Big unfolded code ($QRM_6(1,2)$ with 12 extra $Z$ stabilizers, 64 $\ket{T}$ to $\ket{CCZ}$ factory)} \\
The minimum distance of the big unfolded code is 4. Indeed, among the 10416 weight 4 $Z$ logical operators of $QRM_6(1,2)$, we checked numerically\footnote{\url{https://github.com/vivienlonde/unfolded_Quantum_Reed_Muller}} that exactly 2160 are orthogonal to all $X$ logical operators of the big unfolded code and are thus trivial (i.e. are $Z$ stabilizers) in the big unfolded code. The 8256 others are not orthogonal to at least one $X$ logical operator and are thus non trivial $Z$ logical operators. So there are 8256 logical operators of weight 4 in the big unfolded code. Thus, the dominant term of the logical error probability of the factory is $8256p^4$.

\item \underline{$QRM_7(2,2)$} \\
The Z stabilizer group of $QRM_7(2,2)$ is $RM_7(4)$ and so is its Z logical group.
The minimum distance of $RM_7(4)$ is $8$ and there are $2^4 \frac{(2^7-1)(2^6-1)(2^5-1)}{(2^3-1)(2^2-1)(2^1-1)} = 16 * \frac{127*63*31}{7*3} = 188976$ weight $8$ codewords in $RM_7(4)$. 

\item \underline{Punctured $QRM_7(2,2)$ ($127$ $\ket{T}$ to $\ket{T}$ factory)} \\
The $Z$ stabilizer group of the punctured $QRM_7(2,2)$ is a shortened $RM_7(4)$ and its Z logical group is the punctured $RM_7(4)$. The $Z$ minimum distance is the minimum distance of $RM_7(4)$ minus 1. Namely $8-1=7$. Minimum weight codewords correspond to minimum weight codewords of $RM_7(4)$ which have the punctured qubit in their support. If we consider that the punctured bit is $0000000$ (all bits play the same role in $RM_7(4)$), minimum weight logical operators correspond to dimension 3 linear subspaces of forms over $\mathbb{F}_2$. Therefore, there are $\frac{(2^7-1)(2^6-1)(2^5-1)}{(2^3-1)(2^2-1)(2^1-1)} = \frac{127*63*31}{7*3} = 11811$ $Z$ logical operators of weight 7 in the punctured $QRM_7(2,2)$. Thus, the dominant term of the logical error probability of the factory is $11811p^7$.
\end{itemize}

\section{Conclusion and outlook}

We have described 2D and 3D local layouts for the basis of the $Z$ stabilizer group of some Reed-Muller distillation factories. We have described the local layout of two codes with $Z$-distance 4 and one code with $Z$-distance 7 and in each case give the prefactor $f$ such that the factory's dominant error probability is $p \rightarrow f p^{d_Z}$.
For the small unfolded code, the authors of \cite{ruiz2025unfolded} describe in detail the distillation protocol and verify through numerical simulation that the error floor $f p^{d_Z}$ is reached by their protocol. For the big unfolded code, the rubik's cube code and the punctured $QRM_7(2,2)$ code described in this article, we are confident that a protocol very similar to the one of \cite{ruiz2025unfolded} yields the error floor $f p^{d_Z}$ in each case.  We leave it to future work to verify this through numerical simulations.

\section{Acknowledgements}

I thank Diego Ruiz, Christophe Vuillot, Mazyar Mirrahimi and Hugo Jacinto for fruitfull discussions and feedbacks on this work. I thank Elie Gouzien, Linde Wester Hansen, Jérémie Guillot and everyone at Alice\&Bob for creating a nice research environment.

\appendix

\section{Puncturing a quantum code with 0 logical qubit, the example of punctured $QRM_4(1,1)$ (a.k.a. the small unfolded code)}
\label{sec:puncturing}

$QRM_m(r,r)$ encodes $0$ logical qubit in $2^m$ physical qubits \cite{barg2025geometric}.

The small unfolded code from \cite{ruiz2025unfolded} is obtained from (quantum-)puncturing $QRM_4(1,1)$.  $QRM_4(1,1)$ has parameters $[[n, k]] = [[16, 0]]$ and Z stabilizers of weight at least $w=4$. The small unfolded code has parameters $[[n', k', d']] = [[15, 1, 3]]$. Note that $n'=n-1$, $k'=k+1$ and $d'=w+1$.

When one physical qubit is deleted by (quantum-)puncturing, stabilizer spaces are shortened (as classical codes): one physical qubit is deleted and all stabilizers that have this physical qubit in their support are deleted. $C$ shortened is denoted $C^{\circ}$. This amounts to puncturing the logical operator spaces (as classical codes). Indeed the dual of a shortened classical code is the puncturation of the dual of the original classical code. Classicaly puncturing a code consists in ignoring one physical bit (and keeping the check that acted on a set containing that physical bit: they now act on a set with one less bit). $C$ punctured is denoted $C^*$. \\

$C^{\circ} \subset C^*$ since every element of $C^{\circ}$ is an element of  $C^*$. Note that there is a one to one correspondence between elements of $C^*$ and elements of $C$ (as long as there is no weight $1$ element in $C$ whose support is the punctured qubit). We therefore abuse notations and refer to properties of $C$ and of $C^*$ with the same objects.

When the inclusion $C^{\circ} \subset C^*$ is strict for the Z stabilizer group and for the X stabilizer group, puncturing creates a logical qubit. The inclusion is strict if and only if $r^{\circ} = r - 1$. Let $r_X$ and $r_Z$ respectively denote the ranks of parity-check matrices $H_X$ and $H_Z$. Let's assume that $r_X^{\circ} = r_X - 1$ and $r_Z^{\circ} = r_Z - 1$, where $r_X^{\circ}$ and $r_Z^{\circ}$respectively  are the ranks of the punctered parity-check matrices $H_X$ and $H_Z$. \\

In the rest of this section, we illustrate puncturing with the code $QRM_4(1,1)$. We consider two sets of generators for stabilizer groups of $QRM_4(1,1)$. The first set of generators are standard subcubes and the second set of generators are subcubes of a fixed dimension.

\subsection{Presentation of $QRM_4(1,1)$ with standard subcubes}

Before puncturing, the parity-check matrices are:

$$ H_X^{\text{standard}} = L_X^{\text{standard}} =
\begin{pmatrix}
1 &    1 & 1 & 1 & 1 &    1 & 1 & 1 & 1 & 1 & 1 &    1 & 1 & 1 & 1 &    1 \\

0 &    1 & 0 & 0 & 0 &    1 & 1 & 1 & 0 & 0 & 0 &    1 & 1 & 1 & 0 &    1 \\
0 &    0 & 1 & 0 & 0 &    1 & 0 & 0 & 1 & 1 & 0 &    1 & 1 & 0 & 1 &    1 \\
0 &    0 & 0 & 1 & 0 &    0 & 1 & 0 & 1 & 0 & 1 &    1 & 0 & 1 & 1 &    1 \\
0 &    0 & 0 & 0 & 1 &    0 & 0 & 1 & 0 & 1 & 1 &    0 & 1 & 1 & 1 &    1 \\
\end{pmatrix}
$$

The above matrix has $5$ rows and rank $5$.

$$ H_Z^{\text{standard}} = L_Z^{\text{standard}} =
\begin{pmatrix}
1 &    1 & 1 & 1 & 1 &    1 & 1 & 1 & 1 & 1 & 1 &    1 & 1 & 1 & 1 &    1 \\

0 &    1 & 0 & 0 & 0 &    1 & 1 & 1 & 0 & 0 & 0 &    1 & 1 & 1 & 0 &    1 \\
0 &    0 & 1 & 0 & 0 &    1 & 0 & 0 & 1 & 1 & 0 &    1 & 1 & 0 & 1 &    1 \\
0 &    0 & 0 & 1 & 0 &    0 & 1 & 0 & 1 & 0 & 1 &    1 & 0 & 1 & 1 &    1 \\
0 &    0 & 0 & 0 & 1 &    0 & 0 & 1 & 0 & 1 & 1 &    0 & 1 & 1 & 1 &    1 \\

0 &    0 & 0 & 0 & 0 &    1 & 0 & 0 & 0 & 0 & 0 &    1 & 1 & 0 & 0 &    1 \\
0 &    0 & 0 & 0 & 0 &    0 & 1 & 0 & 0 & 0 & 0 &    1 & 0 & 1 & 0 &    1 \\
0 &    0 & 0 & 0 & 0 &    0 & 0 & 1 & 0 & 0 & 0 &    0 & 1 & 1 & 0 &    1 \\
0 &    0 & 0 & 0 & 0 &    0 & 0 & 0 & 1 & 0 & 0 &    1 & 0 & 0 & 1 &    1 \\
0 &    0 & 0 & 0 & 0 &    0 & 0 & 0 & 0 & 1 & 0 &    0 & 1 & 0 & 1 &    1 \\
0 &    0 & 0 & 0 & 0 &    0 & 0 & 0 & 0 & 0 & 1 &    0 & 0 & 1 & 1 &    1 \\
\end{pmatrix}
$$

The above matrix has $11$ rows and rank $11$.

Note that $L_Z^{\text{standard}} = H_Z^{\text{standard}}$ and $L_X^{\text{standard}} = H_X^{\text{standard}}$, which confirms that no logical qubit is encoded. \\

Shortening consists in deleting the first column and all the rows that have a $1$ on that column, i.e. the first row in both cases:

$$ H_X'^{\text{ standard}} =
\begin{pmatrix}
1 & 0 & 0 & 0 &    1 & 1 & 1 & 0 & 0 & 0 &    1 & 1 & 1 & 0 &    1 \\
0 & 1 & 0 & 0 &    1 & 0 & 0 & 1 & 1 & 0 &    1 & 1 & 0 & 1 &    1 \\
0 & 0 & 1 & 0 &    0 & 1 & 0 & 1 & 0 & 1 &    1 & 0 & 1 & 1 &    1 \\
0 & 0 & 0 & 1 &    0 & 0 & 1 & 0 & 1 & 1 &    0 & 1 & 1 & 1 &    1 \\
\end{pmatrix}
$$

The above matrix has $4$ rows and rank $4$.

$$H_Z'^{\text{ standard}} =
\begin{pmatrix}
1 & 0 & 0 & 0 &    1 & 1 & 1 & 0 & 0 & 0 &    1 & 1 & 1 & 0 &    1 \\
0 & 1 & 0 & 0 &    1 & 0 & 0 & 1 & 1 & 0 &    1 & 1 & 0 & 1 &    1 \\
0 & 0 & 1 & 0 &    0 & 1 & 0 & 1 & 0 & 1 &    1 & 0 & 1 & 1 &    1 \\
0 & 0 & 0 & 1 &    0 & 0 & 1 & 0 & 1 & 1 &    0 & 1 & 1 & 1 &    1 \\

0 & 0 & 0 & 0 &    1 & 0 & 0 & 0 & 0 & 0 &    1 & 1 & 0 & 0 &    1 \\
0 & 0 & 0 & 0 &    0 & 1 & 0 & 0 & 0 & 0 &    1 & 0 & 1 & 0 &    1 \\
0 & 0 & 0 & 0 &    0 & 0 & 1 & 0 & 0 & 0 &    0 & 1 & 1 & 0 &    1 \\
0 & 0 & 0 & 0 &    0 & 0 & 0 & 1 & 0 & 0 &    1 & 0 & 0 & 1 &    1 \\
0 & 0 & 0 & 0 &    0 & 0 & 0 & 0 & 1 & 0 &    0 & 1 & 0 & 1 &    1 \\
0 & 0 & 0 & 0 &    0 & 0 & 0 & 0 & 0 & 1 &    0 & 0 & 1 & 1 &    1 \\
\end{pmatrix}
$$

The above matrix has $10$ rows and rank $10$. \\

Puncturing consists in deleting the first column and keeping every row:
$$ L_X'^{\text{ standard}} =
\begin{pmatrix}
1 & 1 & 1 & 1 &    1 & 1 & 1 & 1 & 1 & 1 &    1 & 1 & 1 & 1 &    1 \\

1 & 0 & 0 & 0 &    1 & 1 & 1 & 0 & 0 & 0 &    1 & 1 & 1 & 0 &    1 \\
0 & 1 & 0 & 0 &    1 & 0 & 0 & 1 & 1 & 0 &    1 & 1 & 0 & 1 &    1 \\
0 & 0 & 1 & 0 &    0 & 1 & 0 & 1 & 0 & 1 &    1 & 0 & 1 & 1 &    1 \\
0 & 0 & 0 & 1 &    0 & 0 & 1 & 0 & 1 & 1 &    0 & 1 & 1 & 1 &    1 \\
\end{pmatrix}
$$

The above matrix has $5$ rows and rank $5$.

$$L_Z'^{\text{ standard}} =
\begin{pmatrix}
1 & 1 & 1 & 1 &    1 & 1 & 1 & 1 & 1 & 1 &    1 & 1 & 1 & 1 &    1 \\

1 & 0 & 0 & 0 &    1 & 1 & 1 & 0 & 0 & 0 &    1 & 1 & 1 & 0 &    1 \\
0 & 1 & 0 & 0 &    1 & 0 & 0 & 1 & 1 & 0 &    1 & 1 & 0 & 1 &    1 \\
0 & 0 & 1 & 0 &    0 & 1 & 0 & 1 & 0 & 1 &    1 & 0 & 1 & 1 &    1 \\
0 & 0 & 0 & 1 &    0 & 0 & 1 & 0 & 1 & 1 &    0 & 1 & 1 & 1 &    1 \\

0 & 0 & 0 & 0 &    1 & 0 & 0 & 0 & 0 & 0 &    1 & 1 & 0 & 0 &    1 \\
0 & 0 & 0 & 0 &    0 & 1 & 0 & 0 & 0 & 0 &    1 & 0 & 1 & 0 &    1 \\
0 & 0 & 0 & 0 &    0 & 0 & 1 & 0 & 0 & 0 &    0 & 1 & 1 & 0 &    1 \\
0 & 0 & 0 & 0 &    0 & 0 & 0 & 1 & 0 & 0 &    1 & 0 & 0 & 1 &    1 \\
0 & 0 & 0 & 0 &    0 & 0 & 0 & 0 & 1 & 0 &    0 & 1 & 0 & 1 &    1 \\
0 & 0 & 0 & 0 &    0 & 0 & 0 & 0 & 0 & 1 &    0 & 0 & 1 & 1 &    1 \\
\end{pmatrix}
$$

The above matrix has $11$ rows and rank $11$.

After puncturing, there is one logical qubit. Indeed:
$$n' = n-1$$
$$r_X' = \rank(H_X'^{\text{ standard}}) = r_X-1$$
$$r_Z' = \rank(H_Z'^{\text{ standard}}) = r_Z-1.$$

Therefore
\begin{align*}
k' &= n' - r_X' - r_Z' \\
k' &= k+1 \\
k' &= 1.
\end{align*}

$d_X'$ is the minimum weight of the first row up to the other ones minus 1 (because of the deleted qubit). In this case, $d_X' = 8-1 = 7$. Similarly, $d_Z'$ is the minimum weight of the first row up to the other ones minus 1 (because of the deleted qubit). In this case, $d_Z' = 4-1 = 3$.

\subsection{Presentation with all subcubes of a fixed dimension}

Reed-Muller codes admit a generating family that consists of all subcubes of a fixed dimension. In this case, all 3-cubes for $H_X$ and all 2-cubes for $H_Z$:

$$H_X^{\text{fixed}} = L_X^{\text{fixed}} =
\begin{pmatrix}
0 &    1 & 0 & 0 & 0 &    1 & 1 & 1 & 0 & 0 & 0 &    1 & 1 & 1 & 0 &    1 \\
0 &    0 & 1 & 0 & 0 &    1 & 0 & 0 & 1 & 1 & 0 &    1 & 1 & 0 & 1 &    1 \\
0 &    0 & 0 & 1 & 0 &    0 & 1 & 0 & 1 & 0 & 1 &    1 & 0 & 1 & 1 &    1 \\
0 &    0 & 0 & 0 & 1 &    0 & 0 & 1 & 0 & 1 & 1 &    0 & 1 & 1 & 1 &    1 \\

1 &    0 & 1 & 1 & 1 &    0 & 0 & 0 & 1 & 1 & 1 &    0 & 0 & 0 & 1 &    0 \\
1 &    1 & 0 & 1 & 1 &    0 & 1 & 1 & 0 & 0 & 1 &    0 & 0 & 1 & 0 &    0 \\
1 &    1 & 1 & 0 & 1 &    1 & 0 & 1 & 0 & 1 & 0 &    0 & 1 & 0 & 0 &    0 \\
1 &    1 & 1 & 1 & 0 &    1 & 1 & 0 & 1 & 0 & 0 &    1 & 0 & 0 & 0 &    0 \\
\end{pmatrix}
$$

The above matrix has $8$ rows and rank $5$. It can be checked through an explicit calculation that $H_X^{\text{standard}}$ and $H_X^{\text{fixed}}$ generate the same column space.

$$H_Z^{\text{fixed}} = L_Z^{\text{fixed}} =
\begin{pmatrix}
0 &    0 & 0 & 0 & 0 &    1 & 0 & 0 & 0 & 0 & 0 &    1 & 1 & 0 & 0 &    1 \\
0 &    0 & 1 & 0 & 0 &    0 & 0 & 0 & 1 & 1 & 0 &    0 & 0 & 0 & 1 &    0 \\
0 &    1 & 0 & 0 & 0 &    0 & 1 & 1 & 0 & 0 & 0 &    0 & 0 & 1 & 0 &    0 \\
1 &    0 & 0 & 1 & 1 &    0 & 0 & 0 & 0 & 0 & 1 &    0 & 0 & 0 & 0 &    0 \\

0 &    0 & 0 & 0 & 0 &    0 & 1 & 0 & 0 & 0 & 0 &    1 & 0 & 1 & 0 &    1 \\
0 &    0 & 0 & 1 & 0 &    0 & 0 & 0 & 1 & 0 & 1 &    0 & 0 & 0 & 1 &    0 \\
0 &    1 & 0 & 0 & 0 &    1 & 0 & 1 & 0 & 0 & 0 &    0 & 1 & 0 & 0 &    0 \\
1 &    0 & 1 & 0 & 1 &    0 & 0 & 0 & 0 & 1 & 0 &    0 & 0 & 0 & 0 &    0 \\

0 &    0 & 0 & 0 & 0 &    0 & 0 & 1 & 0 & 0 & 0 &    0 & 1 & 1 & 0 &    1 \\
0 &    0 & 0 & 0 & 1 &    0 & 0 & 0 & 0 & 1 & 1 &    0 & 0 & 0 & 1 &    0 \\
0 &    1 & 0 & 0 & 0 &    1 & 1 & 0 & 0 & 0 & 0 &    1 & 0 & 0 & 0 &    0 \\
1 &    0 & 1 & 1 & 0 &    0 & 0 & 0 & 1 & 0 & 0 &    0 & 0 & 0 & 0 &    0 \\

0 &    0 & 0 & 0 & 0 &    0 & 0 & 0 & 1 & 0 & 0 &    1 & 0 & 0 & 1 &    1 \\
0 &    0 & 0 & 1 & 0 &    0 & 1 & 0 & 0 & 0 & 1 &    0 & 1 & 0 & 0 &    0 \\
0 &    0 & 1 & 0 & 0 &    1 & 0 & 0 & 0 & 1 & 0 &    0 & 0 & 1 & 0 &    0 \\
1 &    1 & 0 & 0 & 1 &    0 & 0 & 1 & 0 & 0 & 0 &    0 & 0 & 0 & 0 &    0 \\

0 &    0 & 0 & 0 & 0 &    0 & 0 & 0 & 0 & 1 & 0 &    0 & 1 & 0 & 1 &    1 \\
0 &    0 & 0 & 0 & 1 &    0 & 0 & 1 & 0 & 0 & 1 &    0 & 0 & 1 & 0 &    0 \\
0 &    0 & 1 & 0 & 0 &    1 & 0 & 0 & 1 & 0 & 0 &    1 & 0 & 0 & 0 &    0 \\
1 &    1 & 0 & 1 & 0 &    0 & 1 & 0 & 0 & 0 & 0 &    0 & 0 & 0 & 0 &    0 \\

0 &    0 & 0 & 0 & 0 &    0 & 0 & 0 & 0 & 0 & 1 &    0 & 0 & 1 & 1 &    1 \\
0 &    0 & 0 & 0 & 1 &    0 & 0 & 1 & 0 & 1 & 0 &    0 & 1 & 0 & 0 &    0 \\
0 &    0 & 0 & 1 & 0 &    0 & 1 & 0 & 1 & 0 & 0 &    1 & 0 & 0 & 0 &    0 \\
1 &    1 & 1 & 0 & 0 &    1 & 0 & 0 & 0 & 0 & 0 &    0 & 0 & 0 & 0 &    0 \\
\end{pmatrix}
$$

The above matrix has $24$ rows and rank $11$. It can be checked through an explicit calculation that $H_Z^{\text{standard}}$ and $H_Z^{\text{fixed}}$ generate the same column space. \\

Note that $L_X^{\text{fixed}} = H_X^{\text{fixed}}$ and $L_Z^{\text{fixed}} = H_Z^{\text{fixed}}$ since no logical qubit is encoded. 

Shortening consists in deleting the first column and all the rows that have a $1$ on that column:

$$ H_X'^{\text{ fixed}} =
\begin{pmatrix}
1 & 0 & 0 & 0 &    1 & 1 & 1 & 0 & 0 & 0 &    1 & 1 & 1 & 0 &    1 \\
0 & 1 & 0 & 0 &    1 & 0 & 0 & 1 & 1 & 0 &    1 & 1 & 0 & 1 &    1 \\
0 & 0 & 1 & 0 &    0 & 1 & 0 & 1 & 0 & 1 &    1 & 0 & 1 & 1 &    1 \\
0 & 0 & 0 & 1 &    0 & 0 & 1 & 0 & 1 & 1 &    0 & 1 & 1 & 1 &    1 \\
\end{pmatrix}
$$

The above matrix has $4$ rows and rank $4$. Note that $H_X'^{\text{ standard}} = H_X'^{\text{ fixed}}$, which trivially implies that $H_X'^{\text{ standard}}$ and $H_X'^{\text{ fixed}}$ generate the same column space. \\

$$H_Z'^{\text{ fixed}} =
\begin{pmatrix}
0 & 0 & 0 & 0 &    1 & 0 & 0 & 0 & 0 & 0 &    1 & 1 & 0 & 0 &    1 \\
0 & 1 & 0 & 0 &    0 & 0 & 0 & 1 & 1 & 0 &    0 & 0 & 0 & 1 &    0 \\
1 & 0 & 0 & 0 &    0 & 1 & 1 & 0 & 0 & 0 &    0 & 0 & 1 & 0 &    0 \\

0 & 0 & 0 & 0 &    0 & 1 & 0 & 0 & 0 & 0 &    1 & 0 & 1 & 0 &    1 \\
0 & 0 & 1 & 0 &    0 & 0 & 0 & 1 & 0 & 1 &    0 & 0 & 0 & 1 &    0 \\
1 & 0 & 0 & 0 &    1 & 0 & 1 & 0 & 0 & 0 &    0 & 1 & 0 & 0 &    0 \\

0 & 0 & 0 & 0 &    0 & 0 & 1 & 0 & 0 & 0 &    0 & 1 & 1 & 0 &    1 \\
0 & 0 & 0 & 1 &    0 & 0 & 0 & 0 & 1 & 1 &    0 & 0 & 0 & 1 &    0 \\
1 & 0 & 0 & 0 &    1 & 1 & 0 & 0 & 0 & 0 &    1 & 0 & 0 & 0 &    0 \\

0 & 0 & 0 & 0 &    0 & 0 & 0 & 1 & 0 & 0 &    1 & 0 & 0 & 1 &    1 \\
0 & 0 & 1 & 0 &    0 & 1 & 0 & 0 & 0 & 1 &    0 & 1 & 0 & 0 &    0 \\
0 & 1 & 0 & 0 &    1 & 0 & 0 & 0 & 1 & 0 &    0 & 0 & 1 & 0 &    0 \\

0 & 0 & 0 & 0 &    0 & 0 & 0 & 0 & 1 & 0 &    0 & 1 & 0 & 1 &    1 \\
0 & 0 & 0 & 1 &    0 & 0 & 1 & 0 & 0 & 1 &    0 & 0 & 1 & 0 &    0 \\
0 & 1 & 0 & 0 &    1 & 0 & 0 & 1 & 0 & 0 &    1 & 0 & 0 & 0 &    0 \\

0 & 0 & 0 & 0 &    0 & 0 & 0 & 0 & 0 & 1 &    0 & 0 & 1 & 1 &    1 \\
0 & 0 & 0 & 1 &    0 & 0 & 1 & 0 & 1 & 0 &    0 & 1 & 0 & 0 &    0 \\
0 & 0 & 1 & 0 &    0 & 1 & 0 & 1 & 0 & 0 &    1 & 0 & 0 & 0 &    0 \\
\end{pmatrix}
$$

The above matrix has $18$ rows and rank $10$. It can be checked through an explicit calculation that $H_Z'^{\text{ standard}}$ and $H_Z'^{\text{ fixed}}$ generate the same column space. \\

Puncturing consists in deleting the first column and keeping every row:
$$ L_X' =
\begin{pmatrix}
1 & 0 & 0 & 0 &    1 & 1 & 1 & 0 & 0 & 0 &    1 & 1 & 1 & 0 &    1 \\
0 & 1 & 0 & 0 &    1 & 0 & 0 & 1 & 1 & 0 &    1 & 1 & 0 & 1 &    1 \\
0 & 0 & 1 & 0 &    0 & 1 & 0 & 1 & 0 & 1 &    1 & 0 & 1 & 1 &    1 \\
0 & 0 & 0 & 1 &    0 & 0 & 1 & 0 & 1 & 1 &    0 & 1 & 1 & 1 &    1 \\

0 & 1 & 1 & 1 &    0 & 0 & 0 & 1 & 1 & 1 &    0 & 0 & 0 & 1 &    0 \\
1 & 0 & 1 & 1 &    0 & 1 & 1 & 0 & 0 & 1 &    0 & 0 & 1 & 0 &    0 \\
1 & 1 & 0 & 1 &    1 & 0 & 1 & 0 & 1 & 0 &    0 & 1 & 0 & 0 &    0 \\
1 & 1 & 1 & 0 &    1 & 1 & 0 & 1 & 0 & 0 &    1 & 0 & 0 & 0 &    0 \\
\end{pmatrix}
$$

The above matrix has $8$ rows and rank $5$.  It can be checked through an explicit calculation that $L_X'^{\text{ standard}}$ and $L_X'^{\text{ fixed}}$ generate the same column space. \\

$$L_Z'^{\text{ fixed}} =
\begin{pmatrix}
0 & 0 & 0 & 0 &    1 & 0 & 0 & 0 & 0 & 0 &    1 & 1 & 0 & 0 &    1 \\
0 & 1 & 0 & 0 &    0 & 0 & 0 & 1 & 1 & 0 &    0 & 0 & 0 & 1 &    0 \\
1 & 0 & 0 & 0 &    0 & 1 & 1 & 0 & 0 & 0 &    0 & 0 & 1 & 0 &    0 \\
0 & 0 & 1 & 1 &    0 & 0 & 0 & 0 & 0 & 1 &    0 & 0 & 0 & 0 &    0 \\

0 & 0 & 0 & 0 &    0 & 1 & 0 & 0 & 0 & 0 &    1 & 0 & 1 & 0 &    1 \\
0 & 0 & 1 & 0 &    0 & 0 & 0 & 1 & 0 & 1 &    0 & 0 & 0 & 1 &    0 \\
1 & 0 & 0 & 0 &    1 & 0 & 1 & 0 & 0 & 0 &    0 & 1 & 0 & 0 &    0 \\
0 & 1 & 0 & 1 &    0 & 0 & 0 & 0 & 1 & 0 &    0 & 0 & 0 & 0 &    0 \\

0 & 0 & 0 & 0 &    0 & 0 & 1 & 0 & 0 & 0 &    0 & 1 & 1 & 0 &    1 \\
0 & 0 & 0 & 1 &    0 & 0 & 0 & 0 & 1 & 1 &    0 & 0 & 0 & 1 &    0 \\
1 & 0 & 0 & 0 &    1 & 1 & 0 & 0 & 0 & 0 &    1 & 0 & 0 & 0 &    0 \\
0 & 1 & 1 & 0 &    0 & 0 & 0 & 1 & 0 & 0 &    0 & 0 & 0 & 0 &    0 \\

0 & 0 & 0 & 0 &    0 & 0 & 0 & 1 & 0 & 0 &    1 & 0 & 0 & 1 &    1 \\
0 & 0 & 1 & 0 &    0 & 1 & 0 & 0 & 0 & 1 &    0 & 1 & 0 & 0 &    0 \\
0 & 1 & 0 & 0 &    1 & 0 & 0 & 0 & 1 & 0 &    0 & 0 & 1 & 0 &    0 \\
1 & 0 & 0 & 1 &    0 & 0 & 1 & 0 & 0 & 0 &    0 & 0 & 0 & 0 &    0 \\

0 & 0 & 0 & 0 &    0 & 0 & 0 & 0 & 1 & 0 &    0 & 1 & 0 & 1 &    1 \\
0 & 0 & 0 & 1 &    0 & 0 & 1 & 0 & 0 & 1 &    0 & 0 & 1 & 0 &    0 \\
0 & 1 & 0 & 0 &    1 & 0 & 0 & 1 & 0 & 0 &    1 & 0 & 0 & 0 &    0 \\
1 & 0 & 1 & 0 &    0 & 1 & 0 & 0 & 0 & 0 &    0 & 0 & 0 & 0 &    0 \\

0 & 0 & 0 & 0 &    0 & 0 & 0 & 0 & 0 & 1 &    0 & 0 & 1 & 1 &    1 \\
0 & 0 & 0 & 1 &    0 & 0 & 1 & 0 & 1 & 0 &    0 & 1 & 0 & 0 &    0 \\
0 & 0 & 1 & 0 &    0 & 1 & 0 & 1 & 0 & 0 &    1 & 0 & 0 & 0 &    0 \\
1 & 1 & 0 & 0 &    1 & 0 & 0 & 0 & 0 & 0 &    0 & 0 & 0 & 0 &    0 \\
\end{pmatrix}
$$

The above matrix has $24$ rows and rank $11$. It can be checked through an explicit calculation that $L_Z'^{\text{ standard}}$ and $L_Z'^{\text{ fixed}}$ generate the same column space.

\subsection{Summary with dimensions}

We summarize the results of this appendix with two tables:

\begin{tabular}{ |p{3cm}||p{3cm}|p{3cm}|p{3cm}|  }
 \hline
 \multicolumn{3}{|c|}{Before puncturing} \\
 \hline
                                  & number of rows & rank \\
 \hline
 $H_X^{\text{ standard}}$  & 5                       & 5      \\
 $L_X^{\text{ standard}}$   & 5                       & 5      \\
 $H_Z^{\text{ standard}}$  & 11                     & 11     \\
 $L_Z^{\text{ standard}}$   & 11                     & 11     \\
 $H_X^{\text{ fixed}}$        & 8                       & 5      \\
 $L_X^{\text{ fixed}}$         & 8                       & 5      \\
 $H_Z^{\text{ fixed}}$        & 24                     & 11     \\
 $L_Z^{\text{ fixed}}$         & 24                     & 11     \\
 \hline
\end{tabular}

\begin{tabular}{ |p{3cm}||p{3cm}|p{3cm}|p{3cm}|  }
 \hline
 \multicolumn{3}{|c|}{After puncturing} \\
 \hline
                                  & number of rows & rank \\
 \hline
 $H_X'^{\text{ standard}}$  & 4                       & 4      \\
 $L_X'^{\text{ standard}}$   & 5                       & 5      \\
 $H_Z'^{\text{ standard}}$  & 10                     & 10     \\
 $L_Z'^{\text{ standard}}$  & 11                     & 11     \\
 $H_X'^{\text{ fixed}}$       & 4                       & 4      \\
 $L_X'^{\text{ fixed}}$        & 8                       & 5      \\
 $H_Z'^{\text{ fixed}}$       & 18                     & 10     \\
 $L_Z'^{\text{ fixed}}$        & 24                     & 11     \\
 \hline
\end{tabular}

%\section{results from \cite{barg2025geometric}}
%
%\underline{Theorem 7.10 from \cite{barg2025geometric}:} For every $K \in \mathcal{Q}_K$, the signed operator $\widetilde{Z}(k)_{\langle K \rangle}$ implements the logical multi-controlled-Z circuit corresponding to the collection of minimal covers of $K$:
%$$\widetilde{Z}(k)_{\langle K \rangle} \equiv \overline{C^{\mathcal{F}(K)}Z}.$$

% \printbibliography
\bibliographystyle{plain}
\bibliography{local_distillation}

\end{document}